\documentclass[12pt,a4paper]{article}
\usepackage{float, color,epsfig,citesort,amsmath,fullpage,amssymb,enumerate}
\usepackage{algorithm}
\usepackage{algpseudocode,array}
\usepackage{epic}
\def\myendproof{{\hfill \vbox{\hrule\hbox{%
\vrule height1.3ex\hskip0.8ex\vrule}\hrule }}\par}
\newtheorem{theorem}{Theorem}
\newtheorem{lemma}[theorem]{Lemma}

\newtheorem{fact}{Fact}

\newenvironment{proof}{{\it Proof. }}{\myendproof}

\newcommand{\setof}[1]{\{{#1}\}}
\newcommand{\Xomit}[1]{}
\newcommand{\OB}[1]{B_{#1}}
\newcommand{\BB}[1]{\bar{B}_{#1}}

\def\P {P}
\def\dis {\textit{d}}
\def\Buck {\textit{bucket}}
\def\Acc {\textit{acc}}
\def\Pref   {\textit{pred}}
\def\Next   {\textit{succ}}

\def\Head   {\eta}
\def\Minv   {\textit{min\_v}}
\def\DT	{\overline{\nabla}}

\def \skippt{21.5 pt}
\baselineskip \skippt

\def\setuptime {\rho}

\DeclareMathAlphabet{\mathpzc}{OT1}{pzc}{m}{it}
\floatstyle{ruled}
\newfloat{contialgo}{tbphH}

\title{{\bf Broadcasting in Heterogeneous Tree Networks with Edge Weight Uncertainty}}
\author{Cheng-Hsiao Tsou\thanks{Department of Computer Science and
                                 Information Engineering,
                                 National Taiwan University, \newline \mbox{} \hspace{11pt}
                                 Taipei 10617, Taiwan.
                                 Email: \setof{f97922063, ghchen}@csie.ntu.edu.tw.}
                \and
        Ching-Chi Lin\thanks{Department of Computer Science and Engineering,
                             National Taiwan Ocean University, \newline \mbox{} \hspace{11pt}
                             Keelung 20224, Taiwan.
                             Corresponding author.
                             Email: lincc@mail.ntou.edu.tw.}
                             \and
        Gen-Huey Chen$^*$
        }

\date{ }

\begin{document}
\maketitle
\thispagestyle{empty}
\addtocounter{page}{-1}

\begin{abstract}
A broadcasting problem in heterogeneous tree networks with edge weight uncertainty under the postal model is considered in this paper. The broadcasting problem asks for a minmax-regret broadcast center, which minimizes the worst-case loss in the objective function. Due to the presence of edge weight uncertainty, it is not easy to attack the broadcasting problem. An $O(n \log n \log \log n)$-time algorithm is proposed for solving the broadcasting problem.

\bigskip


\end{abstract}


\newpage
\def\skippt{23pt}
\baselineskip \skippt 
\section{Introduction}\label{section:intro}
We consider the broadcasting problem in heterogeneous tree networks $T = (V(T),E(T))$ with edge weight uncertainty, where the edge weight $w_{u,v}$ can take any value from [$w^-_{u,v}$,\;$w^+_{u,v}$] with unknown distribution. A heterogeneous network interconnects computers and other devices that can employ diverse operating systems and communication protocols. Representing such a network as a graph denoted as $G$ proves advantageous, wherein the set of network nodes corresponds to $V(G)$, representing the vertices in $G$, while the set of communication links corresponds to $E(G)$, signifying the edges within the graph. For each link $(u,v)\in E(G)$, there exists a positive weight $w_{u,v}$ that signifies the time required to transmit a message across the link. 

In response to the exact communication service requirements of a network, a variety of well-suited communication models have been devised~\cite{Harutyunyan2011,Harutyunyan2009,Slater1981,Su2016}. Among these communication models, the postal model, as explored in~\cite{Barnoy00,Barnoy94,Barnoy97}, is well-suited for characterizing packet-switching techniques. In the postal model, the process involves a sender $u$ connecting with a receiver $v$ and subsequently transmitting a message to $v$. A sender $u$ is not permitted to connect with two or more receivers simultaneously. However, $u$ can connect with a receiver, while transmitting a message to another receiver. Throughout this paper, we use $\setuptime$ to denote the time consumed for link connection, where $\setuptime$ is a positive constant. Let's consider a scenario where a sender $u$ intends to broadcast a message to its $k$ neighbors $v_1, v_2, \ldots, v_k$. Node $v_1$ will establish a connection with $u$ at time unit $\setuptime$ and receive the message at time unit $\setuptime + w_{u,v_1}$, $v_2$ will connect with $u$ at time unit $2 \cdot \setuptime$ and receive the message at time unit $2 \cdot \setuptime + w_{u,v_2}$, and finally, $v_k$ will connect with $u$ at time unit $k \cdot \setuptime$ and receive the message at time unit $k \cdot \setuptime + w_{u,v_k}$.

In recent years, broadcasting problems~\cite{Barnoy00,Barnoy94,Barnoy97,Slater1981,Albouy2023,Kowalski2023,Maja2017,Bortolussi2020} have gained considerable attention due to the rapid advancements in multimedia and network technologies. When utilizing the postal model, these problems can be viewed as extensions of center problems (with $\setuptime = 0$). The objective of broadcasting problems is to identify one or more nodes, referred to as broadcast centers, from which a message can be efficiently disseminated to other nodes with minimal time delay. Specifically, Su et al.~\cite{Su2016} introduced a linear-time algorithm for finding one broadcast center on a tree network.

Previous research on broadcasting problems typically assumed that the value of $w_{u,v}$ is deterministic. However, in practice, considering the uncertainty of $w_{u,v}$ is essential, as the available bandwidth of a communication link can vary over time. Thus, it is valuable to account for $w_{u,v}$ as uncertain, falling within a range [$w^-_{u,v}$,\;$w^+_{u,v}$], where $w^-_{u,v}$ and $w^+_{u,v}$ are non-negative real numbers. The concept of uncertainty has been previously introduced in various optimization problems, such as facility location~\cite{Dam2023,Mi1979,La1991,Av1997,Av2000a,Av2000b,Av2005,Bhattacharya2014,Wang2020,Ye2015}, resource allocation~\cite{Av2001,Av2004,Conde2005,Arumugam2019}, and job scheduling~\cite{Liao2020,hsu2020,Drwal2020}. There was an approach, named {\em minmax regret}~\cite{Av1997,Av2000a,Av2000b,Av2005,Yu2008,Ko1997,Ye2015,Wang2022}, proposed for solving them. Given an optimization problem with link weight uncertain, the minmax regret approach tries to minimize the worst-case loss to the objective function, without specifying the probability distribution of link weights. For a comprehensive discussion of the minmax regret approach, please refer to~\cite{Ko1997,Yury04}.

In this paper, under the postal model, we investigate the broadcasting problem in which a single broadcast center is designated, and all edge weights are uncertain. When the underlying network has a random topology, this problem is known as NP-hard, even if all edge weights are deterministic~\cite{Slater1981}. The broadcasting problem considered was inspired by the 1-center problem. Note that the broadcasting problem becomes the 1-center problem when $\setuptime = 0$. The 1-center problem is polynomial-time solvable on a random network with deterministic edge weights~\cite {Dvir2004}. Clearly, the problem to be investigated is more intractable than the 1-center problem. With $\setuptime > 0$, each sender has to determine an optimal broadcast sequence for its neighbors so as to minimize the broadcasting time to the whole network. Further, with uncertain edge weights, such an optimal broadcast sequence is more difficult to find, because it relies on the values of edge weights.

In their work~\cite{Bu2002}, Burkard and Dollani addressed uncertain edge weights in the 1-center problem (with $\setuptime = 0$) and successfully solved it on a tree network in $O(n \log n)$ time. In this paper, also on a tree network with uncertain edge weights, we extend Burkard and Dollani's work by solving the broadcasting problem (with $\setuptime > 0$) with one broadcast center in $O(n \log n \log\log n)$ time. In comparison, a multiplier of $O(\log\log n)$ is incurred, as a consequence of the inclusion of $\setuptime > 0$. In Table~\ref{table:result}, we show the best results thus far for the two problems on both a random network with deterministic edge weights and a tree network with uncertain edge weights.

\begin{table}[htb]
    \caption{Comparison between the 1-center problem and the broadcasting problem with one broadcast center.} \vspace{7pt}
    \label{table:result}
    \smallskip
    \begin{minipage}{\textwidth}
    \renewcommand\arraystretch{1.3}
    \centering
    \hspace{-8pt} 
    \footnotesize
    \begin{tabular}{m{4.1cm}| >{\centering} m{3cm} m{3.5cm} m{3.1cm}}
        \hline
        & value of $\setuptime$  & on a random network with deterministic\newline edge wights & on a tree network \newline with uncertain edge wights \\ \hline \hline
1-center problem & $0$   & $O(n^2 \log n + n m )$~\cite{Dvir2004}   & $O(n \log n)$~\cite{Bu2002} \\
\hline
broadcasting  problem with one broadcast center     & a positive constant    &  NP-hard~\cite{Slater1981}  & $O(n \log n \log\log n)$ (this paper) \\
        \hline
    \end{tabular}
    \end{minipage}
\end{table}
\vspace{-3pt}

The execution of the proposed algorithm is iterative and based on the prune-and-search strategy. To begin, the algorithm was executed on the tree $T$. Then, a centroid $x$ of $T$, a worst-case scenario, and a broadcast center $\kappa$ under the worst-case scenario were found, which takes $O(n \log\log n)$ time. As a consequence of Lemma 7, a minmax-regret broadcast center $\kappa^*$ of $T$ is either $x$ or located in the subtree of $T-x$ that contains $\kappa$. If $\kappa^* \not = x$, then the algorithm execution continues on the subtree containing $\kappa$. Since at most $O(\log n)$ iterations are executed, the overall time complexity is $O(n \log n \log\log
n)$.

The most challenging aspect of the proposed algorithm is conducting a worst-case scenario search in $O(n\log\log n)$ time, given the potentially infinite number of scenarios. To achieve this, we successfully reduced the search space to a maximum of $n-1$ scenarios through a series of scenario transformations. Through each transformation, we pinned down a particular edge weight, leading to the reduction of the search space.  Finally, a worst-case scenario could be found in $O(n \log\log n)$ time, by the aid of a preprocessing and some dedicated data structures. 

The rest of this paper is organized as follows. In Section~\ref{section:preliminaries}, some notations and definitions, together with some properties, are introduced. In Section~\ref{section:algorithm}, an $O(n \log n \log\log n)$-time algorithm is proposed for solving the broadcasting problem, with one broadcast center and uncertain edge weights, on a tree network. The correctness of the algorithm is verified in Section~\ref{section:correctness}, and the time complexity is analyzed in Section~\ref{section:time_complexity}. Finally, this paper concludes with some remarks in Section~\ref{section:conclusion}.

%
%
%
%

\section{Preliminaries}\label{section:preliminaries}
Let $T$ be a tree and $w_{u,v}$ be the weight of $(u,v)\in
E(T)$. We suppose that $n=|V(T)|$ and the value of $w_{u,v}$ is randomly generated from [$w^-_{u,v}$,\;$w^+_{u,v}$], where $w^-_{u,v}$ and $w^+_{u,v}$ are two non-negative real numbers.
Define $C$ to be a Cartesian product of all $n-1$ intervals
$[w^-_{u,v},\;w^+_{u,v}]$ for $T$.
For example, refer to the tree of Figure~\ref{figure:simple_tree}, where [$w^-_{u,v_1}$,\;$w^+_{u,v_1}$] = $[2,6]$
and [$w^-_{u,v_2}$,\;$w^+_{u,v_2}$] = $[1,4]$. We have
$C$ = \{$(g, h)\mid 2\le g\le 6$ and $1\le h\le 4$\} or $C$ = \{$(g, h)\mid 1\le g\le 4$ and $2\le h\le 6$\}, where $g$ and $h$ are two real numbers. Throughout this paper, we assume that $C$ is predetermined.
Also, notice that for the notations introduced in this section, if they are defined with respect to $T$, we remove $T$ from them for simplicity (so, we use $C$, instead of $C_T$ or $C(T)$).

\vspace{12pt}
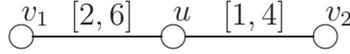
\begin{figure}[h]
\begin{center}
\unitlength=1mm
\begin{picture}(45, 5)
\multiput(2.5,0)(20,0){3}{\circle{3}}
\multiput(4,0)(20,0){2}{\line(1,0){17}}
\put(2.5,2){$v_1$}
\put(22.5,2){$u$}
\put(42.5,2){$v_2$}
\put(9,1.5){$[2,6]$}
\put(29,1.5){$[1,4]$}
\end{picture}
\caption{A tree $T$.}
\label{figure:simple_tree}
\end{center}
\end{figure}
\vspace{-12pt}
Clearly, $C$ is an infinite set. A tuple $s$ = $(s_1,s_2,\ldots, s_{n-1})$ is called a {\em scenario} of $T$, if $s\in C$.
For the example of Figure~\ref{figure:simple_tree}, if $C$ = \{$(g, h)\mid 2\le g\le 6$ and $1\le h\le 4$\}, then $(5,2)$ is a scenario of $T$, but $(2,5)$ is not. In fact, a scenario of $T$
represents a feasible weight assignment of all edges in $E(T)$, in the sequence specified by $C$. Suppose $s \in C$. Define $w^s_{u, v}$ to be the weight of $(u,v)\in E(T)$, ${\textit b\_time}^s(u, G)$ to be the minimum time required for $u$ to broadcast a
message to all other vertices of $G$, and ${\textit B\_Ctr}^s$ to be the set of broadcast centers of $T$, all under the scenario $s$. That is, ${\textit B\_Ctr}^s$ = \{$u\mid {\textit b\_time}^s(u, T) \le {\textit b\_time}^s(v, T)$ for all $v \in V(T)$\}.

Further, for any $x, y \in V(T)$, we define $r^s_{x,y}$ = ${\textit b\_time}^s(x, T) - {\textit b\_time}^s(y, T)$ and ${\textit max\_r}(x) = \max$\{$r^s_{x,y}\mid y \in V(T)$ and $s \in C$\}, where $r^s_{x,y}$ is called {\em the relative regret} of $x$ with respect to $y$ and ${\textit max\_r}(x)$ is called the {\em maximum regret} of $x$. Intuitively, $r^s_{x,y}$ represents the time lost, if $x$ was chosen, instead of $y$, to broadcast a message over $T$ under the scenario $s$. Notice that
if ${\textit max\_r}(x) = r^{s'}_{x,y'}$ for
some $y' \in V(T)$ and $s' \in C$, then $y' \in {\textit B\_Ctr}^{s'}$,
and $s'$ is called a {\em worst-case
scenario} of $T$ with respect to $x$. It is possible that there are multiple worst-case scenarios of $T$ with respect to $x$. The {\em minmax-regret
broadcasting problem on $T$ with edge weight uncertainty} is to find  a {\em minmax-regret broadcast center} $\kappa^*
\in V(T)$ such that its maximum regret is minimized, i.e., ${\textit max\_r}(\kappa^*) \le {\textit max\_r}(v)$ for all $v \in V(T)$.

Suppose $x, y \in V(T)$. Removing $x$ from $T$ will result in some subtrees of $T$. We use $\OB{x, y}$ to denote the subtree that contains $y$, and let $\BB{x, y}$ = $T - \OB{x, y}$ (refer to Figure~\ref{figure:open-branch}).
We have $\OB{x,y} = \BB{y,x}$, if $(x,y) \in E(T)$.

\vspace{12pt}
\begin{figure}[h]
\begin{center}
\unitlength=1mm
\begin{picture}(80, 42)
\put(1, 25){\circle{2}}
\put(18, 25){\circle{2}}
\put(32, 34){\circle{2}}
\put(32, 16){\circle{2}}
\put(18, 7){\circle{2}}
\put(4, 16){\circle{2}}
\put(49, 16){\circle{2}}
\put(63, 7){\circle{2}}
\put(63, 25){\circle{2}}
\put(77, 16){\circle{2}}
\put(77, 34){\circle{2}}

\put(2,25){\line(1,0){15}}
\put(33,16){\line(1,0){15}}
\put(19,25){\line(4,3){12}}
\put(19,25){\line(4,-3){12}}
\put(19,7){\line(4,3){12}}
\put(17,7){\line(-4,3){12}}
\put(50,16){\line(4,3){12}}
\put(50,16){\line(4,-3){12}}
\put(64,25){\line(4,3){12}}
\put(64,25){\line(4,-3){12}}

\put(32, 18){$x$}
\put(63, 27){$y$}

\qbezier[50](-4,21)(-2,42)(17,42)
\qbezier[50](17,42)(36,42)(38,21)
\qbezier[50](-4,21)(-2,0)(17,0)
\qbezier[50](17,0)(36,0)(38,21)

\qbezier[50](43,21)(45,42)(64,42)
\qbezier[50](64,42)(83,42)(85,21)
\qbezier[50](43,21)(45,0)(64,0)
\qbezier[50](64,0)(83,0)(85,21)

\put(-12,16){$\BB{x,y}$}
\put(86,16){$\OB{x,y}$}

\end{picture}
\caption{$\OB{x,y}$ and $\BB{x,y}$.}
\label{figure:open-branch}
\end{center}
\end{figure}
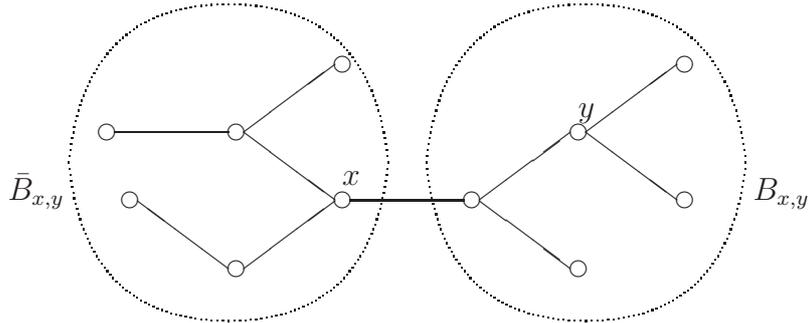
\vspace{-12pt}

\begin{lemma} [\cite{Su2016}] \label{lemma:Break-edge}
Suppose $s \in C$ and $(u,v)\in E(T)$. If ${\textit b\_time}^s(u, \BB{u,v}) \le
{\textit b\_time}^s(v, \BB{v, u})$, then the following hold:
\vspace{-6pt}
\begin{itemize} 
\item ${\textit b\_time}^s(u, T) =  \setuptime + w^s_{u, v} + {\textit b\_time}^s(v, \BB{v, u});$
\vspace{-12pt}
\item ${\textit b\_time}^s(v, T) \le {\textit b\_time}^s(u, T)$.
\end{itemize}
\end{lemma}

A tree is a {\em star}, if it has one or two vertices, or has exactly one vertex whose degree is greater than one. The latter is called the {\em center} of the star. For the one-vertex star, the vertex is the center. For the two-vertex star, either vertex can be the center.

\begin{lemma} [\cite{Su2016}] \label{lemma:centers-as-star}
The subgraph of $T$ induced by ${\textit B\_Ctr}^s$ is a star.
\end{lemma}

Define $N_T(u) = \{v\mid (u,v) \in E(T)\}$ to be the set of neighboring vertices of $u$ in $T$.
A vertex $\hat{\kappa} \in {\textit B\_Ctr}^s$ is a {\em prime broadcast center} of $T$ under the scenario $s$ if and only if ${\textit b\_time}^s(\hat{\kappa}, \BB{\hat{\kappa}, u}) \ge {\textit b\_time}^s(u,
\BB{u, \hat{\kappa}})$ for all $u\in N_T(\hat{\kappa})$. Further, according to Lemma~\ref{lemma:Break-edge}, if $\hat{\kappa}$ is a prime broadcast center of $T$ under the scenario $s$, then ${\textit b\_time}^s(\hat{\kappa}, T) \le {\textit b\_time}^s(u, T)$ for all $u\in N_T(\hat{\kappa})$.

\begin{lemma}
\label{lemma:prime-must-exist} A prime broadcast
center of $T$ always exists under any scenario of $T$.
\end{lemma}
\begin{proof}
Let $s \in C$ be an arbitrary scenario of $T$. When $|V(T)| = 1$, it is easy to see that the only vertex
is a prime broadcast center of $T$. So, we assume $|V(T)| \ge 2$. According to Lemma~\ref{lemma:centers-as-star}, the subgraph of $T$
induced by ${\textit B\_Ctr}^s$ is a star. Let $\kappa$
be the center of the star. If $\kappa$ is not a prime broadcast center of $T$,
then there exists $x\in N_T(\kappa)$ such that ${\textit b\_time}^s(\kappa, \BB{\kappa, x}) < {\textit b\_time}^s(x, \BB{x, \kappa})$.
It follows, as a consequence of Lemma~\ref{lemma:Break-edge}, that ${\textit b\_time}^s(x, T) \le {\textit b\_time}^s(\kappa,
T)$, which implies $x\in {\textit B\_Ctr}^s$.
Clearly, $x$ is a prime broadcast center of $T$, if $N_T(x)-\setof{\kappa} = \emptyset$.
Otherwise, for each $v \in N_T(x)-\setof{\kappa}$, we have
$v \not\in {\textit B\_Ctr}^s$, i.e., ${\textit b\_time}^s(x, T) < {\textit b\_time}^s(v,T)$
(as a consequence of $\kappa$ being the center of the star).
Again, as a consequence of Lemma~\ref{lemma:Break-edge}, we have ${\textit b\_time}^s(x, \BB{x, v}) > {\textit b\_time}^s(v, \BB{v, x})$. Therefore, $x$ is a prime broadcast center of $T$.
\end{proof}

\bigskip

Suppose $x, y \in V(T)$. 
Let $\P_{x, y}$ be the set of edges contained in the $x$-to-$y$ path of $T$, $\dis_{x, y}$ = $|\P_{x, y}|$, and $\tilde{w}^s_{x, y}$ be the total weight of all edges in $\P_{x, y}$ under the
scenario $s$, i.e. $\tilde{w}^s_{x,y} = \sum_{(u,v) \in \P_{x,y}}w^s_{u,v}$. The following lemma shows a relation between ${\textit b\_time}^s(x,
T)$ and ${\textit b\_time}^s(\hat{\kappa}, \BB{\hat{\kappa}, x})$, where $\hat{\kappa}$ is a prime broadcast center of $T$ under the scenario $s$.

\begin{lemma}
\label{lemma:Direct-to-center} Suppose that $s \in C$ and $\hat{\kappa}$ is a prime
broadcast center of $T$ under $s$. If $x \in V(T) -\setof{\hat{\kappa}}$, then ${\textit b\_time}^s(x,
T) =  \dis_{x, \hat{\kappa}} \cdot \setuptime   + \tilde{w}^s_{x, \hat{\kappa}} + {\textit b\_time}^s(\hat{\kappa}, \BB{\hat{\kappa}, x})$.
\end{lemma}
\begin{proof}
Suppose that $(u_0, u_1,\ldots,u_h)$ is the $x$-to-$\hat{\kappa}$ path of $T$,
where $u_0 = x$ and $u_h = \hat{\kappa}$. 
Refer to Figure~\ref{figure:proof-direct-to-center}.
For $0 \le i \le h-1$,
$\BB{u_i, u_{i+1}}$ is a subgraph of $\BB{u_{h-1}, \hat{\kappa}}$, and
$\BB{\hat{\kappa}, u_{h-1}}$ is a subgraph of $\BB{u_{i+1}, u_i}$.
Now that $ {\textit b\_time}^s(\hat{\kappa},
\BB{\hat{\kappa}, u_{h-1}}) \ge {\textit b\_time}^s(u_{h-1}, \BB{u_{h-1}, \hat{\kappa}}) $, we have ${\textit b\_time}^s(u_i, \BB{u_i, u_{i+1}}) \le
{\textit b\_time}^s(u_{i+1}, \BB{u_{i+1}, u_i})$ for $0 \le i \le h-1$. Then, according to Lemma~\ref{lemma:Break-edge}, we have
\begin{equation}
{\textit b\_time}^s(u_i, T) =  \setuptime + w^s_{u_i, u_{i+1}} + {\textit b\_time}^s(u_{i+1}, \BB{u_{i+1}, u_i})\mathrm{~for~}0 \le i \le h-1.
\label{equation:broadcasting-time}
\end{equation}

Observe $ \setof{(u_i, u_{i+1})} \cup E(\BB{u_{i+1}, u_i}) \subseteq E(\BB{u_i, u_{i-1}}) \subset E(T)$ for $1 \le i \le h-1$. It implies
$\setuptime + w^s_{u_i, u_{i+1}} + {\textit b\_time}^s(u_{i+1}, \BB{u_{i+1}, u_i}) \le {\textit b\_time}^s(u_i, \BB{u_i, u_{i-1}}) \le {\textit b\_time}^s(u_i, T)$ for $1 \le i \le h-1$, where the two equalities hold, as a consequence of (\ref{equation:broadcasting-time}). That is, we have
${\textit b\_time}^s(u_i, \BB{u_i, u_{i-1}}) = \setuptime + w^s_{u_i, u_{i+1}} + {\textit b\_time}^s(u_{i+1}, \BB{u_{i+1}, u_i})$ for
$1 \le i \le h-1$.
It follows that ${\textit b\_time}^s(x, T) =  \dis_{x, \hat{\kappa}} \cdot \setuptime + \tilde{w}^s_{x, \hat{\kappa}} + {\textit b\_time}^s(\hat{\kappa},
\BB{\hat{\kappa}, u_{h-1}})$,
where $\BB{\hat{\kappa}, u_{h-1}} = \BB{\hat{\kappa}, x}$.
\end{proof}

\vspace{6pt}
\begin{figure}[h]
\begin{center}
\unitlength=1mm
\begin{picture}(80, 45)
\put(1, 20){\circle{2}}
\put(11, 20){\circle{2}}
\put(29, 20){\circle{2}}
\put(39, 20){\circle{2}}
\put(57, 20){\circle{2}}
\put(67, 20){\circle{2}}

\put(2,20){\line(1,0){8}}
\put(30,20){\line(1,0){8}}
\put(58,20){\line(1,0){8}}

\put(12,20){\line(1,0){3}}
\put(17.5,20){$\ldots$}
\put(28,20){\line(-1,0){3}}
\put(40,20){\line(1,0){3}}
\put(45.5,20){$\ldots$}
\put(56,20){\line(-1,0){3}}

\put(-9, 22){$x = u_0$}
\put(9, 22){$u_1$}
\put(27, 22){$u_i$}
\put(37, 22){$u_{i+1}$}
\put(55, 22){$u_{h-1}$}
\put(66, 22){$u_h = \hat{\kappa}$}

\qbezier(67,21)(109, 45)(109,20)
\qbezier(67,19)(109, -5)(109,20)

\qbezier(57,21)(-15, 55)(-15,20)
\qbezier(57,19)(-15, -15)(-15,20)

\qbezier[140](39,21)(111, 55)(111,20)
\qbezier[140](39,19)(111, -15)(111,20)

\qbezier[100](29,21)(-13, 45)(-13,20)
\qbezier[100](29,19)(-13, -5)(-13,20)

\put(0, 40){$\BB{u_{h-1},\hat{\kappa}}$}
\put(90, 27){$\BB{\hat{\kappa}, u_{h-1}}$}
\put(-3, 28){$\BB{u_i,u_{i+1}}$}
\put(87, 40){$\BB{u_{i+1}, u_i}$}

\end{picture}
\vspace{-12pt}
\caption{proof of Lemma~\ref{lemma:Direct-to-center}.}
\label{figure:proof-direct-to-center}
\end{center}
\end{figure}
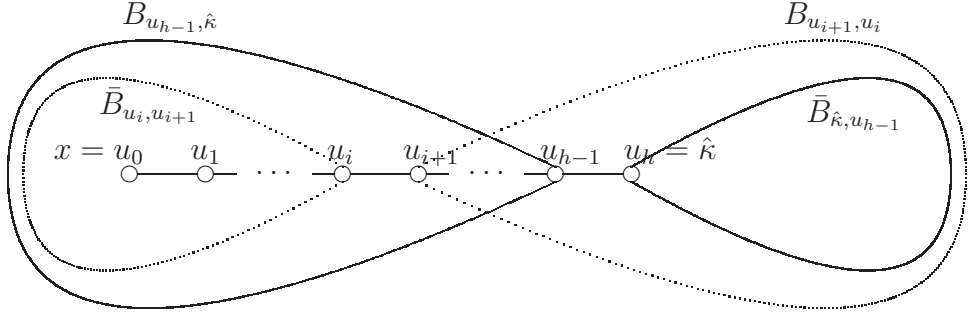
\vspace{-12pt}

\bigskip

Also notice that $\dis_{x, \hat{\kappa}} \cdot \setuptime  + \tilde{w}^s_{x, \hat{\kappa}}$ in Lemma~\ref{lemma:Direct-to-center} is the minimum time requirement for $x$ to broadcast a message along the $x$-to-$\hat{\kappa}$ path. Actually, Lemma~\ref{lemma:Direct-to-center} reveals that ${\textit b\_time}^s(x,T)$ can be calculated as
the minimum time requirement for $x$ to broadcast a message over a subgraph, i.e., the union of the $x$-to-$\hat{\kappa}$ path and $\BB{\hat{\kappa},x}$, of $T$. Moreover,  the proof of
Lemma~\ref{lemma:Direct-to-center} shows that an optimal transmission sequence for ${\textit b\_time}^s(x,T)$ has each $u_i$ connected to $u_{i+1}$ immediately after $u_i$ has received the message.
In general, for any $x,y \in V(T)$ and $x \not= y$,
if $y$ is not a prime broadcast center of $T$ under the scenario $s$, then we have
${\textit b\_time}^s(x,
T) \ge  \dis_{x,
y} \cdot \setuptime  + \tilde{w}^s_{x, y} + {\textit b\_time}^s(y, \BB{y, x})$.
Therefore, we have the following lemma for $x \not= y$.

\begin{lemma}
\label{lemma:general-x-y}
Suppose $s \in C$ and $x,y \in V(T)$. If $x \not= y$, then
%
${\textit b\_time}^s(x,
T) \ge \dis_{x,
y} \cdot \setuptime + \tilde{w}^s_{x, y} + {\textit b\_time}^s(y, \BB{y, x})$.
\end{lemma}


\begin{lemma}\label{lemma:prime-is-the-center-of-star}
Suppose $s \in C$. A prime broadcast center $\hat{\kappa}$ of $T$
under $s$ is a center of the
star induced by ${\textit B\_Ctr}^s$.
\end{lemma}
\begin{proof}
If $\hat{\kappa}$ is not a center of the star induced by ${\textit B\_Ctr}^s$
(refer to Lemma~\ref{lemma:centers-as-star}), then there
exist two broadcast centers $\kappa_1$ and $\kappa_2$ in
${\textit B\_Ctr}^s$ such that $(\hat{\kappa}, \kappa_1,\kappa_2)$
is a path in $T$. According to
Lemma~\ref{lemma:Direct-to-center}, we have ${\textit
b\_time}^{s}(\kappa_1, T) = \dis_{\kappa_1, \hat{\kappa}} \cdot \setuptime +
\tilde{w}^{s}_{\kappa_1, \hat{\kappa}} + {\textit
b\_time}^{s}(\hat{\kappa}, \BB{\hat{\kappa}, \kappa_1})$ and
 ${\textit b\_time}^{s}(\kappa_2, T) = \dis_{\kappa_2, \hat{\kappa}} \cdot \setuptime +
\tilde{w}^{s}_{\kappa_2, \hat{\kappa}} + {\textit
b\_time}^{s}(\hat{\kappa}, \BB{\hat{\kappa}, \kappa_2})$.
Since $\P_{\kappa_1, \hat{\kappa}} \subset \P_{\kappa_2,
\hat{\kappa}}$ and $\BB{\hat{\kappa}, \kappa_1} =
\BB{\hat{\kappa}, \kappa_2}$, we have
${\textit b\_time}^{s}(\kappa_1, T) < {\textit
b\_time}^{s}(\kappa_2, T)$, a contradiction.
\end{proof}

%

\section{Finding a minmax-regret broadcast center}
\label{section:algorithm}
In this section, assuming that $T$ is with edge weight uncertainty,
we intend to find a minmax-regret broadcast center on $T$ under the postal model. It was shown in \cite{Su2016} that given any $x \in V(T)$ and $s \in C$, both ${\textit B\_Ctr}^s$ and ${\textit b\_time}^s(x, T)$ can be determined in $O(n)$ time. Suppose that a worst-case scenario $s'$ of $T$ with respect to $u \in V(T)$ is available, i.e., ${\textit max\_r}(u) = r^{s'}_{u,v'}$ for some $v' \in V(T)$. Since $v' \in {\textit B\_Ctr}^{s'}$, we can compute ${\textit max\_r}(u)$ in $O(n)$ time, using the results of \cite{Su2016}. Further, a minmax-regret broadcast center of $T$ can be obtained, if ${\textit max\_r}(u)$ is computed for all $u \in V(T)$.

In subsequent discussion, we use $\ddot{s}(u)$ to denote some worst-case scenario of $T$ with respect to $u$. If $Q(n)$ is the time requirement for finding $\ddot{s}(u)$, then the approach above takes $O(n^2 + n \cdot Q(n))$ time to find a minmax-regret broadcast center of $T$, because it examines all $n$ vertices of $T$. In this section, a more efficient approach to finding a minmax-regret broadcast center of $T$ is suggested. Inspired by the following lemma, the new approach can result in an $O(n \log n + Q(n)\cdot \log n)$ time algorithm.


\begin{lemma} \label{lemma:Contain-optimal-solution}
Suppose $x \in V(T)$. If $x \in {\textit B\_Ctr}^{\ddot{s}(x)}$, then $x$ is a
minmax-regret broadcast center of $T$. Otherwise, a minmax-regret broadcast center of $T$ can be found in $V(\OB{x, \kappa}) \cup
\setof{x}$, where $\kappa \in {\textit B\_Ctr}^{\ddot{s}(x)}$.
\end{lemma}
\begin{proof}
Clearly, when $x \in {\textit B\_Ctr}^{\ddot{s}(x)}$,
we have ${\textit max\_r}(x) = 0$, implying that $x$ is a
minmax-regret broadcast center of $T$, because ${\textit max\_r}(v) \ge 0$ for every $v \in V(T)$. Then, we assume $x \not \in {\textit B\_Ctr}^{\ddot{s}(x)}$. It suffices to show that for every $x'  \in V(T) -
V(\OB{x, \kappa}) - \setof{x}$, we have ${\textit max\_r}(x') > {\textit max\_r}(x)$.
Suppose that $\hat{\kappa}$ is a prime broadcast center
of $T$ under $\ddot{s}(x)$.
According to
Lemma~\ref{lemma:Direct-to-center}, we have ${\textit b\_time}^{\ddot{s}(x)}(x', T) = \dis_{x', \hat{\kappa}} \cdot \setuptime + \tilde{w}^{\ddot{s}(x)}_{x', \hat{\kappa}} +
{\textit b\_time}^{\ddot{s}(x)}(\hat{\kappa}, \BB{\hat{\kappa}, x'})$ and ${\textit b\_time}^{\ddot{s}(x)}(x, T) = \dis_{x, \hat{\kappa}} \cdot \setuptime + \tilde{w}^{\ddot{s}(x)}_{x, \hat{\kappa}} + {\textit b\_time}^{\ddot{s}(x)}(\hat{\kappa},
\BB{\hat{\kappa}, x})$.
As a consequence of Lemma~\ref{lemma:centers-as-star},
we have $\hat{\kappa} \in {\textit B\_Ctr}^{\ddot{s}(x)} \subseteq V(\OB{x,
\kappa})$.
Since $\P_{x,\hat{\kappa}}
\subset \P_{x',\hat{\kappa}}$, we have $\dis_{x',
\hat{\kappa}} \cdot \setuptime + \tilde{w}^{\ddot{s}(x)}_{x', \hat{\kappa}} > \dis_{x, \hat{\kappa}} \cdot \setuptime + \tilde{w}^{\ddot{s}(x)}_{x, \hat{\kappa}}$. It follows
that we have ${\textit b\_time}^{\ddot{s}(x)}(x', T) > {\textit b\_time}^{\ddot{s}(x)}(x, T)$, because $\BB{\hat{\kappa}, x'} = \BB{\hat{\kappa}, x}$. So, we have
${\textit max\_r}(x') \ge r^{\ddot{s}(x)}_{x', \hat{\kappa}} = {\textit b\_time}^{\ddot{s}(x)}(x', T) - {\textit b\_time}^{\ddot{s}(x)}(\hat{\kappa}, T)
> {\textit b\_time}^{\ddot{s}(x)}(x, T) - {\textit b\_time}^{\ddot{s}(x)}(\hat{\kappa}, T)=r^{\ddot{s}(x)}_{x, \hat{\kappa}}={\textit max\_r}(x)$.
\end{proof}

\bigskip

With Lemma~\ref{lemma:Contain-optimal-solution}, a minmax-regret broadcast center of $T$ can be found by the prune-and-search strategy. The search starts on $T$. A vertex $z$ is first determined such that the largest open $z$-branch has minimum size, i.e., $\max_{y \in V(T) - \setof{z}}\{|\OB{z,y}|\} \le \max_{y \in V(T) - \setof{v}}\{|\OB{v, y}|\}$ for every $v \in V(T)$. The vertex $z$ thus found is called a {\em centroid} of $T$. If $z \in {\textit B\_Ctr}^{\ddot{s}(z)}$, then $z$ is a minmax-regret broadcast center of $T$ and the search terminates.

Otherwise, a broadcast center $\kappa \in {\textit B\_Ctr}^{\ddot{s}(z)}$ of $T$ under $\ddot{s}(z)$ is found, and the search continues on the subtree of $T$ induced by $V(\OB{z, \kappa}) \cup \setof{z}$. The resulting subtree is denoted by $T'$. A centroid $z'$ of $T'$
is then determined. Again, if $z' \in {\textit B\_Ctr}^{\ddot{s}(z')}$, then the search terminates with $z'$ being a minmax-regret broadcast center of $T$. If
$z' \not\in {\textit B\_Ctr}^{\ddot{s}(z')}$, then some $\kappa' \in {\textit B\_Ctr}^{\ddot{s}(z')}$ is found and the search continues on the subtree of $T'$ induced by $V(T') \cap \{V(\OB{z', \kappa'}) \cup \setof{z'}\}$.

The search will proceed until $z$ is a minmax-regret broadcast center of $T$ or the size of the induced subtree is small enough. For the latter case, a minmax-regret broadcast center of $T$ can be found with an additional $O(n + Q(n))$ time. The detailed execution is expressed as Algorithm~\ref{algorithm:mrbc} below.

\bigskip

\begin{algorithm}[htb]
\caption{Determining a minmax-regret broadcast center of a tree}
\label{algorithm:mrbc}
\begin{algorithmic} [1]
\Require a tree $T$.
\Ensure a minmax-regret broadcast center of $T$.
\State $T' \leftarrow T$;
\While {$|V(T')| \ge 3$}
\State determine a centroid $z$ of $T'$;
\State find ${\ddot{s}(z)}$ and ${\textit B\_Ctr}^{\ddot{s}(z)}$;
\If    {$z \in {\textit B\_Ctr}^{\ddot{s}(z)}$}
\State \Return $z$;
%
\Else
\State find $\kappa \in {\textit B\_Ctr}^{\ddot{s}(z)}$;
\State $T' \leftarrow$ the subtree of $T'$ induced by $V(T') \cap \{V(\OB{z, \kappa}) \cup \setof{z}\}$;
\EndIf
\EndWhile
\State determine $v \in V(T')$ whose ${\textit max\_r}(v)$ is minimum;
\State \Return $v$.
\end{algorithmic}
\end{algorithm}

It was shown in~\cite{Kang75,Goldman71}
that a centroid $z$ of $T$ can be found in $O(n)$ time.
Besides, the largest open $z$-branch contains at most
$\lfloor \frac{n}{2} \rfloor$ vertices. However, it is not easy to find $\ddot{s}(z)$, because there are an infinite number of scenarios. Suppose $V(T) = \{v_1, v_2, \ldots, v_n\}$.
For $x \in V(T)$,
we would like to construct a finite set $\ddot{S}(x) = \bigcup_{1\le \ell \le n}\Psi_{\ell}(x)$ of scenarios such that a worst-case scenario of $T$ with respect to $x$ can be found in some $\Psi_{\ell}(x)$, i.e., $\ddot{s}(x) \in \Psi_{\ell}(x)$. 
The construction of $\ddot{S}(x)$ is described below.

In the rest of this paper,
we assume $x = v_n$, without losing generality.
We set $\Psi_n(x) = \setof{c}$, where $c\in C$ is an arbitrary scenario of $T$.
Define $\alpha_{x, v_i}$ to be the {\em base scenario}
of $T$ with respect to $x$ and $v_i$, where $1 \le i \le n-1$, which satisfies the following:
\begin{itemize}
\item $w^{\alpha_{x, v_i}}_{a,b} = w^+_{a, b}$, if $(a, b) \in \P(x, v_i) \cup E(\BB{v_i, x})$;
\item $w^{\alpha_{x, v_i}}_{a,b} = w^-_{a, b}$, else.
\end{itemize}

Refer to Figure~\ref{figure:scenario}, where an illustrative example is shown with $\setuptime = 1$. In Figure~\ref{figure:scenario}(b), we use $\overline{7}$ and $\underline{5}$, for example, to denote $w^+_{a,b}$ and $w^-_{a,b}$, respectively. For any $s \in C$, we have $\dis_{x, v_i} \cdot \setuptime + \tilde{w}^{s}_{x, v_i} + {\textit b\_time}^{s}(v_i, \BB{v_i, x}) \le \dis_{x, v_i} \cdot \setuptime + \tilde{w}^{\alpha_{x, v_i}}_{x, v_i} + {\textit b\_time}^{\alpha_{x, v_i}}(v_i, \BB{v_i, x})$. To construct $\Psi_i(x)$, we need to create new scenarios, which are modifications of $\alpha_{x, v_i}$. The edge weights under $\alpha_{x, v_i}$ are to be changed according to a predefined sequence that is suggested by the following lemma.

\vspace{12pt}
\begin{figure}[h]
\begin{center}
\unitlength=1mm
\begin{picture}(150, 60)(-3, -10)
\scriptsize
\put(0, 29){\circle{2}}
\put(17, 29){\circle{2}}
\put(34, 29){\circle{2}}
\put(51, 29){\circle{2}}
\put(34, 46){\circle{2}}
\put(51, 46){\circle{2}}
\put(0, 12){\circle{2}}
\put(17, 12){\circle{2}}
\put(4, 39){\circle{2}}
\put(8, -2){\circle{2}}
\put(26, -2){\circle{2}}
\put(60, 15){\circle{2}}
\put(60, -2){\circle{2}}
\put(46, 6){\circle{2}}

\put(1,29){\line(1,0){15}}
\put(18,29){\line(1,0){15}}
\put(35,29){\line(1,0){15}}
\put(34,30){\line(0,1){15}}
\put(51,30){\line(0,1){15}}
\put(0,13){\line(0,1){15}}
\put(17,13){\line(0,1){15}}
\put(8,-1){\line(3,4){9}}
\put(26,-1){\line(-3,4){9}}
\put(17,30){\line(-4,3){12}}
\put(60,16){\line(-3,4){9}}
\put(60,-1){\line(0,1){15}}
\put(47,6){\line(4,3){12}}

\put(5, 25.5){$[0, 7]$}
\put(23, 30.5){$[1, 2]$}
\put(39, 30.5){$[1, 2]$}
\put(34.5, 38){$[0, 3]$}
\put(51.5, 38){$[5, 6]$}
\put(0.5, 18){$[2, 5]$}
\put(17.5, 18){$[2, 5]$}
\put(23, 4){$[1, 6]$}
\put(4.5, 4){$[3, 4]$}
\put(10.5, 36){$[5, 7]$}
\put(56, 22){$[2, 4]$}
\put(47.5, 12){$[2, 3]$}
\put(60.5, 6){$[1, 4]$}

\put(80, 29){\circle{2}}
\put(97, 29){\circle{2}}
\put(114, 29){\circle{2}}
\put(131, 29){\circle{2}}
\put(114, 46){\circle{2}}
\put(131, 46){\circle{2}}
\put(80, 12){\circle{2}}
\put(97, 12){\circle{2}}
\put(84, 39){\circle{2}}
\put(88, -2){\circle{2}}
\put(106, -2){\circle{2}}
\put(140, 15){\circle{2}}
\put(140, -2){\circle{2}}
\put(126, 6){\circle{2}}

\put(81,29){\line(1,0){15}}
\put(98,29){\line(1,0){15}}
\put(115,29){\line(1,0){15}}
\put(114,30){\line(0,1){15}}
\put(131,30){\line(0,1){15}}
\put(80,13){\line(0,1){15}}
\put(97,13){\line(0,1){15}}
\put(88,-1){\line(3,4){9}}
\put(106,-1){\line(-3,4){9}}
\put(97,30){\line(-4,3){12}}
\put(140,16){\line(-3,4){9}}
\put(140,-1){\line(0,1){15}}
\put(127,6){\line(4,3){12}}

\put(87.5, 25.5){$\overline{7}$}
\put(104.5, 30){$\overline{2}$}
\put(121.5, 30){$\overline{2}$}
\put(114.5, 37){$\underline{0}$}
\put(131.5, 37){$\underline{5}$}
\put(80.5, 18){$\overline{5}$}
\put(97.5, 18){$\overline{5}$}
\put(103, 4){$\overline{6}$}
\put(89, 4){$\overline{4}$}
\put(90, 36){$\overline{7}$}
\put(136, 22){$\underline{2}$}
\put(130.5, 11){$\underline{2}$}
\put(140.5, 6){$\underline{1}$}

\normalsize
\put(52, 29){$x$}
\put(18, 30){$v_i$}
\put(132, 29){$x$}
\put(98, 30){$v_i$}
\put(77, 31){$u_{i,1}$}
\put(98, 13){$u_{i,2}$}
\put(81, 41.5){$u_{i,3}$}

\put(28, -10){(a)}
\put(108, -10){(b)}

\end{picture}
\caption{An illustrative example with $\setuptime = 1$. (a) A tree $T$. (b) The base scenario $\alpha_{x, v_i}$.}
\label{figure:scenario}
\end{center}
\end{figure}
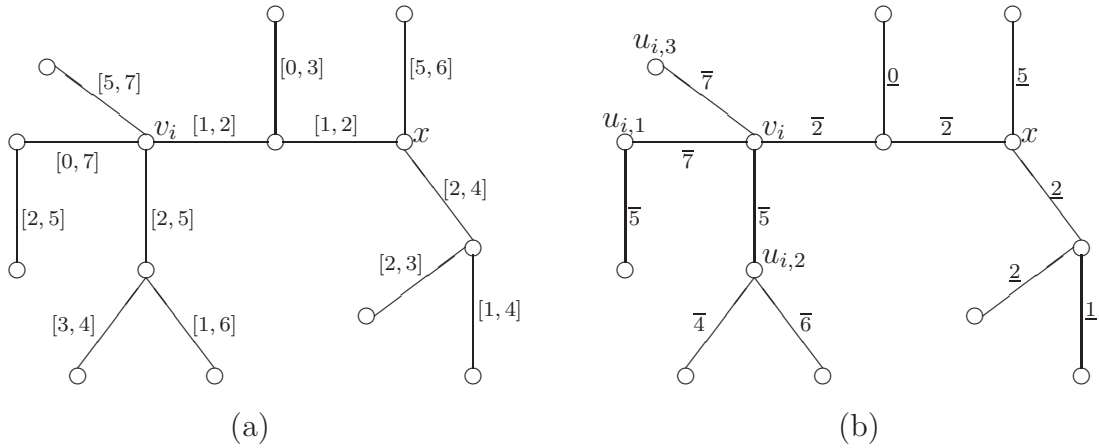
\vspace{-12pt}

\begin{lemma}
$\cite{Su2016}$
\label{lemma:optimal-sequence} Suppose $x \in V(T)$, $s \in C$, and
$N_T(x) = \{u_1, u_2, \ldots, u_h\}$, where $h = |N_T(x)|$. If
$w^s_{x, u_k} + {\textit b\_time}^{s}(u_k, \BB{u_k, x})$
is nonincreasing as $k$ increases from $1$ to $h$,
then $u_1, u_2, \ldots, u_h$ is an optimal sequence (with respect to minimum broadcast time)
for $x$ to broadcast a message to $N_T(x)$ under $s$.
That is, ${\textit b\_time}^{s}(x, T) = \max\{k \cdot \setuptime + w^s_{x, u_k} + {\textit b\_time}^{s}(u_k, \BB{u_k, x})\mid 1 \le k \le h\}$.
\end{lemma}

Let $h_i$ = $|N_{\BB{v_i, x}}(v_i)|$. If $h_i = 0$,
then we set $\Psi_i(x) = \emptyset$.
Otherwise, the vertices of $N_{\BB{v_i, x}}(v_i)$ are arranged as $u_{i,1}, u_{i,2}, \ldots, u_{i,h_i}$
so that
$w^{\alpha_{x, v_i}}_{v_i, u_{i,k}} + {\textit b\_time}^{\alpha_{x, v_i}}(u_{i,k}, \BB{u_{i,k}, v_i})$ is nonincreasing as $k$ increases from 1 to $h_i$. For example, refer to Figure~\ref{figure:scenario}(b) again, where $w^{\alpha_{x, v_i}}_{v_i, u_{i,k}} + {\textit b\_time}^{\alpha_{x, v_i}}(u_{i,k}, \BB{u_{i,k}, v_i})=$ 13, 12, 7, as  $k$ = 1, 2, 3, respectively.
Then, we set $\Psi_{i}(x)= \{\beta_{x, v_i}^1, \beta_{x, v_i}^2, \ldots, \beta_{x, v_i}^{h_i}\}$,
where each $\beta_{x, v_i}^j$, $1 \le j \le h_i$, is a scenario of $T$ with respect to $x$ and $v_i$ that satisfies the following:
\begin{itemize}
\item $w^{\beta_{x, v_i}^j}_{a,b} = w^-_{a, b}$, if $(a, b) \in \bigcup_{j < k \le h_i}(\setof{(v_i, u_{i,k})}+E(\BB{u_{i,k}, v_i}))$;
\item $w^{\beta_{x, v_i}^j}_{a,b} = w^{\alpha_{x, v_i}}_{a, b}$, else.
\end{itemize}

Refer to Figure~\ref{figure:beta-scenario}, where $\beta_{x, v_i}^1$ and $\beta_{x, v_i}^2$ are illustrated for the example of Figure~\ref{figure:scenario}. Compared with $\alpha_{x, v_i}$, the changes (i.e., $w^{\beta_{x, v_i}^j}_{a,b} = w^-_{a, b}$) for $\beta_{x, v_i}^j$ are based on the sequence of $u_{i,1}, u_{i,2}, \ldots, u_{i,h_i}$. That is, $\beta_{x, v_i}^1$ changes
the weights of all edges in $(\setof{(v_i, u_{i,2})}+E(\BB{u_{i,2}, v_i})) \cup (\setof{(v_i, u_{i,3})}+E(\BB{u_{i,3}, v_i}))$, $\beta_{x, v_i}^2$ changes the weights of all edges in $(\setof{(v_i, u_{i,3})}+E(\BB{u_{i,3}, v_i}))$, and $\beta_{x, v_i}^3$ remains the same as $\alpha_{x, v_i}$. For this example, if $v_i$ is a prime broadcast center of $T$ under some worst-case scenario of $T$ with respect to $x$, then $\beta_{x, v_i}^1$ is a worst-case scenario of $T$ with respect to $x$ (i.e., $\beta_{x, v_i}^1 = \ddot{s}(x)$), as explained below.

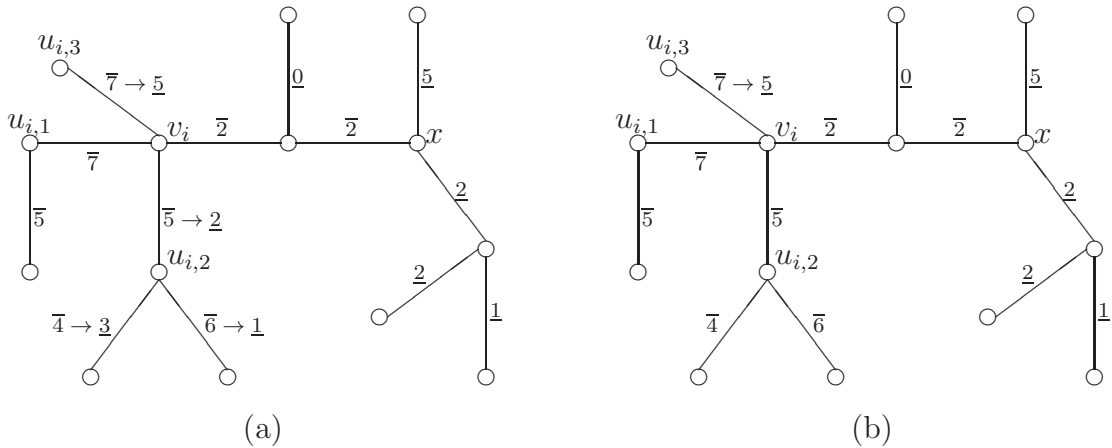
\begin{figure}[h]
\begin{center}
\unitlength=1mm
\begin{picture}(150, 60)(-3, -10)
\scriptsize
\put(0, 29){\circle{2}}
\put(17, 29){\circle{2}}
\put(34, 29){\circle{2}}
\put(51, 29){\circle{2}}
\put(34, 46){\circle{2}}
\put(51, 46){\circle{2}}
\put(0, 12){\circle{2}}
\put(17, 12){\circle{2}}
\put(4, 39){\circle{2}}
\put(8, -2){\circle{2}}
\put(26, -2){\circle{2}}
\put(60, 15){\circle{2}}
\put(60, -2){\circle{2}}
\put(46, 6){\circle{2}}

\put(1,29){\line(1,0){15}}
\put(18,29){\line(1,0){15}}
\put(35,29){\line(1,0){15}}
\put(34,30){\line(0,1){15}}
\put(51,30){\line(0,1){15}}
\put(0,13){\line(0,1){15}}
\put(17,13){\line(0,1){15}}
\put(8,-1){\line(3,4){9}}
\put(26,-1){\line(-3,4){9}}
\put(17,30){\line(-4,3){12}}
\put(60,16){\line(-3,4){9}}
\put(60,-1){\line(0,1){15}}
\put(47,6){\line(4,3){12}}

\put(7.5, 25.5){$\overline{7}$}
\put(24.5, 30){$\overline{2}$}
\put(41.5, 30){$\overline{2}$}
\put(34.5, 37){$\underline{0}$}
\put(51.5, 37){$\underline{5}$}
\put(0.5, 18){$\overline{5}$}
\put(17.5, 18){$\overline{5}\rightarrow\underline{2}$}
\put(23, 4){$\overline{6}\rightarrow\underline{1}$}
\put(3, 4){$\overline{4}\rightarrow\underline{3}$}
\put(10, 36){$\overline{7}\rightarrow\underline{5}$}
\put(56, 22){$\underline{2}$}
\put(50.5, 11){$\underline{2}$}
\put(60.5, 6){$\underline{1}$}

\put(80, 29){\circle{2}}
\put(97, 29){\circle{2}}
\put(114, 29){\circle{2}}
\put(131, 29){\circle{2}}
\put(114, 46){\circle{2}}
\put(131, 46){\circle{2}}
\put(80, 12){\circle{2}}
\put(97, 12){\circle{2}}
\put(84, 39){\circle{2}}
\put(88, -2){\circle{2}}
\put(106, -2){\circle{2}}
\put(140, 15){\circle{2}}
\put(140, -2){\circle{2}}
\put(126, 6){\circle{2}}

\put(81,29){\line(1,0){15}}
\put(98,29){\line(1,0){15}}
\put(115,29){\line(1,0){15}}
\put(114,30){\line(0,1){15}}
\put(131,30){\line(0,1){15}}
\put(80,13){\line(0,1){15}}
\put(97,13){\line(0,1){15}}
\put(88,-1){\line(3,4){9}}
\put(106,-1){\line(-3,4){9}}
\put(97,30){\line(-4,3){12}}
\put(140,16){\line(-3,4){9}}
\put(140,-1){\line(0,1){15}}
\put(127,6){\line(4,3){12}}

\put(87.5, 25.5){$\overline{7}$}
\put(104.5, 30){$\overline{2}$}
\put(121.5, 30){$\overline{2}$}
\put(114.5, 37){$\underline{0}$}
\put(131.5, 37){$\underline{5}$}
\put(80.5, 18){$\overline{5}$}
\put(97.5, 18){$\overline{5}$}
\put(103, 4){$\overline{6}$}
\put(89, 4){$\overline{4}$}
\put(90, 36){$\overline{7}\rightarrow\underline{5}$}
\put(136, 22){$\underline{2}$}
\put(130.5, 11){$\underline{2}$}
\put(140.5, 6){$\underline{1}$}

\normalsize
\put(52, 29){$x$}
\put(18, 30){$v_i$}
\put(-3, 31){$u_{i,1}$}
\put(18, 13){$u_{i,2}$}
\put(1, 41.5){$u_{i,3}$}
\put(132, 29){$x$}
\put(98, 30){$v_i$}
\put(77, 31){$u_{i,1}$}
\put(98, 13){$u_{i,2}$}
\put(81, 41.5){$u_{i,3}$}

\put(28, -10){(a)}
\put(108, -10){(b)}

\end{picture}
\caption{Two instances of $\beta_{x, v_i}^j$. (a) $\beta_{x, v_i}^1$. (b) $\beta_{x, v_i}^2$.}
\label{figure:beta-scenario}
\end{center}
\end{figure}

According to Lemma~\ref{lemma:Direct-to-center}, we have ${\textit b\_time}^{\ddot{s}(x)}(x, T)
= \dis_{x,v_i} \cdot \setuptime + \tilde{w}^{\ddot{s}(x)}_{x, v_i} + {\textit b\_time}^{\ddot{s}(x)}(v_i, \BB{v_i, x})$, i.e., ${\textit b\_time}^{\ddot{s}(x)}(x, T) - {\textit b\_time}^{\ddot{s}(x)}(v_i, \BB{v_i, x}) = \dis_{x,v_i} \cdot \setuptime + \tilde{w}^{\ddot{s}(x)}_{x, v_i}$. Further, we have $\dis_{x,v_i} \cdot \setuptime + \tilde{w}^{\ddot{s}(x)}_{x, v_i} \le \dis_{x,v_i} \cdot 1 + \tilde{w}^{\beta_{x, v_i}^1}_{x, v_i} = 2 + 4 = 6$, as a consequence of $\beta_{x, v_i}^1$ and $\alpha_{x, v_i}$. Since $r^{\ddot{s}(x)}_{x, v_i} =
{\textit b\_time}^{\ddot{s}(x)}(x, T) - {\textit b\_time}^{\ddot{s}(x)}(v_i, T) \le {\textit b\_time}^{\ddot{s}(x)}(x, T) - {\textit b\_time}^{\ddot{s}(x)}(v_i, \BB{v_i, x})$, we have $r^{\ddot{s}(x)}_{x, v_i} \le \dis_{x,v_i} \cdot \setuptime + \tilde{w}^{\beta_{x, v_i}^1}_{x, v_i} = 6$. On the other hand, we have $r^{\beta_{x, v_i}^1}_{x, v_i} = {\textit b\_time}^{\beta_{x, v_i}^1}(x, T) - {\textit b\_time}^{\beta_{x, v_i}^1}(v_i, T)$ $= 20 - 14 = 6 \ge r^{\ddot{s}(x)}_{x, v_i}$. Now that $v_i$ is a prime broadcast center of $T$ under $\ddot{s}(x)$, we have  $r^{\ddot{s}(x)}_{x, v_i}$ = ${\textit max\_r}(x)$. So, we have $r^{\beta_{x, v_i}^1}_{x, v_i}$ = ${\textit max\_r}(x)$, implying that $\beta_{x, v_i}^1$ is a worst-case scenario of $T$ with respect to $x$ (i.e., $\beta_{x, v_i}^1 = \ddot{s}(x)$).

In general, if $\ddot{s}(x) \not\in \Psi_n(x)$, then there exists $\beta_{x, v_i}^j = \ddot{s}(x)$ for some $1 \le i < n$ and $1 \le j \le h_i$,
where $v_i$ is a prime broadcast center of $T$ under some worst-case scenario of $T$ with respect to $x$. After $\ddot{S}(x)$ $(=\bigcup_{1\le \ell \le n}\Psi_{\ell}(x))$ is constructed, $\ddot{s}(x)$ can be determined from $\ddot{S}(x)$, as a consequence that $\max\{r^s_{x,y} \mid s \in \ddot{S}(x)$ and $y \in {\textit B\_Ctr}^{s}\}$
occurs when $s = \ddot{s}(x)$. The detailed execution for finding $\ddot{s}(x)$ is expressed as Algorithm 2 below.

\bigskip
\begin{algorithm}[H]
\caption{Finding $\ddot{s}(x)$}
\label{algorithm:worst-case}
\begin{algorithmic} [1]
\Require a tree $T$ with vertices $v_1, v_2, \ldots, v_n$ and a vertex $x$ $(=v_n)$ of $T$.
\Ensure $\ddot{s}(x)$.
\State set $\Psi_n(x) = \setof{c}$, where $c$ is an arbitrary scenario of $T$;
\For   {each $v_i \in V(T) - \setof{x}$}
\State construct $\alpha_{x, v_i}$;
\vspace{4pt}
\State $h_i \leftarrow |N_{\BB{v_i, x}}(v_i)|$;
\vspace{4pt}
%
\If {$h_i = 0$}
\State set $\Psi_i(x) = \emptyset$;
\Else
\State arrange the vertices of $N_{\BB{v_i, x}}(v_i)$ as $u_{i,1}, u_{i,2}, \ldots, u_{i,h_i}$ so that $w^{\alpha_{x, v_i}}_{v_i, u_{i,k}} +$
\Statex \hspace{35pt}${\textit b\_time}^{\alpha_{x, v_i}}(u_{i,k}, \BB{u_{i,k}, v_i})$ is nonincreasing as $k$ increases from 1 to $h_i$;
\vspace{4pt}
\State construct $\beta_{x, v_i}^1, \beta_{x, v_i}^2, \ldots, \beta_{x, v_i}^{h_i}$;
\vspace{4pt}
\State set $\Psi_i(x) = \{\beta_{x, v_i}^1, \beta_{x, v_i}^2, \ldots, \beta_{x, v_i}^{h_i}\}$;
\EndIf
\EndFor
\State determine $s' \in \bigcup_{1\le \ell \le n}\Psi_{\ell}(x)$ so that $r^{s'}_{x,y'} = \max\{r^s_{x,y} \mid s \in \bigcup_{1\le \ell \le n}\Psi_{\ell}(x)$ and
\Statex $y \in {\textit B\_Ctr}^{s}\}$ for some $y' \in {\textit B\_Ctr}^{s'}$;
\State \Return $s'$.
\end{algorithmic}
\end{algorithm}

The main result of this paper is stated as the following theorem, whose proof is shown in the following two sections.
\begin{theorem}
\label{theorem:correctness}
A minmax-regret broadcast center on a tree of $n$ vertices with edge weight uncertainty can be found in $O(n \log n \log \log n)$ time.
\end{theorem}

\section{Correctness}\label{section:correctness}
Clearly, with Lemma~\ref{lemma:Contain-optimal-solution},
a minmax-regret broadcast center $\kappa^*$ of $T$ can be found by Algorithm~\ref{algorithm:mrbc}, provided $\ddot{s}(x)$ can be determined for any $x \in V(T)$. So, it suffices to determine $\ddot{s}(x)$, i.e., to verify Algorithm~\ref{algorithm:worst-case}, in order to prove the correctness of Algorithm~\ref{algorithm:mrbc}. Recall that we have $\Psi_n(x) = \setof{c}$, where $c \in C$. If $x$ is a prime broadcast center of $T$ under $\ddot{s}(x)$, then
we have ${\textit max\_r}(x) = 0$.
Since $r^c_{x,y}
\ge 0 = {\textit max\_r}(x)$ for any $y \in {\textit B\_Ctr}^c$, we have $c = \ddot{s}(x) \in \Psi_n(x)$.
In the rest of this section, we assume that $x$ is not a prime broadcast center of $T$ under
any worst-case scenario of $T$ with respect to $x$ (which implies $|V(T)| \ge 3$). Besides, unless specified particularly, whenever a worst-case scenario $c$ is mentioned, it means that $c$ is a worst-case scenario of $T$ with respect to $x$.
%


In order to identify a worst-case scenario in some $\Psi_\ell(x)$,
where $1 \le \ell < n$, we perform a series of scenario transformations, starting with a worst-case scenario (it does exist, although its edge weights are unknown), denoted by $\ddot{s}(x)$. First, we transform $\ddot{s}(x)$ into another worst-case scenario by specifying some edge weights. This transformation can be realized by the following fact 
whose proof will
be presented in Section~\ref{subsection:fact1}.

\vspace{-5pt}
\begin{fact} \label{fact:Tx-is-base}
Suppose $x, y \in V(T)$ and $y$ is a
prime broadcast center of $T$ under $\ddot{s}(x)$.
Let $s$ be a scenario of $T$ that satisfies $w^{s}_{a, b} =
w^{\alpha_{x, y}}_{a, b}$ if $(a, b) \in \P_{x, y} \cup
E(\OB{y,x})$, and $w^{s}_{a, b} = w^{\ddot{s}(x)}_{a, b}$
else.
Then, $s$ is a
worst-case scenario of $T$ with respect to $x$. Besides, we have $y
\in {\textit B\_Ctr}^{s}$.
\end{fact}

\vspace{-5pt}\vspace{-5pt}
According to Fact~\ref{fact:Tx-is-base}, $\ddot{s}(x)$ is transformed into $s$, which is also a worst-case scenario. The weights of edges $(a, b) \in \P_{x, y} \cup
E(\OB{y,x})$ are specified under $s$. However, $y$ is not necessarily a prime broadcast center under $s$. The following fact finds out a worst-case scenario and an associated prime broadcast center.

%
\vspace{-5pt}
\begin{fact} \label{fact:Tx-is-base-with-y-prime}
Suppose $x \in V(T)$. There exist
$\ddot{s}(x)$ and $y \in V(T)$ satisfying
$w^{\ddot{s}(x)}_{a,b} = w^{\alpha_{x, y}}_{a,b}$ for each
$(a, b) \in \P_{x, y} \cup E(\OB{y, x})$. Besides, $y$ is a
prime broadcast center of $T$ under $\ddot{s}(x)$.
\end{fact}

\vspace{-5pt}\vspace{-5pt}
Fact~\ref{fact:Tx-is-base-with-y-prime} will be proved in Section~\ref{subsection:fact2}, where a transformation from $s$ (the worst-case scenario obtained by Fact~\ref{fact:Tx-is-base}) to another worst-case scenario is shown and an associated prime broadcast center is found.
To simplify notations, we use $\ddot{s}(x)$ and $y$ to represent the resulting worst-case scenario and prime broadcast center, respectively. They may differ from those used in Fact~\ref{fact:Tx-is-base}.

Finally, we transform $\ddot{s}(x)$ into another worst-case scenario
by specifying the weights of the other edges, i.e., all edges $(a, b) \in E(\BB{y, x})$
($=E(T)-\{P_{x, y} \cup E(\OB{y,x})\}$). The resulting scenario is identical with
some $\beta^j_{x, y}$ ($\in \Psi_\ell(x)$, as $y = v_\ell$), and thus Algorithm~\ref{algorithm:worst-case} is verified. The transformation, which depends on two cases, can be carried out by the following two facts, where the notations
$\mu$, $q_{0, \tau_0}$, $h$, and $t^*$ are introduced in subsequent paragraphs
($h = h_\ell$ and $\beta_{x,y}^h, \beta_{x,y}^{t^*} \in \Psi_\ell(x)$, as
$y = v_\ell$). Their proofs, which will be presented in Section~\ref{subsection:fact3} and Section~\ref{subsection:fact4},
show two transformations from $\ddot{s}(x)$ to $\beta^{h}_{x, y}$ and $\beta^{t^*}_{x, y}$, respectively.

\vspace{-5pt}
\begin{fact} \label{fact:critical-after}
If $w^{\ddot{s}(x)}_{y, \mu} +
{\textit b\_time}^{\ddot{s}(x)}(\mu, \OB{y, x}) \ge
w^{\ddot{s}(x)}_{y, q_{0, \tau_0}} + {\textit
b\_time}^{\ddot{s}(x)}(q_{0, \tau_0}, \BB{q_{0, \tau_0},
y})$, then $\beta^{h}_{x, y}$ is a
worst-case scenario of $T$ with respect to $x$.
\end{fact}
\vspace{-5pt}
\begin{fact}
\label{fact:critical-before}If $w^{\ddot{s}(x)}_{y, \mu} +
{\textit b\_time}^{\ddot{s}(x)}(\mu, \OB{y, x}) <
w^{\ddot{s}(x)}_{y, q_{0, \tau_0}} + {\textit
b\_time}^{\ddot{s}(x)}(q_{0, \tau_0}, \BB{q_{0, \tau_0},
y})$, then
$\beta^{t^*}_{x, y}$ is a worst-case
scenario of $T$ with respect to $x$.
\end{fact}

\vspace{-5pt}\vspace{-5pt}
We use $\mu$ to denote the neighbor of $y$ in $\OB{y, x}$, and 
$u_1, u_2, \ldots, u_{h}$
to denote the neighbors of $y$ in $\BB{y, x}$ such that $w^{\alpha_{x,
y}}_{y, u_k} + {\textit b\_time}^{\alpha_{x, y}}(u_k, \BB{u_k,
y})$ is nonincreasing as $k$ increases from $1$ to $h$.
%
%
%
Refer to Figure~\ref{figure:gamma_scenario}.
Let $\gamma_{x, y}^0, \gamma_{x, y}^1, \gamma_{x, y}^2, \ldots, \gamma_{x, y}^{h}$
denote $h+1$ scenarios of $T$ with respect to $x$ and $y$, where $\gamma_{x, y}^0 = \ddot{s}(x)$ ($\ddot{s}(x)$ was obtained by Fact~\ref{fact:Tx-is-base-with-y-prime}) and
each $\gamma_{x, y}^j$, $1 \le j \le h$, satisfies the following:

\vspace{-6pt}
\begin{itemize}
\setlength{\itemsep}{0pt}
\item $w^{\gamma_{x, y}^j}_{a,b} = w^{\alpha_{x, y}}_{a, b}$, if $(a, b) \in \bigcup_{1 \le k \le j}(\setof{(y, u_{k})} + E(\BB{u_{k}, y}))$;
\item $w^{\gamma_{x, y}^j}_{a,b} = w^{\ddot{s}(x)}_{a, b}$, else.
\end{itemize}
\vspace{-6pt}
Since the weights of edges in $\setof{(y, \mu)} \cup E(\OB{y, x})$, i.e.,
the subtree at the right-bottom of Figure~\ref{figure:gamma_scenario},
were specified by Fact~\ref{fact:Tx-is-base-with-y-prime},
only the weights of edges in $\bigcup_{j < k \le h}(\setof{(y, u_{k})} + E(\BB{u_{k}, y}))$
are not specified under $\gamma_{x,y}^j$, where
$\gamma_{x,y}^j$ and $\beta_{x,y}^j$
differ (refer to Figure~\ref{figure:gamma_scenario}).

\begin{figure}[h]
\begin{center}
\unitlength=1mm
\begin{picture}(104, 72)(-20,0)
\put(29, 30){\circle{2}}
\put(41, 20){\circle{2}}
\put(59, 20){\circle{2}}
\put(46, 33){\circle{2}}
\put(38, 44){\circle{2}}
\put(12, 33){\circle{2}}
\put(20, 44){\circle{2}}
\put(12, 27){\circle{2}}

\put(30,30){\line(1,-1){10}}
\put(30,30){\line(5,1){15}}
\put(29,31){\line(3,4){9}}
\put(28,30){\line(-5,1){15}}
\put(29,31){\line(-3,4){9}}
\put(28,30){\line(-5,-1){15}}

\put(42,20){\line(1,0){3}}
\put(47.5,20){$\ldots$}
\put(58,20){\line(-1,0){3}}
\put(60,20){\line(1,0){3}}
\put(65.5,20){$\ldots$}

\put(28, 26){$y$}
\put(39, 22){$\mu$}
\put(42, 35){$u_h$}
\put(11, 35){$u_2$}
\put(57, 22){$x$}
\put(11, 23){$u_1$}
\put(39, 41){$u_{j+1}$}
\put(16, 41){$u_{j}$}

\qbezier(41,21)(76, 33)(76,20)
\qbezier(41,19)(76, 7)(76,20)

\qbezier(46, 34)(73, 50)(78, 41)
\qbezier(46, 32)(83, 27)(78, 41)

\qbezier(38,45)(40, 77)(52,72)
\qbezier(39,44)(62, 66)(52,72)

\qbezier(12,34)(-15, 50)(-20,41)
\qbezier(12,32)(-25, 27)(-20,41)

\qbezier(20,45)(18, 77)(6,72)
\qbezier(19,44)(-4, 66)(6,72)

\qbezier(12,26)(-15, 10)(-20,19)
\qbezier(12,28)(-25, 33)(-20,19)

\put(66, 23){$\OB{y, x}$}
\put(65, 36){$\BB{u_h, y}$}
\put(-15, 36){$\BB{u_2, y}$}
\put(42, 64){$\BB{u_{j+1}, y}$}
\put(6, 64){$\BB{u_j, y}$}
\put(-15, 21){$\BB{u_1, y}$}

\put(-26,0){\dashbox(50, 77)[tl]{\shortstack[l]{\\ \\$\hspace{6pt} w^{\gamma^j_{x,y}}_{a,b} = w^{\alpha_{x,y}}_{a,b}$ \\ $\hspace{6.5pt} ~~~~~~= w^{\beta^j_{x,y}}_{a,b}$}}}
\put(34,0){\dashbox(55, 77)[tr]{\shortstack[r]{\\ \\ $\hspace{6pt} w^{\gamma^j_{x,y}}_{a,b} = w^{\ddot{s}(x)}_{a,b}~$ \\ $\hspace{6pt} w^{\beta^j_{x,y}}_{a,b} = w^{-}_{a,b}~~\hspace{1.5pt}$}}}
\put(34,0){\dashbox(55, 77)[br]{\shortstack[r]{$w^{\gamma^j_{x,y}}_{a,b} = w^{\ddot{s}(x)}_{a,b} = w^{\alpha_{x,y}}_{a,b} = w^{\beta^j_{x,y}}_{a,b}~$\\ \\ $ $}}}
\put(34,28.5){\dashbox(55, 0)}

\linethickness{1.5pt}
\qbezier[2](18,34)(19.25, 35.25)(20.5,36.5)
\qbezier[2](40,34)(38.75, 35.25)(37.5,36.5)

\end{picture}
\caption{Comparison between $\gamma^j_{x,y}$ and $\beta^j_{x,y}$.}
\label{figure:gamma_scenario}
\end{center}
\end{figure}
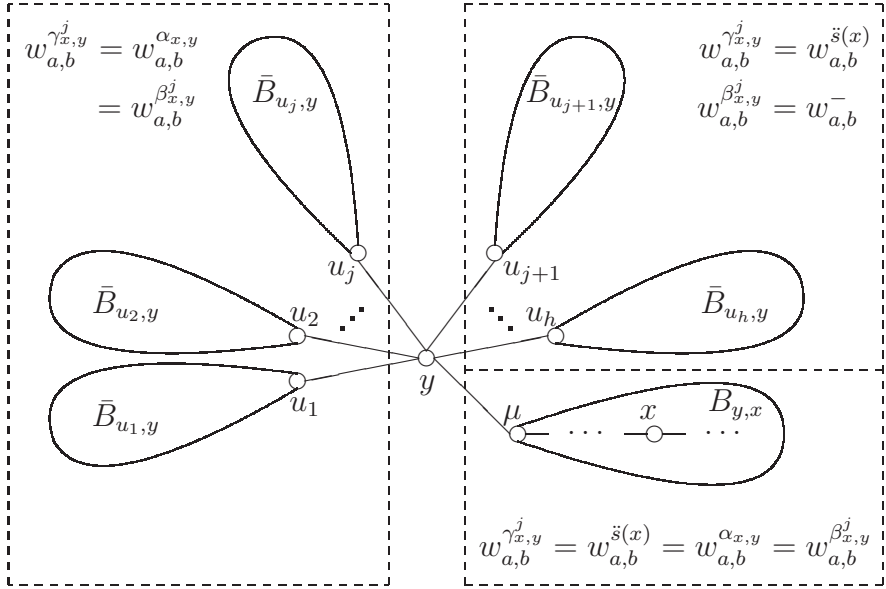

%

Let
$Q_0 = (q_{0,1}, q_{0,2}, \ldots, q_{0,h})$ be an arrangement of
$u_1, u_2, \ldots, u_h$ such that $w^{\gamma_{x, y}^0}_{y,
q_{0,k}} + {\textit b\_time}^{\gamma_{x, y}^0}(q_{0,k},
\BB{q_{0,k}, y})$ is nonincreasing as $k$ increases from $1$ to
$h$. According to Lemma~\ref{lemma:optimal-sequence},
$Q_0$ is an optimal sequence for $y$ to broadcast a message to
$N_{\BB{y, x}}(y)$ under $\gamma_{x, y}^0$. Moreover, for $1 \le j \le h$, let $Q_j  = (q_{j,1}, q_{j,2}, \ldots, q_{j,h})$ be the sequence obtained from $Q_{j-1} = (q_{j-1,1}, q_{j-1,2}, \ldots, q_{j-1,h})$ by cyclically shifting $q_{j-1,j}, q_{j-1,j+1}, \ldots, q_{j-1,\ell}$ one position toward the right, where $q_{j-1, \ell} = u_j$ is assumed and $j \le \ell$.
That is, we have $(q_{j,1}, \ldots, q_{j,j-1}) =
(q_{j-1,1},  \ldots, q_{j-1,j-1})$,
$(q_{j,j}, q_{j,j+1}, q_{j,j+2}, \ldots, q_{j,\ell}) =
(q_{j-1,\ell}, q_{j-1,j}, q_{j-1,j+1}, \ldots, q_{j-1,\ell-1})$, and
$(q_{j, \ell+1}, \ldots, q_{j,h}) = (q_{j-1,\ell+1}, \ldots, q_{j-1,h})$. It is not difficult to see that
the first $j$ vertices in
$Q_j$ are $u_1, u_2, \ldots, u_j$.

Also notice that $w^{\gamma_{x, y}^j}_{y,
q_{j,k}} + {\textit b\_time}^{\gamma_{x, y}^j}(q_{j,k},
\BB{q_{j,k}, y})$ is nonincreasing as $k$ increases from $1$ to $h$, as a consequence of $\gamma_{x,y}^j$ (refer to Figure~\ref{figure:gamma_scenario}). According to
Lemma~\ref{lemma:optimal-sequence}, $Q_j$ is an optimal sequence
for $y$ to broadcast a message to $N_{\BB{y, x}}(y)$ under
$\gamma_{x, y}^j$.
Let $q_{j, \tau_j}$ be the first vertex in $Q_j$ with
${\textit b\_time}^{\gamma_{x, y}^j}(y, \BB{y, x}) =
\tau_j \cdot \setuptime + w^{\gamma_{x, y}^j}_{y, q_{j, \tau_j}} + {\textit
b\_time}^{\gamma_{x, y}^j}(q_{j, \tau_j}, \BB{q_{j, \tau_j}, y})$, where the right-hand side is the broadcast time induced by
$\BB{q_{j, \tau_j}, y}$ ($q_{0, \tau_0}$ has the same definition as $q_{j, \tau_j}$).

Define $Q_{t^*}$ to be the
first sequence in $Q_1, Q_2, \ldots, Q_h$ that has $q_{t^*, \tau_{t^*}} =q_{t^*, t^*}$. The proof of Fact~\ref{fact:critical-before} shows that $\gamma_{x,y}^{t^*}$ is a worst-case scenario and consequently, $\beta^{t^*}_{x, y}$ is a worst-case scenario (recall that $\gamma_{x,y}^{t^*}$ and $\beta^{t^*}_{x, y}$ differ in the weights of edges in $\bigcup_{t^* < k \le h}(\setof{(y, u_{k})} + E(\BB{u_{k}, y}))$.
Below the existence of $Q_{t^*}$ is guaranteed, by
showing that
$Q_{\tau_h}$ has $q_{\tau_h, \tau_{\tau_h}} =q_{\tau_h, \tau_h}$ (hence, $t^* \le \tau_h$).


Since $Q_{\tau_h}$ is an optimal sequence
for $y$ to broadcast a message to $N_{\BB{y, x}}(y)$ under
$\gamma_{x, y}^{\tau_h}$ and
$u_{\tau_h} (=q_{\tau_h, \tau_h})$ is the $\tau_h$-th
vertex in $Q_{\tau_h}$, we have
${\textit b\_time}^{\gamma_{x, y}^{\tau_h}}(y, \BB{y, x}) \ge \tau_h \cdot \setuptime +
w^{\gamma_{x, y}^{\tau_h}}_{y, q_{\tau_h, \tau_h}} + {\textit
b\_time}^{\gamma_{x, y}^{\tau_h}}(q_{\tau_h,\tau_h},
\BB{q_{\tau_h,\tau_h}, y})$, i.e., the broadcast time induced by $\BB{q_{\tau_h,\tau_h}, y}$. Besides, we have
 $q_{\tau_h, \tau_h} = u_{\tau_h} = q_{h, \tau_h}$,
because $u_{\tau_h}$ is also the $\tau_h$-th vertex in $Q_h$.
Since $\gamma_{x,y}^{\tau_h}$ and
$\gamma_{x,y}^{h}$ have the same weights of those
edges in $\setof{(y, u_{\tau_h})} + E(\BB{u_{\tau_h}, y})$, we have
$w^{\gamma_{x, y}^{\tau_h}}_{y, u_{\tau_h}} + {\textit
b\_time}^{\gamma_{x, y}^{\tau_h}}(u_{\tau_h},
\BB{u_{\tau_h}, y}) =
w^{\gamma_{x, y}^{h}}_{y, u_{\tau_h}} + {\textit
b\_time}^{\gamma_{x, y}^{h}}(u_{\tau_h},
\BB{u_{\tau_h}, y})$.
Consequently, we have
\medskip

$\hspace{20pt}{\textit b\_time}^{\gamma_{x, y}^{\tau_h}}(y, \BB{y, x})~\ge~ \tau_h \cdot \setuptime +
w^{\gamma_{x, y}^{\tau_h}}_{y, q_{\tau_h, \tau_h}} + {\textit
b\_time}^{\gamma_{x, y}^{\tau_h}}(q_{\tau_h,\tau_h},
\BB{q_{\tau_h,\tau_h}, y})$


$\hspace{110.5pt}~=~ \tau_h \cdot \setuptime +
w^{\gamma_{x, y}^{h}}_{y, q_{h, \tau_h}} + {\textit
b\_time}^{\gamma_{x, y}^{h}}(q_{h,\tau_h},
\BB{q_{h,\tau_h}, y})$

$\hspace{110.5pt}~=~ {\textit b\_time}^{\gamma_{x, y}^{h}}(y, \BB{y, x})$.

\medskip
On the other hand, since $\tau_h \le h$, we have ${\textit b\_time}^{\gamma_{x, y}^{\tau_h}}(y, \BB{y, x})
\le {\textit b\_time}^{\gamma_{x, y}^{h}}(y, \BB{y, x})$ (refer to Figure~\ref{figure:gamma_scenario}).
Therefore, ${\textit b\_time}^{\gamma_{x, y}^{\tau_h}}(y, \BB{y, x}) =
{\textit b\_time}^{\gamma_{x, y}^{h}}(y, \BB{y, x})$ holds, which implies
${\textit b\_time}^{\gamma_{x, y}^{\tau_h}}(y, \BB{y, x})
= \tau_h \cdot \setuptime + w^{\gamma_{x, y}^{\tau_h}}_{y, q_{\tau_h, \tau_h}} + {\textit
b\_time}^{\gamma_{x, y}^{\tau_h}}(q_{\tau_h,\tau_h}, \BB{q_{\tau_h,\tau_h}, y})$.
Moreover, we have 

\medskip
$\hspace{20pt}{\textit b\_time}^{\gamma_{x, y}^{\tau_h}}(y, \BB{y, x})
~=~ {\textit b\_time}^{\gamma_{x, y}^{h}}(y, \BB{y, x})$


$\hspace{110.5pt}~>~ k \cdot \setuptime + w^{\gamma_{x, y}^{h}}_{y, q_{h, k}} + {\textit
b\_time}^{\gamma_{x, y}^{h}}(q_{h,k}, \BB{q_{h,k}, y})$

$\hspace{110.5pt}~=~ k \cdot \setuptime + w^{\gamma_{x, y}^{\tau_h}}_{y, q_{\tau_h, k}} + {\textit
b\_time}^{\gamma_{x, y}^{\tau_h}}(q_{\tau_h,k}, \BB{q_{\tau_h,k}, y})$
\medskip
\\for all $1 \le k < \tau_h$. It follows that
we have $q_{\tau_h, \tau_{\tau_h}}= q_{\tau_h,\tau_h}$.

In the rest of this section, we show the proofs of Fact~\ref{fact:Tx-is-base}, Fact~\ref{fact:Tx-is-base-with-y-prime}, Fact~\ref{fact:critical-after}, and Fact~\ref{fact:critical-before}.

\subsection{Proof of Fact~\ref{fact:Tx-is-base}}
\label{subsection:fact1}
In order to prove Fact~\ref{fact:Tx-is-base}, it suffices to show
$r^{s}_{x,y} \ge r^{\ddot{s}(x)}_{x, y}$. As explained below, we first claim \vspace{-5pt}

\vspace{-12pt}
\begin{center}
${\textit b\_time}^{s}(y, T) \le
{\textit b\_time}^{\ddot{s}(x)}(y, T) + \tilde{w}^{s}_{x, y} -
\tilde{w}^{\ddot{s}(x)}_{x, y}.$
\end{center}

\vspace{-5pt}
Without loss of generality, suppose $N_T(y) = \{v_1, v_2, \ldots,
v_h\}$, where
$w^{\ddot{s}(x)}_{y, v_k} + \linebreak {\textit b\_time}^{\ddot{s}(x)}(v_k,
\BB{v_k, y})$ is nonincreasing as $k$ increases from $1$ to $h$.
According to Lemma~\ref{lemma:optimal-sequence}, ($v_1, v_2, \ldots, v_h$) is
an optimal sequence for $y$ to broadcast a message to $N_T(y)$
under $\ddot{s}(x)$. Besides, we have ${\textit b\_time}^{\ddot{s}(x)}(y, T) =
\max\{k \cdot \setuptime + w^{\ddot{s}(x)}_{y, v_k} + {\textit
b\_time}^{\ddot{s}(x)}(v_k, \BB{v_k, y})\mid 1 \le k \le h\}$.
Since ${\textit b\_time}^{s}(y, T) \le
\max\{k \cdot \setuptime + w^{s}_{y, v_k} + {\textit b\_time}^{s}(v_k, \BB{v_k,
y})\mid 1 \le k \le h\}$, the claim
is true provided $w^{s}_{y, v_k} + {\textit b\_time}^{s}(v_k,
\BB{v_k, y}) \le (w^{\ddot{s}(x)}_{y, v_k} + {\textit
b\_time}^{\ddot{s}(x)}(v_k, \BB{v_k, y})) + \tilde{w}^{s}_{x,
y} - \tilde{w}^{\ddot{s}(x)}_{x, y}$ for $1 \le k \le h$.
When $x \not\in V(\BB{v_k, y})$, we have
$w^{s}_{y, v_k} + {\textit
b\_time}^{s}(v_k, \BB{v_k, y}) = w^{\ddot{s}(x)}_{y, v_k} +
{\textit b\_time}^{\ddot{s}(x)}(v_k, \BB{v_k, y})$ and
$\tilde{w}^{s}_{x, y} = \tilde{w}^{\alpha_{x, y}}_{x, y} \ge \tilde{w}^{\ddot{s}(x)}_{x, y}$, as a consequence of the scenario $s$. The claim then follows.

So, we assume $x \in V(\BB{v_k, y})$. Let $s'$ be a scenario of $T$, which has
$w^{s'}_{a, b} = w^{\alpha_{x, y}}_{a, b}$ if \linebreak $(a, b) \in \P_{x,
y}$, and $w^{s'}_{a, b} = w^{\ddot{s}(x)}_{a, b}$ else. Then, we have ${\textit
b\_time}^{s}(v_k, \BB{v_k, y}) \le {\textit b\_time}^{s'}(v_k,
\BB{v_k, y})$. Since $s'$ and $\ddot{s}(x)$ differ in the edge weights of $\P_{x, y}$, we have ${\textit b\_time}^{s'}(v_k,
\BB{v_k, y}) \le \linebreak {\textit b\_time}^{\ddot{s}(x)}(v_k, \BB{v_k,
y}) + \tilde{w}^{s'}_{v_k, x} - \tilde{w}^{\ddot{s}(x)}_{v_k, x
}$ (the right-hand side is an upper bound on the time requirement for $v_k$ to broadcast over $\BB{v_k, y}$ under $s'$ using the optimal broadcast sequence with respect to ${\textit b\_time}^{\ddot{s}(x)}(v_k, \BB{v_k,
y})$). Further, since $\tilde{w}^{s'}_{v_k, x}=\tilde{w}^{s}_{v_k, x}$, $\tilde{w}^{s}_{x, y} =
w^{s}_{y, v_k} + \tilde{w}^{s}_{v_k, x}$, and $\tilde{w}^{\ddot{s}(x)}_{x, y} = w^{\ddot{s}(x)}_{y, v_k} + \tilde{w}^{\ddot{s}(x)}_{v_k,
x}$, the claim follows, because

\medskip
$\hspace{61.5pt}w^{s}_{y, v_k} + {\textit
b\_time}^{s}(v_k, \BB{v_k, y})$

$\hspace{37.5pt} ~\le~ w^{s}_{y, v_k} + {\textit
b\_time}^{\ddot{s}(x)}(v_k, \BB{v_k, y}) + \tilde{w}^{s}_{v_k,
x} - \tilde{w}^{\ddot{s}(x)}_{v_k, x}$


$\hspace{37.5pt}~=~ w^{\ddot{s}(x)}_{y, v_k} + {\textit
b\_time}^{\ddot{s}(x)}(v_k, \BB{v_k, y}) + \tilde{w}^{s}_{x,
y} - \tilde{w}^{\ddot{s}(x)}_{x, y}$.
\medskip
%
%

Then, the lemma holds, as a consequence of

\medskip

$\hspace{20pt}r^{s}_{x, y} ~=~  {\textit b\_time}^{s}(x, T) - {\textit b\_time}^{s}(y, T)$

$\hspace{37.5pt}~\ge~ (\dis_{x, y} \cdot \setuptime + \tilde{w}^{s}_{x, y} + {\textit b\_time}^{s}(y, \BB{y, x}))
                    - ({\textit b\_time}^{\ddot{s}(x)}(y, T) + \tilde{w}^{s}_{x, y} - \tilde{w}^{\ddot{s}(x)}_{x, y})$

\hspace{80pt} (by Lemma~\ref{lemma:general-x-y} and the claim)

              $\hspace{37.5pt}~=~ \dis_{x, y} \cdot \setuptime + \tilde{w}^{\ddot{s}(x)}_{x, y} + {\textit b\_time}^{\ddot{s}(x)}(y, \BB{y, x}) -
                    {\textit b\_time}^{\ddot{s}(x)}(y, T)$

\hspace{80pt}   ($\because {\textit b\_time}^{s}(y, \BB{y, x})= {\textit b\_time}^{\ddot{s}(x)}(y, \BB{y, x})$)

              $\hspace{37.5pt}~=~   {\textit b\_time}^{\ddot{s}(x)}(x, T) - {\textit b\_time}^{\ddot{s}(x)}(y, T)$ ~~(by Lemma~\ref{lemma:Direct-to-center})

              $\hspace{37.5pt}~=~   r^{\ddot{s}(x)}_{x,y}.$

\bigskip

\subsection{Proof of Fact~\ref{fact:Tx-is-base-with-y-prime}}
\label{subsection:fact2}


Let $s_1$ be the worst-case scenario and $y_1 \in {\textit B\_Ctr}^{s_1}$, both obtained by Fact~\ref{fact:Tx-is-base}. If $y_1$ is a
prime broadcast center of $T$ under
$s_1$, then $s_1$ is a desired $\ddot{s}(x)$. Otherwise, it is implied by
Lemma~\ref{lemma:prime-is-the-center-of-star} that the
prime broadcast center, say $\hat{\kappa}_1$, of $T$ under $s_1$
is neighboring to $y_1$ in $T$. If $\hat{\kappa}_1$ $(\not=x)$
is on the $x$-to-$y_1$ path,
then we have $\P_{x, \hat{\kappa}_1} \cup
E(\OB{\hat{\kappa}_1, x}) \subset \P_{x, y_1} \cup
E(\OB{y_1, x})$. Besides, 
we have $w^{s_1}_{a, b} = w^{\alpha_{x, y_1}}_{a, b} =
w^{\alpha_{x, \hat{\kappa}_1}}_{a, b}$ for each $(a, b) \in \P_{x, \hat{\kappa}_1} \cup
E(\OB{\hat{\kappa}_1,x})$, implying that
$s_1$ is a desired $\ddot{s}(x)$.

If $\hat{\kappa}_1$ is
not on the $x$-to-$y_1$ path, i.e.,  $\dis_{x,
\hat{\kappa}_1} = \dis_{x, y_1} +1$,
then according to Fact~\ref{fact:Tx-is-base}, there is another worst-case scenario $s_2$, which has $w^{s_2}_{a, b} =
w^{\alpha_{x, \hat{\kappa}_1}}_{a, b}$ if $(a, b) \in \P_{x, \hat{\kappa}_1} \cup
E(\OB{\hat{\kappa}_1,x})$, and $w^{s_2}_{a, b} = w^{s_1}_{a, b}$ else, where $\hat{\kappa}_1
\in {\textit B\_Ctr}^{s_2}$. Similarly, it can be concluded that either $s_2$ is a desired $\ddot{s}(x)$ or $\dis_{x, \hat{\kappa}_2} =
\dis_{x, \hat{\kappa}_1} +1$ holds,
where $\hat{\kappa}_2$ is a prime broadcast center of $T$ under $s_2$.
Thus, a desired $\ddot{s}(x)$ will be obtained eventually, if the above process is done repeatedly.

\subsection{Proof of Fact~\ref{fact:critical-after}}
\label{subsection:fact3}

Notice that $\ddot{s}(x)$ and $y$ used in Fact~\ref{fact:critical-after} are inherited from Fact~\ref{fact:Tx-is-base-with-y-prime}. The notations $\mu$, $h$, $\gamma^0_{x,y}$,
$Q_0$, and $\tau_0$, which will be used below, were defined earlier (after Fact~\ref{fact:critical-before}) in this section. Since $\beta^{h}_{x, y} = \alpha_{x, y}$, it suffices to show $r^{\alpha_{x,
y}}_{x,y} \ge r^{\ddot{s}(x)}_{x, y}$, in order to prove Fact~\ref{fact:critical-after}.

Let $Q'$ be the sequence that is obtained from $Q_0$ by inserting $\mu$
between $q_{0, p-1}$ and $q_{0, p}$, where $w^{\ddot{s}(x)}_{y,
q_{0, p-1}} + {\textit b\_time}^{\ddot{s}(x)}(q_{0, p-1},
\BB{q_{0, p-1}, y}) > w^{\ddot{s}(x)}_{y,
\mu} + {\textit b\_time}^{\ddot{s}(x)}(\mu,
\OB{y,x}) \ge w^{\ddot{s}(x)}_{y,
q_{0, p}} + {\textit b\_time}^{\ddot{s}(x)}(q_{0, p},
\BB{q_{0, p}, y})$. That is,
we have $Q' = (q'_{1}, \ldots,q'_{p-1}, q'_{p}, q'_{p+1}, \ldots, q'_{h+1})
= (q_{0, 1}, \ldots, \linebreak q_{0, p-1},  \mu, q_{0, p}, q_{0,p+1}, \ldots, q_{0, h})$.
According to Lemma~\ref{lemma:optimal-sequence}, $Q'$ is
an optimal sequence for $y$ to broadcast a message to $N_T(y)$ under
$\ddot{s}(x)$. Since
$q'_{p} = \mu$ and $q'_{\tau_0+1} = q_{0, \tau_0}$, we have

\medskip
$\hspace{13.5pt}{\textit
b\_time}^{\ddot{s}(x)}(y, T) ~=~\max\{p \cdot \setuptime +
w^{\ddot{s}(x)}_{y, \mu} + {\textit b\_time}^
{\ddot{s}(x)}(\mu, \OB{y, x}),$

$\hspace{143.5pt} (1 + \tau_0) \cdot \setuptime +
w^{\ddot{s}(x)}_{y, q_{0, \tau_0}} + {\textit b\_time}^
{\ddot{s}(x)}(q_{0, \tau_0}, \BB{q_{0, \tau_0}, y})\}$
\vspace{-8pt}
\begin{equation}
\hspace{110.5pt}~=~ \max\{p \cdot \setuptime +
w^{\ddot{s}(x)}_{y, \mu} + {\textit b\_time}^
{\ddot{s}(x)}(\mu, \OB{y, x}), \setuptime + {\textit
b\_time}^{\ddot{s}(x)}(y, \BB{y, x})\}.
\label{equation:y-T-time}
\end{equation}

\vspace{-5pt}
Recall that $Q_h = (q_{h,1}, q_{h,2}, \ldots, q_{h,h})$
is an optimal sequence for $y$ to broadcast a message to $N_{\BB{y,x}}(y)$ under
$\gamma^h_{x,y}$ ($=\alpha_{x,y}$).
Similarly, inserting $\mu$ between $q_{h, p-1}$ and $q_{h, p}$,
we have ${\textit b\_time}^{\alpha_{x,
y}}(y, T) \le \max\{p \cdot \setuptime + w^{\alpha_{x, y}}_{y,
\mu} + {\textit b\_time}^ {\alpha_{x, y}}(\mu, \OB{y,
x}), \setuptime+{\textit b\_time}^{\alpha_{x, y}}(y, \BB{y,
x})\}$.

If $p \cdot \setuptime + w^{\alpha_{x, y}}_{y, \mu} +
{\textit b\_time}^{\alpha_{x, y}}(\mu, \OB{y, x}) \ge \setuptime+
{\textit b\_time}^{\alpha_{x, y}}(y, \BB{y, x})$,
then ${\textit b\_time}^{\alpha_{x,
y}}(y, T) \le p \cdot \setuptime + w^{\alpha_{x, y}}_{y,
\mu} + {\textit b\_time}^ {\alpha_{x, y}}(\mu, \OB{y,
x})$.
Since
$w^{\ddot{s}(x)}_{y, \mu} + {\textit
b\_time}^{\ddot{s}(x)}(\mu, \OB{y, x}) = w^{\alpha_{x,
y}}_{y, \mu} + {\textit b\_time}^{\alpha_{x,
y}}(\mu, \OB{y, x})$ (as a consequence of Fact~\ref{fact:Tx-is-base-with-y-prime}), we have

\medskip
$\hspace{40pt} {\textit
b\_time}^{\ddot{s}(x)}(y, T) ~\ge~ p \cdot \setuptime +
w^{\ddot{s}(x)}_{y, \mu} + {\textit b\_time}^
{\ddot{s}(x)}(\mu, \OB{y, x})$

$\hspace{118.5pt}~=~ p \cdot \setuptime + w^{\alpha_{x,
y}}_{y, \mu} + {\textit b\_time}^{\alpha_{x,
y}}(\mu, \OB{y, x})$

$\hspace{118.5pt}~\ge~ {\textit b\_time}^{\alpha_{x, y}}(y, T)$

%
\medskip
Further, since ${\textit
b\_time}^{\alpha_{x, y}}(x, \BB{y, x}) \ge {\textit
b\_time}^{\ddot{s}(x)}(x, \BB{y, x})$, we have

\medskip
$\hspace{20pt}r^{\alpha_{x, y}}_{x,y} ~=~ {\textit b\_time}^{\alpha_{x, y}}(x, T) - {\textit b\_time}^{\alpha_{x, y}}(y, T)$

$\hspace{42.75pt}~\ge~ \dis_{x, y} \cdot \setuptime + \tilde{w}^{\alpha_{x, y}}_{x, y} + {\textit b\_time}^{\alpha_{x, y}}(y, \BB{y, x}) - {\textit b\_time}^{\alpha_{x, y}}(y, T)$ ~~(by Lemma~\ref{lemma:general-x-y})

$\hspace{42.75pt}~\ge~ \dis_{x, y} \cdot \setuptime + \tilde{w}^{\ddot{s}(x)}_{x, y} + {\textit b\_time}^{\ddot{s}(x)}(y, \BB{y, x}) - {\textit b\_time}^{\alpha_{x, y}}(y, T)$

$\hspace{42.75pt}~\ge~ \dis_{x, y} \cdot \setuptime + \tilde{w}^{\ddot{s}(x)}_{x, y} + {\textit b\_time}^{\ddot{s}(x)}(y, \BB{y, x}) - {\textit b\_time}^{\ddot{s}(x)}(y, T)$

$\hspace{42.75pt}~=~ {\textit b\_time}^{\ddot{s}(x)}(x, T) - {\textit b\_time}^{\ddot{s}(x)}(y, T)$ ~~ (by Lemma~\ref{lemma:Direct-to-center})
\vspace{-12pt}
\begin{equation}
\hspace{-284pt}~=~ r^{\ddot{s}(x)}_{x, y}.
\label{equation:fact3-case1}
\end{equation}

%
%
%
If $p \cdot \setuptime + w^{\alpha_{x, y}}_{y, \mu} +
{\textit b\_time}^{\alpha_{x, y}}(\mu, \OB{y, x}) < \setuptime +
{\textit b\_time}^{\alpha_{x, y}}(y, \BB{y,
x})$, then ${\textit b\_time}^{\alpha_{x,
y}}(y, T) \le \setuptime+ \linebreak {\textit b\_time}^{\alpha_{x, y}}(y, \BB{y,
x})$, and hence

\medskip
$\hspace{17pt}r^{\alpha_{x, y}}_{x,y} ~=~ {\textit b\_time}^{\alpha_{x,
y}}(x, T) - {\textit b\_time}^{\alpha_{x,
y}}(y, T)$

$\hspace{39.5pt}~\ge~ \dis_{x, y} \cdot \setuptime +
\tilde{w}^{\alpha_{x, y}}_{x, y} + {\textit
b\_time}^{\alpha_{x, y}}(y, \BB{y, x})
- {\textit b\_time}^{\alpha_{x,
y}}(y, T)$ ~~(by Lemma~\ref{lemma:general-x-y})

$\hspace{39.5pt}~\ge~ \dis_{x, y} \cdot \setuptime +
\tilde{w}^{\alpha_{x, y}}_{x, y} + {\textit
b\_time}^{\alpha_{x, y}}(y, \BB{y, x})
- (\setuptime+{\textit b\_time}^{\alpha_{x, y}}(y, \BB{y,
x}))$

$\hspace{39.5pt}~=~ \dis_{x, y} \cdot \setuptime +
\tilde{w}^{\alpha_{x, y}}_{x, y} - \setuptime$
\vspace{0pt}
\begin{equation}
\hspace{-147.5pt}~=~ \dis_{x,
y} \cdot \setuptime + \tilde{w}^{\ddot{s}(x)}_{x, y} - \setuptime \textrm{~~~(by Fact~\ref{fact:Tx-is-base-with-y-prime})}.
\label{equation:regret-xy}
\end{equation}

\vspace{5pt}
Moreover, we have ${\textit
b\_time}^{\ddot{s}(x)}(x, T) =\dis_{x, y} \cdot \setuptime +
\tilde{w}^{\ddot{s}(x)}_{x, y} + {\textit
b\_time}^{\ddot{s}(x)}(y, \BB{y, x})$ by
Lemma~\ref{lemma:Direct-to-center}. Since
${\textit b\_time}^{\ddot{s}(x)}(y, T) \ge \setuptime+{\textit
b\_time}^{\ddot{s}(x)}(y, \BB{y, x})$ (according to (\ref{equation:y-T-time})), we have

\medskip
$\hspace{20pt}r^{\ddot{s}(x)}_{x,y} ~=~ {\textit b\_time}^{\ddot{s}(x)}(x,
T) - {\textit b\_time}^{\ddot{s}(x)}(y, T)$

$\hspace{41.5pt}~\le~ \dis_{x, y} \cdot \setuptime +
\tilde{w}^{\ddot{s}(x)}_{x, y} + {\textit
b\_time}^{\ddot{s}(x)}(y, \BB{y, x})
- (\setuptime+{\textit
b\_time}^{\ddot{s}(x)}(y, \BB{y, x}))$

$\hspace{41.5pt}~=~ \dis_{x,
y} \cdot \setuptime + \tilde{w}^{\ddot{s}(x)}_{x, y} - \setuptime$

$\hspace{41.5pt}~\le~ r^{\alpha_{x, y}}_{x,y}$ (according to (\ref{equation:regret-xy})).

\subsection{Proof of Fact~\ref{fact:critical-before}}
\label{subsection:fact4}
Notice that $\ddot{s}(x)$ and $y$ used in Fact~\ref{fact:critical-before} are inherited from Fact~\ref{fact:Tx-is-base-with-y-prime} and the notations $\mu$, $h$, $t^*$, $\gamma^j_{x,y}$,
$Q_j$, and $\tau_j$, which will be used below, were defined earlier in this section.
In order to prove Fact~\ref{fact:critical-before}, three lemmas are introduced first.

\begin{lemma} \label{lemma:critical-always-front}
$w^{\gamma_{x,y}^j}_{y,
q_{j, \tau_{j}}} + {\textit b\_time}^{\gamma_{x,y}^j}(q_{j,\tau_{j}},
\BB{q_{j,\tau_{j}}, y})$ is nondecreasing as $j$ increases from $0$ to $t^*$.
\end{lemma}
\begin{proof}
It is equivalent to show
\vspace{-5pt}
\begin{center}
$w^{\gamma_{x, y}^j}_{y, q_{j,\tau_j}} + {\textit
b\_time}^{\gamma_{x, y}^j}(q_{j,\tau_j}, \BB{q_{j,\tau_j}, y}) \le
w^{\gamma_{x, y}^{j+1}}_{y, q_{j+1, \tau_{j+1}}} + {\textit
b\_time}^{\gamma_{x, y}^{j+1}}(q_{j+1, \tau_{j+1}},
\BB{q_{j+1, \tau_{j+1}}, y})$
\end{center}
\vspace{-5pt}
for $0 \le j < t^*$. We first claim that $Q_j$ has
$q_{j,\tau_j} \not\in \{q_{j,1}, q_{j,2}, \ldots,
q_{j, j}\}$ under $\gamma_{x,y}^j$ for the following reasons.
Clearly, the claim is true for $j = 0$.
For $1 \le j < t^*$, it suffices
to show \linebreak ${\textit b\_time}^{\gamma_{x,
y}^{j}}(y, \BB{y, x}) > k \cdot \setuptime + w^{\gamma_{x, y}^{j}}_{y, q_{j,k}} +
{\textit b\_time}^{\gamma_{x, y}^{j}}(q_{j,k}, \BB{q_{j,k},
y})$ for all $1 \le k \le j$. As a consequence of $\gamma_{x,y}^{j}$ and
$\gamma_{x,y}^{k}$,
we have
${\textit b\_time}^{\gamma_{x,
y}^{j}}(y, \BB{y, x}) \ge {\textit b\_time}^{\gamma_{x,
y}^{k}}(y, \BB{y, x})$ and $
w^{\gamma_{x, y}^{j}}_{y, q_{j,k}} +
{\textit b\_time}^{\gamma_{x, y}^{j}}(q_{j,k}, \BB{q_{j,k},
y}) = w^{\gamma_{x, y}^{k}}_{y, q_{k,k}} +
{\textit b\_time}^{\gamma_{x, y}^{k}}(q_{k,k}, \BB{q_{k,k},
y})$. Besides, since $k < t^*$,
we have ${\textit b\_time}^{\gamma_{x,
y}^{k}}(y, \BB{y, x}) > k \cdot \setuptime + w^{\gamma_{x, y}^{k}}_{y, q_{k,k}} +
{\textit b\_time}^{\gamma_{x, y}^{k}}(q_{k,k}, \BB{q_{k,k},
y})$. Thus, the claim follows.

Suppose
$q_{j, \ell} = u_{j+1}$, where $j < \ell \le h$. We first assume $\tau_j <
\ell$. For $1 \le k \le h$ and $k \not= j+1$, since $w^{\gamma_{x, y}^{j}}_{y, q_{j+1,k}} = w^{\gamma_{x, y}^{j+1}}_{y, q_{j+1,k}}$
and ${\textit
b\_time}^{\gamma_{x, y}^{j}}(q_{j+1,k}, \BB{q_{j+1,k}, y})
= {\textit
b\_time}^{\gamma_{x, y}^{j+1}}(q_{j+1,k}, \BB{q_{j+1,k}, y})$, we have

$\hspace{58.5pt} k \cdot \setuptime +
w^{\gamma_{x, y}^{j+1}}_{y, q_{j+1,k}} + {\textit
b\_time}^{\gamma_{x, y}^{j+1}}(q_{j+1,k}, \BB{q_{j+1,k}, y})$

$\hspace{34.5pt}~=~ k \cdot \setuptime + w^{\gamma_{x, y}^{j}}_{y, q_{j+1,k}} + {\textit
b\_time}^{\gamma_{x, y}^{j}}(q_{j+1,k}, \BB{q_{j+1,k}, y})$


$\hspace{34.5pt}~\le~ \setuptime + \max\{(k-1) \cdot \setuptime + w^{\gamma_{x, y}^{j}}_{y, q_{j,k-1}} + {\textit
b\_time}^{\gamma_{x, y}^{j}}(q_{j,k-1}, \BB{q_{j,k-1}, y}),$

$\hspace{97pt}
(k \cdot \setuptime + w^{\gamma_{x, y}^{j}}_{y, q_{j,k}}  + {\textit
b\_time}^{\gamma_{x, y}^{j}}(q_{j,k}, \BB{q_{j,k}, y}))-\setuptime\}$ ~~ ($\because q_{j+1,k} \in
\{q_{j,k-1}, q_{j,k}\}$)


$\hspace{34.5pt}~\le~ \setuptime + (\tau_j \cdot \setuptime + w^{\gamma_{x,
y}^{j}}_{y, q_{j, \tau_j}} + {\textit b\_time}^{\gamma_{x,
y}^{j}}(q_{j, \tau_j}, \BB{q_{j, \tau_j}, y}))$

$\hspace{34.5pt}~=~ (\tau_j+1) \cdot \setuptime + w^{\gamma_{x,
y}^{j+1}}_{y, q_{j+1, \tau_j+1}} + {\textit b\_time}^{\gamma_{x,
y}^{j+1}}(q_{j+1, \tau_j+1}, \BB{q_{j+1, \tau_j+1}, y})$,
\medskip
\\where the last equality holds, because the claim implies $\tau_j +1 > j+ 1$
and hence $q_{j+1, \tau_j+1} = q_{j, \tau_j}$.

Consequently, we have $q_{j+1, \tau_{j+1}} \in \{q_{j+1, \tau_j+1},
q_{j+1, j+1}\}$.
If $q_{j+1, \tau_{j+1}} = q_{j+1, \tau_j+1}$ (= $q_{j,\tau_j}$), then
$w^{\gamma_{x, y}^j}_{y, q_{j,\tau_j}} + {\textit
b\_time}^{\gamma_{x, y}^j}(q_{j,\tau_j}, \BB{q_{j,\tau_j}, y}) =
w^{\gamma_{x, y}^{j+1}}_{y, q_{j+1, \tau_{j+1}}} + {\textit
b\_time}^{\gamma_{x, y}^{j+1}}(q_{j+1, \tau_{j+1}},
\BB{q_{j+1, \tau_{j+1}}, y})$. If $q_{j+1, \tau_{j+1}} =
q_{j+1, j+1}$, then

\medskip
$\hspace{61.5pt}w^{\gamma_{x, y}^j}_{y, q_{j,\tau_j}} + {\textit
b\_time}^{\gamma_{x, y}^j}(q_{j,\tau_j}, \BB{q_{j,\tau_j}, y})$

$\hspace{37.5pt}~=~ w^{\gamma_{x, y}^{j+1}}_{y, q_{j+1,\tau_j+1}} + {\textit
b\_time}^{\gamma_{x, y}^{j+1}}(q_{j+1,\tau_j+1}, \BB{q_{j+1,\tau_j+1}, y})$


$\hspace{37.5pt}~\le~ w^{\gamma_{x, y}^{j+1}}_{y, q_{j+1, \tau_{j+1}}} + {\textit
b\_time}^{\gamma_{x, y}^{j+1}}(q_{j+1, \tau_{j+1}},
\BB{q_{j+1, \tau_{j+1}}, y})$

\hspace{61.5pt}(recall that
$w^{\gamma_{x, y}^{j+1}}_{y,
q_{j+1,k}} + {\textit b\_time}^{\gamma_{x, y}^{j+1}}(q_{j+1,k},
\BB{q_{j+1,k}, y})$ is nonincreasing as

\hspace{61.5pt}~$k$ increases from $1$ to $h$).
\medskip

Next we assume $\tau_j \ge \ell$. Similarly, we have $q_{j+1, \tau_{j+1}} \in
\{q_{j+1,j+1}, q_{j+1,j+2},
\ldots, q_{j+1, \ell}\} \cup \setof{q_{j+1, \tau_j}}
$, i.e., $\tau_{j+1} \le \tau_j$. Since ${\textit b\_time}^{\gamma_{x,
y}^{j}}(y, \BB{y, x}) \le {\textit b\_time}^{\gamma_{x,
y}^{j+1}}(y, \BB{y, x})$, we have $w^{\gamma_{x, y}^j}_{y, q_{j,\tau_j}} + {\textit
b\_time}^{\gamma_{x, y}^j}(q_{j,\tau_j}, \BB{q_{j,\tau_j}, y}) \le
w^{\gamma_{x, y}^{j+1}}_{y, q_{j+1, \tau_{j+1}}} + {\textit
b\_time}^{\gamma_{x, y}^{j+1}}(q_{j+1, \tau_{j+1}},
\BB{q_{j+1, \tau_{j+1}}, y})$.
\end{proof}

\medskip

\begin{lemma} \label{lemma:critical-always-front-2}
If $w^{\ddot{s}(x)}_{y, \mu} + {\textit
b\_time}^{\ddot{s}(x)}(\mu, \OB{y, x}) < w^{\ddot{s}(x)}_{y,
q_{0, \tau_{0}}} + {\textit b\_time}^{\ddot{s}(x)}(q_{0,\tau_{0}},
\BB{q_{0,\tau_{0}}, y})$, then \linebreak ${\textit b\_time}^{\gamma_{x,
y}^{j}}(y, T) - {\textit b\_time}^{\gamma_{x, y}^{j}}(y,
\BB{y,x})$ is nonincreasing as $j$ increases from $0$
to $t^*$.
\end{lemma}
\begin{proof}
It is equivalent to show
\vspace{-5pt}
\begin{center}
${\textit b\_time}^{\gamma_{x,
y}^{j-1}}(y, T) - {\textit b\_time}^{\gamma_{x, y}^{j-1}}(y,
\BB{y,x}) \ge
{\textit b\_time}^{\gamma_{x,
y}^{j}}(y, T) - {\textit b\_time}^{\gamma_{x, y}^{j}}(y,
\BB{y,x})$
\end{center}
\vspace{-5pt}
for $1 \le j \le t^*$.
First, we define $\hat{Q}_j = (\hat{q}_{j,1}, \hat{q}_{j,2}, \ldots, \hat{q}_{j,h+1})$ as follows. Initially, let
$\hat{Q}_0$ be \linebreak obtained from $Q_0$ by inserting $\mu$ between $q_{0,p-1}$ and $q_{0,p}$, where $w^{\gamma_{x, y}^0}_{y,
q_{0,p-1}} + \linebreak {\textit b\_time}^{\gamma_{x, y}^0}(q_{0,p-1},
\BB{q_{0,p-1}, y}) \ge w^{\gamma_{x, y}^0}_{y,
\mu} + {\textit b\_time}^{\gamma_{x, y}^0}(\mu,
\OB{y,x}) > w^{\gamma_{x, y}^0}_{y,
q_{0,p}} + {\textit b\_time}^{\gamma_{x, y}^0}(q_{0,p},
\BB{q_{0,p}, y})$.

Then,
for $1 \le j \le t^*$,
$\hat{Q}_j$ is obtained from $\hat{Q}_{j-1}$ by cyclically shifting $\hat{q}_{j-1,c}, \hat{q}_{j-1,c+1}, \ldots, \linebreak \hat{q}_{j-1,\ell}$ one position toward the right, where $\hat{q}_{j-1, \ell} = u_j$ is assumed,
$c = j$ if $w^{\gamma_{x, y}^{j}}_{y,
\hat{q}_{{j-1},\ell}} + {\textit b\_time}^{\gamma_{x, y}^{j}}(\hat{q}_{{j-1},\ell},
\BB{\hat{q}_{{j-1},\ell}, y}) \ge w^{\gamma_{x, y}^j}_{y,
\mu} + {\textit b\_time}^{\gamma_{x, y}^j}(\mu,
\OB{y,x})$, and $c = j+1$ else. If we remove $\mu$ from $\hat{Q}_j$, then $Q_j$
results. Since $w^{\gamma_{x, y}^j}_{y,
\hat{q}_{j,k}} + {\textit b\_time}^{\gamma_{x, y}^j}(\hat{q}_{j,k},
\BB{\hat{q}_{j,k}, y})$ is
nonincreasing as $k$ increases from $1$ to $h+1$,
Lemma~\ref{lemma:optimal-sequence} guarantees that
$\hat{Q}_j$ is
an optimal sequence for $y$ to broadcast a message to $N_T(y)$ under
$\gamma_{x, y}^j$.

Notice that we have
\vspace{-8pt}
\begin{equation}
(\hat{q}_{j, 1},\hat{q}_{j, 2},\ldots,\hat{q}_{j, \tau_{j}}) = ({q}_{j, 1},{q}_{j, 2},\ldots, q_{j, \tau_{j}}), \vspace{-12pt}
\label{equation:Q-with-hat-Q}
\end{equation}
because

$\hspace{76.5pt}w^{\gamma^{j}_{x,y}}_{y, \mu}
+ {\textit b\_time}^{\gamma^{j}_{x,y}}(\mu, \OB{y, x})$

$\hspace{52.5pt}~=~ w^{\ddot{s}(x)}_{y, \mu}
+ {\textit b\_time}^{\ddot{s}(x)}(\mu, \OB{y, x})$ \hspace{10pt}(as a consequence of $\gamma_{x,y}^{j}$)


$\hspace{52.5pt}~<~ w^{\gamma^0_{x,y}}_{y,
q_{0, \tau_{0}}} + {\textit b\_time}^{\gamma^0_{x,y}}(q_{0,\tau_{0}},
\BB{q_{0,\tau_{0}}, y})$\hspace{10pt}($\because \ddot{s}(x) = \gamma^0_{x,y}$)
\vspace{-8pt}
\begin{equation}
\hspace{-32.5pt}~\le~ w^{\gamma^{j}_{x,y}}_{y, q_{j, \tau_{j}}} + {\textit
b\_time}^{\gamma^{j}_{x,y}}(q_{j, \tau_{j}}, \BB{q_{j, \tau_{j}}, y}) \mbox{\hspace{10pt}(by Lemma~\ref{lemma:critical-always-front}).}
\label{equation:tau-front}
\end{equation}

\vspace{-5pt}
As a consequence of (\ref{equation:Q-with-hat-Q}), we have

\medskip
$\hspace{40pt}{\textit b\_time}^{\gamma^{j}_{x,y}}(y, T)~\ge~ \tau_{j} \cdot \setuptime + w^{\gamma^{j}_{x,y}}_{y, \hat{q}_{j, \tau_{j}}} + {\textit
b\_time}^{\gamma^{j}_{x,y}}(\hat{q}_{j, \tau_{j}}, \BB{\hat{q}_{j, \tau_{j}}, y})$


$\hspace{118.5pt}~>~ k \cdot \setuptime + w^{\gamma^{j}_{x,y}}_{y, \hat{q}_{j, k}} + {\textit
b\_time}^{\gamma^{j}_{x,y}}(\hat{q}_{j, k}, \BB{\hat{q}_{j, k}, y})$
\medskip
\\for $1 \le k < \tau_{j}$.
Let $\hat{q}_{j, \varrho_j}$ be the first vertex in
$\hat{Q}_j$ with ${\textit
b\_time}^{\gamma_{x, y}^j}(y, T) =  \varrho_j \cdot \setuptime + w^{\gamma_{x,
y}^j}_{y, \hat{q}_{j, \varrho_j}} + {\textit b\_time}^{\gamma_{x,
y}^j}(\hat{q}_{j, \varrho_j}, \BB{\hat{q}_{j, \varrho_j}, y})$. Clearly, we have $\varrho_{j} \ge \tau_{j}$. When $\varrho_{j} = \tau_{j}$, the
lemma follows, because ${\textit b\_time}^{\gamma_{x, y}^{j}}(y, T) =
{\textit b\_time}^{\gamma_{x, y}^{j}}(y, \BB{y, x})$ and ${\textit b\_time}^{\gamma_{x, y}^{j-1}}(y, T) \ge
{\textit b\_time}^{\gamma_{x, y}^{j-1}}(y, \BB{y, x})$.
In subsequent discussion, we assume
$\varrho_{j} > \tau_{j}$.

Suppose $q_{j-1,{\ell'}} = u_{j}$, where
$j \le {\ell'} \le h$.
We first consider $\tau_{j-1} <
{\ell'}$.
Since $q_{{j-1}, \tau_{j-1}} = q_{{j}, \tau_{j-1}+1}$,
we have
\medskip

$\hspace{10pt}{\textit b\_time}^{\gamma_{x,
y}^{j}}(y, \BB{y, x})~\ge~ (\tau_{j-1}+1) \cdot \setuptime + w^{\gamma_{x,
y}^{j}}_{y, q_{j, \tau_{j-1}+1}} + {\textit b\_time}^{\gamma_{x,
y}^{j}}(q_{j, \tau_{j-1}+1}, \BB{q_{j, \tau_{j-1}+1}, y})$

$\hspace{100.8pt}~=~ \setuptime + (\tau_{j-1} \cdot \setuptime + w^{\gamma_{x,
y}^{{j}}}_{y, q_{{j-1}, \tau_{j-1}}} + {\textit b\_time}^{\gamma_{x,
y}^{{j}}}(q_{{j-1}, \tau_{j-1}}, \BB{q_{{j-1}, \tau_{j-1}}, y}))$

$\hspace{100.8pt}~=~ \setuptime + (\tau_{j-1} \cdot \setuptime + w^{\gamma_{x,
y}^{{j-1}}}_{y, q_{{j-1}, \tau_{j-1}}} + {\textit b\_time}^{\gamma_{x,
y}^{{j-1}}}(q_{{j-1}, \tau_{j-1}}, \BB{q_{{j-1}, \tau_{j-1}}, y}))$
\vspace{-8pt}
\begin{equation}
\hspace{-75.5pt}~=~ \setuptime + {\textit b\_time}^{\gamma_{x,
y}^{{j-1}}}(y, \BB{y, x}).
\label{equation:Tx-time}
\end{equation}

\vspace{-5pt}
Since $\varrho_{j} > \tau_{j} \ge j$, we have $\hat{q}_{j, \varrho_{j}} \not= q_{j, j} = u_{j}$ (as a consequence of (\ref{equation:Q-with-hat-Q})),
which implies $\hat{q}_{j, \varrho_{j}} \in \{\hat{q}_{j-1, \varrho_{j}-1}, \hat{q}_{j-1, \varrho_{j}}\}$.
Hence, we have

\medskip
$\hspace{0pt}{\textit b\_time}^{\gamma_{x,
y}^{j}}(y, T)~=~ \varrho_{j} \cdot \setuptime + w^{\gamma_{x,
y}^{j}}_{y, \hat{q}_{j, \varrho_{j}}} + {\textit b\_time}^{\gamma_{x,
y}^{j}}(\hat{q}_{j, \varrho_{j}}, \BB{\hat{q}_{j, \varrho_{j}}, y})$


$\hspace{78.75pt}~=~ \setuptime + \max\{({\varrho_{j}}-1) \cdot \setuptime + w^{\gamma_{x, y}^{j}}_{y, \hat{q}_{j-1,{\varrho_{j}}-1}} + {\textit
b\_time}^{\gamma_{x, y}^{j}}(\hat{q}_{j-1,{\varrho_{j}}-1}, \BB{\hat{q}_{j-1,{\varrho_{j}}-1}, y}),$

$\hspace{147.5pt}
({\varrho_{j}} \cdot \setuptime + w^{\gamma_{x, y}^{j}}_{y, \hat{q}_{j-1,{\varrho_{j}}}}  + {\textit
b\_time}^{\gamma_{x, y}^{j}}(\hat{q}_{j-1,{\varrho_{j}}}, \BB{\hat{q}_{j-1,{\varrho_{j}}}, y}))-\setuptime\}$

\vspace{6pt}
$\hspace{78.75pt}~=~ \setuptime + \max\{({\varrho_{j}}-1) \cdot \setuptime + w^{\gamma_{x, y}^{j-1}}_{y, \hat{q}_{j-1,{\varrho_{j}}-1}} + {\textit
b\_time}^{\gamma_{x, y}^{j-1}}(\hat{q}_{j-1,{\varrho_{j}}-1}, \BB{\hat{q}_{j-1,{\varrho_{j}}-1}, y}),$

$\hspace{147.5pt}
({\varrho_{j} \cdot \setuptime} + w^{\gamma_{x, y}^{j-1}}_{y, \hat{q}_{j-1,{\varrho_{j}}}}  + {\textit
b\_time}^{\gamma_{x, y}^{j-1}}(\hat{q}_{j-1,{\varrho_{j}}}, \BB{\hat{q}_{j-1,{\varrho_{j}}}, y}))-\setuptime\}$
\vspace{-8pt}
\begin{equation}
\hspace{-131pt}~\le~ \setuptime + {\textit b\_time}^{\gamma_{x,
y}^{j-1}}(y, T).
\label{equation:T-time}
\end{equation}
Thus, the lemma follows, as a consequence of (\ref{equation:Tx-time}) and (\ref{equation:T-time}).

Next we consider
$\tau_{j-1} \ge {\ell'}$. Suppose $\hat{q}_{j, \eta} = \mu$.
We have $\varrho_{j} \ge \eta$, because
for $1 \le k < \eta$,

\bigskip
$\hspace{111.5pt}  k \cdot \setuptime + w^{\gamma^{j}_{x,y}}_{y, \hat{q}_{j, k}}
+ {\textit b\_time}^{\gamma^{j}_{x,y}}(\hat{q}_{j, k}, \BB{\hat{q}_{j, k}, y})$

$\hspace{87.5pt}~=~ k \cdot \setuptime + w^{\gamma^{j}_{x,y}}_{y, q_{j, k}}
+ {\textit b\_time}^{\gamma^{j}_{x,y}}({q}_{j, k}, \BB{{q}_{j, k}, y})$


$\hspace{87.5pt}~\le~ \tau_{j} \cdot \setuptime + w^{\gamma^{j}_{x,y}}_{y, q_{j, \tau_{j}}}
+ {\textit b\_time}^{\gamma^{j}_{x,y}}({q}_{j, \tau_{j}}, \BB{{q}_{j, \tau_{j}}, y})$

$\hspace{87.5pt}~=~ \tau_{j} \cdot \setuptime + w^{\gamma^{j}_{x,y}}_{y, \hat{q}_{j, \tau_{j}}}
+ {\textit b\_time}^{\gamma^{j}_{x,y}}(\hat{q}_{j, \tau_{j}}, \BB{\hat{q}_{j, \tau_{j}}, y})$

$\hspace{87.5pt}~<~ \varrho_{j} \cdot \setuptime + w^{\gamma^{j}_{x,y}}_{y, \hat{q}_{j, \varrho_{j}}}
+ {\textit b\_time}^{\gamma^{j}_{x,y}}({\hat{q}}_{j, \varrho_{j}}, \BB{{\hat{q}}_{j, \varrho_{j}}, y})$.
\bigskip
\\Besides, $u_{j}$ (= $q_{j-1,{\ell'}}$) precedes $\mu$
in $\hat{Q}_{j-1}$, because
$w^{\gamma^{j-1}_{x,y}}_{y, \mu} +
{\textit b\_time}^{\gamma^{j-1}_{x,y}}(\mu, \OB{y, x}) <
w^{\gamma^{j-1}_{x,y}}_{y, q_{{j-1}, \tau_{j-1}}} + {\textit
b\_time}^{\gamma^{j-1}_{x,y}}(q_{{j-1}, \tau_{j-1}}, \BB{q_{{j-1}, \tau_{j-1}}, y})
\le w^{\gamma^{j-1}_{x,y}}_{y, q_{{j-1}, {\ell'}}} + {\textit
b\_time}^{\gamma^{j-1}_{x,y}}(q_{{j-1}, {\ell'}}, \BB{q_{{j-1}, {\ell'}}, y})$,
where the relation $<$ holds with
arguments similar to (\ref{equation:tau-front}).
So, we have $\hat{q}_{j, \varrho_{j}} = \hat{q}_{j-1,\varrho_{j}}$,
and hence

\bigskip
$\hspace{20pt}{\textit b\_time}^{\gamma_{x,
y}^{j}}(y, T) ~=~ \varrho_{j} \cdot \setuptime + w^{\gamma^{j}_{x,y}}_{y, \hat{q}_{j, \varrho_{j}}} +
{\textit b\_time}^{\gamma^{j}_{x,y}}(\hat{q}_{j, \varrho_{j}}, \BB{\hat{q}_{j, \varrho_{j}}, y})$

$\hspace{98.75pt}~=~\varrho_{j} \cdot \setuptime + w^{\gamma^{j}_{x,y}}_{y, \hat{q}_{j-1, \varrho_{j}}} +
{\textit b\_time}^{\gamma^{j}_{x,y}}(\hat{q}_{j-1, \varrho_{j}}, \BB{\hat{q}_{j-1, \varrho_{j}}, y})$

$\hspace{98.75pt}~=~ \varrho_{j} \cdot \setuptime + w^{\gamma^{j-1}_{x,y}}_{y, \hat{q}_{j-1, \varrho_{j}}} +
{\textit b\_time}^{\gamma^{j-1}_{x,y}}(\hat{q}_{j-1, \varrho_{j}}, \BB{\hat{q}_{j-1, \varrho_{j}}, y})$

$\hspace{98.75pt}~\le~ {\textit b\_time}^{\gamma_{x,
y}^{j-1}}(y, T)$.
\medskip

On the other hand, since
${\textit b\_time}^{\gamma_{x,
y}^{j}}(y, T) \ge {\textit b\_time}^{\gamma_{x, y}^{j-1}}(y, T)$,
we have ${\textit b\_time}^{\gamma_{x,
y}^{j}}(y, T) = {\textit b\_time}^{\gamma_{x, y}^{j-1}}(y, T)$.
The lemma then
follows, because ${\textit b\_time}^{\gamma_{x, y}^{j}}(y,
\BB{y, x}) \ge {\textit b\_time}^{\gamma_{x, y}^{j-1}}(y, \BB{y,
x})$.

\end{proof}
%
%
\begin{lemma} \label{lemma:worst-case-candidate}
If
$w^{\ddot{s}(x)}_{y, \mu} + {\textit
b\_time}^{\ddot{s}(x)}(\mu, \OB{y, x}) < w^{\ddot{s}(x)}_{y,
q_{0, \tau_{0}}} + {\textit b\_time}^{\ddot{s}(x)}(q_{0,\tau_{0}},
\BB{q_{0,\tau_{0}}, y})$, then for $0 \le j \le t^*$,
$\gamma_{x, y}^{j}$ is a worst-case
scenario of $T$ with respect to $x$. Besides, we have ${\textit
max\_r}(x) =  \dis_{x, y} \cdot \setuptime + \tilde{w}^{\gamma_{x, y}^{j}}_{x, y} +
{\textit b\_time}^{\gamma_{x, y}^{j}}(y, \BB{y, x}) - {\textit
b\_time}^{\gamma_{x, y}^{j}}(y, T)$.
\end{lemma}
\begin{proof}
This lemma can be proved by induction on $j$.
Recall that we have $\gamma_{x, y}^{0} = \ddot{s}(x)$, where
$y$ is a prime
broadcast center of $T$ under $\ddot{s}(x)$.
The lemma holds for $j=0$,
as a consequence of
Lemma~\ref{lemma:Direct-to-center}.
Suppose that the lemma
holds for $j=k \ge 0$. When $j = k+1$, we have

$\hspace{20pt}r^{\gamma_{x, y}^{k+1}}_{x, y}  ~=~ {\textit b\_time}^{\gamma_{x, y}^{k+1}}(x, T) - {\textit b\_time}^{\gamma_{x, y}^{k+1}}(y, T)$

$\hspace{44.5pt}~\ge~ \dis_{x, y} \cdot \setuptime + \tilde{w}^{\gamma_{x, y}^{k+1}}_{x, y} + {\textit b\_time}^{\gamma_{x, y}^{k+1}}(y, \BB{y,x}) - {\textit b\_time}^{\gamma_{x, y}^{k+1}}(y, T)$ ~~(by Lemma~\ref{lemma:general-x-y})



$\hspace{44.5pt}~\ge~ \dis_{x, y} \cdot \setuptime + \tilde{w}^{\gamma_{x, y}^{k+1}}_{x, y} + {\textit b\_time}^{\gamma_{x, y}^{k}}(y, \BB{y,x}) - {\textit b\_time}^{\gamma_{x, y}^{k}}(y, T)$ ~~(by
Lemma~\ref{lemma:critical-always-front-2})

$\hspace{44.5pt}~=~ \dis_{x, y} \cdot \setuptime + \tilde{w}^{\gamma_{x, y}^{k}}_{x, y} + {\textit b\_time}^{\gamma_{x, y}^{k}}(y, \BB{y,x}) - {\textit b\_time}^{\gamma_{x, y}^{k}}(y, T)$

$\hspace{44.5pt}~=~ {\textit max\_r}(x)$ ~~(by the induction hypothesis).
\end{proof}

\bigskip

Now it is ready to prove Fact~\ref{fact:critical-before}.
First observe that Fact~\ref{fact:critical-before} and Lemma~\ref{lemma:worst-case-candidate} have the same sufficient condition, but different necessary conditions.
In the following, we show that the necessary condition (i.e., $\beta^{t^*}_{x,y}$ is a worst-case scenario) of Fact~\ref{fact:critical-before} is derivable from the necessary conditions (i.e., $\gamma^{t^*}_{x,y}$ is a worst-case scenario and ${\textit
max\_r}(x) = \dis_{x, y} \cdot \setuptime + \tilde{w}^{\gamma_{x, y}^{t^*}}_{x, y} +
{\textit b\_time}^{\gamma_{x, y}^{t^*}}(y, \BB{y, x}) - {\textit
b\_time}^{\gamma_{x, y}^{t^*}}(y, T)$) of Lemma~\ref{lemma:worst-case-candidate}. It suffices to show
$r^{\beta^{t^*}_{x,y}}_{x,y} \ge {\textit max\_r}(x)$, in order to show
$\beta^{t^*}_{x,y}$ a worst-case scenario.

Let $Q'' = ({q''_1,q''_2, \ldots, q''_h})$ be an arrangement of
$u_1, u_2, \ldots, u_h$, which has
$q''_k = u_k$ ($=q_{t^*, k}$) for $1 \le k \le t^*$ and
$w^{\beta_{x, y}^{t^*}}_{y,
q''_{k}} + {\textit b\_time}^{\beta_{x, y}^{t^*}}(q''_{k},
\BB{q''_{k}, y})$ nonincreasing as $k$ increases from $t^* + 1$ to
$h$. Since $\gamma_{x,y}^{t^*}$ and $\beta_{x,y}^{t^*}$
differ in the weights of edges in $\bigcup_{{t^*} < k \le h}(\setof{(y, u_{k})} + E(\BB{u_{k}, y}))$,
$w^{\beta_{x, y}^{t^*}}_{y,
q''_{k}} + {\textit b\_time}^{\beta_{x, y}^{t^*}}(q''_{k},
\BB{q''_{k}, y})$ is also nonincreasing as $k$ increases from $1$ to
$t^*$ (refer to Figure~\ref{figure:gamma_scenario}).
Actually, it is nonincreasing as $k$ increases from $1$ to
$h$, because $w^{\beta_{x, y}^{t^*}}_{y,
q''_{t^*}} + {\textit b\_time}^{\beta_{x, y}^{t^*}}(q''_{t^*},
\BB{q''_{t^*}, y}) \ge w^{\beta_{x, y}^{t^*}}_{y,
q''_{t^*+1}} + {\textit b\_time}^{\beta_{x, y}^{t^*}}(q''_{t^*+1},
\BB{q''_{t^*+1}, y})$, as explained below.

Suppose $q''_{t^*+1} = q_{t^*, p}$, where $t^* < p \le h$. We have

\medskip
$\hspace{76.5pt}w^{\beta_{x, y}^{t^*}}_{y,
q''_{t^*+1}} + {\textit b\_time}^{\beta_{x, y}^{t^*}}(q''_{t^*+1},
\BB{q''_{t^*+1}, y})$

$\hspace{52.5pt}~\le~ w^{\gamma_{x, y}^{t^*}}_{y,
q_{t^*, p}} + {\textit b\_time}^{\gamma_{x, y}^{t^*}}(q_{t^*, p},
\BB{q_{t^*, p}, y})$ \hspace{10pt}(refer to Figure~\ref{figure:gamma_scenario})


$\hspace{52.5pt}~\le~ w^{\gamma_{x, y}^{t^*}}_{y,
q_{t^*, t^*}} + {\textit b\_time}^{\gamma_{x, y}^{t^*}}(q_{t^*, t^*},
\BB{q_{t^*, t^*}, y})$

$\hspace{52.5pt}~=~ w^{\beta_{x, y}^{t^*}}_{y,
q''_{t^*}} + {\textit b\_time}^{\beta_{x, y}^{t^*}}(q''_{t^*},
\BB{q''_{t^*}, y})$.
\medskip

Then, according to Lemma~\ref{lemma:optimal-sequence},
$Q''$ is an optimal sequence for $y$ to broadcast a message to
$N_{\BB{y, x}}(y)$ under $\beta_{x, y}^{t^*}$.
So, we have

\medskip
$\hspace{20pt}{\textit b\_time}^{\beta^{t^*}_{x, y}}(y, \BB{y, x})  ~\ge~  t^* \cdot \setuptime + w^{\beta^{t^*}_{x, y}}_{y, q''_{t^*}} + {\textit b\_time}^{\beta^{t^*}_{x, y}}(q''_{t^*}, \BB{q''_{t^*}, y})$



$\hspace{111pt}~=~ t^* \cdot \setuptime + w^{\gamma_{x,y}^{t^*}}_{y, q_{t^*,t^*}} + {\textit b\_time}^{\gamma_{x,y}^{t^*}}(q_{t^*,t^*}, \BB{q_{t^*,t^*}, y})$

$\hspace{111pt}~=~ {\textit b\_time}^{\gamma_{x,y}^{t^*}}(y, \BB{y,x})$.
\medskip

On the other hand, since ${\textit
b\_time}^{\beta^{t^*}_{x, y}}(y, \BB{y, x}) \le
{\textit b\_time}^{\gamma_{x,y}^{t^*}}(y, \BB{y,
x})$, we have \linebreak
${\textit b\_time}^{\beta^{t^*}_{x, y}}(y,
\BB{y, x}) = {\textit
b\_time}^{\gamma_{x,y}^{t^*}}(y, \BB{y, x})$.
Further, since $\tilde{w}^{\beta^{t^*}_{x, y}}_{x, y} = \tilde{w}^{\gamma_{x,y}^{t^*}}_{x, y}$ and ${\textit
b\_time}^{\beta^{t^*}_{x, y}}(y, T) \le
{\textit b\_time}^{\gamma_{x,y}^{t^*}}(y, T)$
(refer to Figure~\ref{figure:gamma_scenario}), we have

\medskip
$\hspace{20pt}r^{\beta^{t^*}_{x, y}}_{x,y} ~=~ {\textit b\_time}^{\beta^{t^*}_{x, y}}(x, T) - {\textit b\_time}^{\beta^{t^*}_{x, y}}(y, T)$

$\hspace{41.5pt}~\ge~ \dis_{x, y} \cdot \setuptime + \tilde{w}^{\beta^{t^*}_{x, y}}_{x, y} + {\textit b\_time}^{\beta^{t^*}_{x, y}}(y, \BB{y, x}) - {\textit b\_time}^{\beta^{t^*}_{x, y}}(y, T)$ ~~(by Lemma~\ref{lemma:general-x-y})

$\hspace{41.5pt}~=~ \dis_{x, y} \cdot \setuptime + \tilde{w}^{\gamma_{x,y}^{t^*}}_{x, y} + {\textit b\_time}^{\gamma_{x,y}^{t^*}}(y, \BB{y, x}) - {\textit b\_time}^{\beta_{x,y}^{t^*}}(y, T)$

$\hspace{41.5pt}~\ge~ \dis_{x, y} \cdot \setuptime + \tilde{w}^{\gamma_{x,y}^{t^*}}_{x, y} + {\textit b\_time}^{\gamma_{x,y}^{t^*}}(y, \BB{y, x}) - {\textit b\_time}^{\gamma_{x,y}^{t^*}}(y, T)$
\vspace{-8pt}
\begin{equation}
\hspace{-173pt}~=~ {\textit max\_r}(x) \textrm{~~(by Lemma~\ref{lemma:worst-case-candidate}).}
\label{equation:fact4}
\end{equation}


\section{Time complexity}
\label{section:time_complexity}
It is known that given a scenario $s \in C$ and a vertex $v\in V(T)$, ${\textit B\_Ctr}^{s}$ and ${\textit b\_time}^{s}(v, T)$ can be determined in $O(n)$ time ~\cite{Su2016}. Since there are at most $n-1$ scenarios in
$\Psi_1(x)\cup\Psi_2(x)\cup\ldots\cup\Psi_n(x)$, it takes $O(n^2)$ time for Algorithm~\ref{algorithm:worst-case} to find $\ddot{s}(x)$. In this section, we show that the time complexity of Algorithm~\ref{algorithm:worst-case} can be reduced to $O(n\log\log n)$, provided an $O(n \log n)$ time preprocessing is made.
Consequently, Algorithm~\ref{algorithm:mrbc} requires $O(n \log n) + O(n\log\log n) \cdot O(\log n) = O(n \log n \log\log n)$ time. We assume below that $x$ is not a prime broadcast center of $T$ under $\ddot{s}(x)$, i.e., $\ddot{s}(x) \in
\Psi_1(x)\cup\Psi_2(x)\cup\ldots\cup\Psi_{n-1}(x)$ (refer to the first paragraph of Section~\ref{section:correctness}). The following lemma is useful to find $\ddot{s}(x)$.

\begin{lemma} \label{lemma:transform-problem-to-broadcasting-time}
Suppose $x \in V(T)$. If $\dis_{x,y} \cdot \setuptime + \tilde{w}^{\beta^{j}_{x, y}}_{x,y} +
{\textit b\_time}^{\beta^{j}_{x, y}}(y, \BB{y,x}) - {\textit
b\_time}^{\beta^{j}_{x, y}}(y, T)$ has maximum value as $y = y'$ and $j = j'$, then we have
$\beta^{j'}_{x, y'} = \ddot{s}(x)$.
\end{lemma}
\begin{proof}
The notations $\mu, q_{0, \tau_0}, h$,
and $t^*$ used in this proof inherit from
Fact~\ref{fact:critical-after} and
Fact~\ref{fact:critical-before}. First we claim that
there exists $\beta^{j}_{x, y}$, satisfying $\beta^{j}_{x, y} = \ddot{s}(x)$ and
${\textit max\_r}(x) = \dis_{x, y} \cdot \setuptime + \tilde{w}^{\beta^{j}_{x, y}}_{x, y} +
{\textit b\_time}^{\beta^{j}_{x, y}}(y, \BB{y, x}) - {\textit b\_time}^{\beta^{j}_{x, y}}(y, T)$, as explained below. If $w^{\ddot{s}(x)}_{y, \mu} +
{\textit b\_time}^{\ddot{s}(x)}(\mu, \OB{y, x}) \ge
w^{\ddot{s}(x)}_{y, q_{0, \tau_0}} + {\textit
b\_time}^{\ddot{s}(x)}(q_{0, \tau_0}, \BB{q_{0, \tau_0},
y})$,
then
we have
$r^{\alpha_{x, y}}_{x,y} \ge \dis_{x, y} \cdot \setuptime + \tilde{w}^{\alpha_{x, y}}_{x, y} + {\textit b\_time}^{\alpha_{x, y}}(y, \BB{y, x}) - {\textit b\_time}^{\alpha_{x, y}}(y, T) \ge r^{\ddot{s}(x)}_{x,y}$,
as a consequence of (\ref{equation:fact3-case1}) and (\ref{equation:regret-xy}).
Also, since $r^{\ddot{s}(x)}_{x,y} = {\textit max\_r}(x) \ge r^{\alpha_{x, y}}_{x,y}$,
we have ${\textit max\_r}(x) = r^{\alpha_{x, y}}_{x,y} = \dis_{x, y} \cdot \setuptime + \tilde{w}^{\alpha_{x, y}}_{x, y} + {\textit b\_time}^{\alpha_{x, y}}(y, \BB{y, x}) - {\textit b\_time}^{\alpha_{x, y}}(y, T)$ (recall that $\alpha_{x, y} = \beta^h_{x,y}$). If $w^{\ddot{s}(x)}_{y, \mu} +
{\textit b\_time}^{\ddot{s}(x)}(\mu, \OB{y, x}) <
w^{\ddot{s}(x)}_{y, q_{0, \tau_0}} + {\textit
b\_time}^{\ddot{s}(x)}(q_{0, \tau_0}, \BB{q_{0, \tau_0},
y})$, then
we have ${\textit max\_r}(x) = r^{\beta^{t^*}_{x, y}}_{x,y} =
\dis_{x, y} \cdot \setuptime + \tilde{w}^{\beta^{t^*}_{x, y}}_{x, y} +
{\textit b\_time}^{\beta^{t^*}_{x, y}}(y, \BB{y, x}) - {\textit b\_time}^{\beta^{t^*}_{x, y}}(y, T)$, similarly,
as a consequence of (\ref{equation:fact4}).

It follows that we have

\medskip
$\hspace{20pt}r^{\beta^{j'}_{x, y'}}_{x,y'} ~=~ {\textit b\_time}^{\beta^{j'}_{x, y'}}(x, T) - {\textit b\_time}^{\beta^{j'}_{x, y'}}(y', T)$

$\hspace{44pt}~\ge~ \dis_{x, y'} \cdot \setuptime + \tilde{w}^{\beta^{j'}_{x, y'}}_{x, y'} + {\textit b\_time}^{\beta^{j'}_{x, y'}}(y', \BB{y', x}) - {\textit b\_time}^{\beta^{j'}_{x, y'}}(y', T)$ ~~(by Lemma~\ref{lemma:general-x-y})

$\hspace{44pt}~\ge~ \dis_{x, y} \cdot \setuptime + \tilde{w}^{\beta_{x,y}^{j}}_{x, y} + {\textit b\_time}^{\beta_{x,y}^{j}}(y, \BB{y, x}) - {\textit b\_time}^{\beta_{x,y}^{j}}(y, T)$

$\hspace{44pt}~=~ {\textit max\_r}(x)$,

\noindent which implies $\beta^{j'}_{x, y'} =
\ddot{s}(x)$.
\end{proof}

\medskip

With Lemma~\ref{lemma:transform-problem-to-broadcasting-time}, it suffices
to find a scenario $\beta^{j}_{x, v_i}$ that can maximize the
value of $\dis_{x,v_i} \cdot \setuptime + \tilde{w}^{\beta^{j}_{x, v_i}}_{x,v_i} +
{\textit b\_time}^{\beta^{j}_{x, v_i}}(v_i, \BB{v_i,x}) - {\textit
b\_time}^{\beta^{j}_{x, v_i}}(v_i, T)$, in order to determine $\ddot{s}(x)$. By a depth-first
traversal of $T$, $\dis_{x,v_i}$ and $\tilde{w}^{\alpha_{x,
v_i}}_{x,v_i}$ (=$\tilde{w}^{\beta^{j}_{x, v_i}}_{x,v_i}$ for all $1 \le j \le h_i$) for all $1 \le i < n$ can be determined in $O(n)$ time. On the other hand, as shown below, ${\textit
b\_time}^{\beta^{j}_{x, v_i}}(v_i, \BB{v_i,x})$ (and ${\textit
b\_time}^{\beta^{j}_{x, v_i}}(v_i, T)$ similarly) for all $1 \le i < n$ and $1 \le j \le h_i$ can be determined in $O(n \log\log n)$ time. Therefore, $\ddot{s}(x)$ can be determined in $O(n \log\log n)$ time.

We suppose $N_{\BB{v_i, x}}(v_i) = \{u_{i,1}, u_{i,2}, \ldots, u_{i,h_i}\}$ in the rest of this section.
The following fact will be shown in Section~\ref{subsection:preprocessing}.

\begin{fact} \label{fact:preprocessing-broadcast-time}
With an $O(n\log n)$-time preprocessing,
${\textit b\_time}^{\beta^{j}_{x,
v_i}}(u_{i,k}, \BB{u_{i,k}, v_i})$ can be determined in constant time, where
$1 \le j \le h_i$ and $1 \le k \le h_i$.
\end{fact}

With Fact~\ref{fact:preprocessing-broadcast-time}, ${\textit b\_time}^{\beta^{j}_{x,
v_i}}(v_i, \BB{v_i,x})$ can be determined in additional $O(h_i)$ time, as described below.
Suppose $s \in C$, and let $\Buck[1]$, $\Buck[2]$, \ldots, $\Buck[h_i]$ be $h_i$ buckets. We assume, without loss of generality, that $u_{i,1}$ can maximize the value of $w^{s}_{v_i, u_{i,k}} + {\textit
b\_time}^{s}(u_{i,k}, \BB{u_{i,k}, v_i})$. Like~\cite{Su2016}, we insert $u_{i, k}$ into $\Buck[\ell]$, if $\ell - 1 \le (w^{s}_{v_i, u_{i,1}} +
{\textit
b\_time}^{s}(u_{i,1}, \BB{u_{i,1}, v_i})) - (w^{s}_{v_i, u_{i,k}}
+ {\textit b\_time}^{s}(u_{i,k}, \BB{u_{i,k}, v_i})) < \ell$, where $1 \le k \le h_i$ and $1 \le \ell \le h_i$.
Define $\Acc(\ell)
= \sum_{1 \le r \le \ell}|\Buck[r]|$, which is the total number of vertices contained in $\Buck[1],  \Buck[2], \ldots, \Buck[\ell]$, and $\Minv(\ell) =\min\{w^{s}_{v_i, u_{i,k}} + {\textit b\_time}^{s}(u_{i,k},
\BB{u_{i,k}, v_i}) \mid u_{i, k} \in \Buck[\ell]\}$.
Previously, $\Acc(\ell)$ and $\Minv(\ell)$ were successfully
used in \cite{Su2016} to
reduce the time requirement for determining
an optimal broadcast sequence.

Suppose $u_{i,k}\in\Buck[\ell]$,
satisfying $\Minv(\ell) =w^{s}_{v_i, u_{i,k}} + {\textit b\_time}^{s}(u_{i,k},
\BB{u_{i,k}, v_i})$, where $1 \le \ell \le h_i$.
Then, as a consequence of Lemma~\ref{lemma:optimal-sequence},
$u_{i,k}$ is the $\Acc(\ell)$-th vertex in some optimal sequence for $v_i$ to broadcast a message to $N_{\BB{v_i,x}}(v_i)$.
Also, for any other $u_{i,k'} \in \Buck[\ell]$,
we have $u_{i,k'}$ preceding $u_{i,k}$
in the optimal sequence, implying that
the broadcast time (i.e., $\Acc(\ell) \cdot \setuptime + w^{s}_{v_i, u_{i,k}} + {\textit b\_time}^{s}(u_{i,k},
\BB{u_{i,k}, v_i})$) induced by $\BB{u_{i,k}, v_i}$ is at least as much as the broadcast time
induced by $\BB{u_{i,k'}, v_i}$.
Thus, we have ${\textit b\_time}^{s}(v_i, \BB{v_i,x})=\max\{\Minv(\ell) +
\Acc(\ell) \cdot \setuptime \mid 1 \le \ell \le h_i\}$, which can be computed
in $O(h_i)$ time.

It follows that ${\textit b\_time}^{\beta^{j}_{x, v_i}}(v_i,
\BB{v_i,x})$ for all $1\le j \le h_i$ can be determined in additional $O(h_i^2)$ time. As elaborated below, the time complexity can be reduced to $O(h_i \log\log h_i)$ (hence, it takes $O(n \log\log n)$ time to determine ${\textit b\_time}^{\beta^{j}_{x, v_i}}(v_i,
\BB{v_i,x})$ for all $1 \le i < n$ and $1\le j \le h_i$).
The following fact will be shown in Section~\ref{subsection:preprocessing} as well.

\begin{fact}
\label{fact:preprocessing-increasing-order}
With an $O(n\log n)$-time preprocessing, a vertex ordering $(u_{i,1}, u_{i,2}, \ldots, u_{i,h_i})$ of $N_{\BB{v_i, x}}(v_i)$ can be determined in $O(h_i)$ time such that
$w^{\beta^{h_i}_{x, v_i}}_{v_i, u_{i,k}} + {\textit
b\_time}^{\beta^{h_i}_{x, v_i}}(u_{i,k}, \BB{u_{i,k}, v_i})$ is nonincreasing as $k$ increases from $1$ to $h_i$.
\end{fact}

With Fact~\ref{fact:preprocessing-increasing-order},
$(w^{\beta^{h_i}_{x, v_i}}_{v_i, u_{i,1}}
+ {\textit b\_time}^{\beta^{h_i}_{x, v_i}}(u_{i,1}, \BB{u_{i,1}, v_i})) - (w^{\beta^{h_i}_{x, v_i}}_{v_i, u_{i,k}}
+ {\textit b\_time}^{\beta^{h_i}_{x, v_i}}(u_{i,k}, \BB{u_{i,k}, v_i}))$
is nondecreasing as $k$ increases from $1$ to $h_i$. Under $\beta^{h_i}_{x, v_i}$, we have $u_{i, 1} \in \Buck[1]$.\linebreak Moreover, there exists $1 \le h_i' \le h_i$ such that $u_{i,1}, u_{i, 2}, \ldots, u_{i, h_i'}$ are contained in\linebreak buckets, while $u_{i,h_i'+1}, u_{i, h_i'+2}, \ldots, u_{i, h_i}$ are not contained
in buckets. Recall that
when\linebreak $h_i' + 1 \le j \le h_i$ and $1 \le k \le h_i'$, we have
$w^{\beta^{h_i'}_{x, v_i}}_{v_i, u_{i,k}} + {\textit b\_time}^{\beta^{h_i'}_{x, v_i}}(u_{i,k},
\BB{u_{i,k}, v_i}) = w^{\beta^{j}_{x, v_i}}_{v_i, u_{i,k}} + {\textit b\_time}^{\beta^{j}_{x, v_i}}(u_{i,k},
\BB{u_{i,k}, v_i})$ (refer to Figure~\ref{figure:gamma_scenario}),
implying that $u_{i,k}$ is contained in the same bucket under
$\beta^{h_i'}_{x, v_i}$ and $\beta^{j}_{x, v_i}$.
Similarly, when $h_i' + 1 \le j \le h_i$ and $h_i' + 1 \le k \le h_i$,
we have $w^{\beta^{h_i'}_{x, v_i}}_{v_i, u_{i,k}} + {\textit b\_time}^{\beta^{h_i'}_{x, v_i}}(u_{i,k},
\BB{u_{i,k}, v_i}) \le w^{\beta^{j}_{x, v_i}}_{v_i, u_{i,k}} + {\textit b\_time}^{\beta^{j}_{x, v_i}}(u_{i,k},
\BB{u_{i,k}, v_i}) \le w^{\beta^{h_i}_{x, v_i}}_{v_i, u_{i,k}} + {\textit b\_time}^{\beta^{h_i}_{x, v_i}}(u_{i,k},
\BB{u_{i,k}, v_i})$, implying that $u_{i, k}$
is not contained in any bucket under
$\beta^{h_i'}_{x, v_i}$ and $\beta^{j}_{x, v_i}$.
Consequently, $\beta^{h_i'}_{x, v_i}$ and $\beta^{j}_{x, v_i}$ can induce the same values
of $\Minv(\ell)$ and $\Acc(\ell)$ for all $1 \le \ell \le h_i$, i.e.,
${\textit b\_time}^{\beta^{h_i'}_{x, v_i}}(v_i, \BB{v_i, x})
= {\textit b\_time}^{\beta^{j}_{x, v_i}}(v_i, \BB{v_i, x})$.

So, we only need to compute ${\textit b\_time}^{\beta^{j}_{x, v_i}}(v_i, \BB{v_i, x})$
for all $1 \le j \le h_i'$.
Suppose $u_{i, j} \in \Buck[\tau_j]$ under $\beta^{j}_{x, v_i}$.
Let $\Pref(j)$ ($\Next(j)$) be a list of non-empty buckets such that $\Buck[\ell]$ is included in $\Pref(j)$ ($\Next(j)$) if and only if $\Buck[\ell]$ is not empty,
$\ell < \tau_j$ ($\ell > \tau_j$), and
$\Minv(\ell) + \Acc(\ell) \cdot \setuptime \ge \Minv(t) + \Acc(t) \cdot \setuptime$ for all $1\le t \le \ell$ ($\ell \le t \le h_i$). \linebreak
Assume $\Pref(j) = (\Buck[\pi_{j,1}],
\Buck[\pi_{j,2}], \ldots, \Buck[\pi_{j,p_j}])$ and $\Next(j) = (\Buck[\sigma_{j,1}], \linebreak \Buck[\sigma_{j,2}],\ldots,
\Buck[\sigma_{j,q_j}])$, where
$\pi_{j,1} < \pi_{j,2} < \ldots < \pi_{j,p_j}$
and $\sigma_{j,1} < \sigma_{j,2} < \ldots < \sigma_{j,q_j}$.
Then, we have ${\textit b\_time}^{\beta^{j}_{x, v_i}}(v_i,
\BB{v_i,x}) = \max\{\Minv(\pi_{j,p_j}) +
\Acc(\pi_{j,p_j}) \cdot \setuptime, \Minv(\tau_j) +
\Acc(\tau_j) \cdot \setuptime, \Minv(\sigma_{j,1}) + \Acc(\sigma_{j,1}) \cdot \setuptime \}$. Hence, it suffices to find
the values of $\Minv(\ell) +
\Acc(\ell) \cdot \setuptime$ for all $\ell \in \{\pi_{j,p_j}, \tau_j, \sigma_{j,1}\}$, in order to determine ${\textit b\_time}^{\beta^{j}_{x, v_i}}(v_i,
\BB{v_i,x})$.
The use of
$\Pref(j)$ and $\Next(j)$ (together with $\Acc(\ell)$ and $\Minv(\ell)$) is crucial
for reducing the overall time complexity to $O(h_i \log\log h_i)$.

First, with Fact~\ref{fact:preprocessing-broadcast-time},
additional $O(h_i)$ time is enough to determine ${\textit b\_time}^{\beta^{h_i'}_{x, v_i}}(v_i,
\BB{v_i,x})$.
Under $\beta^{h_i'}_{x, v_i}$, it takes $O(h_i)$ time to find the buckets (and hence $\tau_{h_i'}$)
where $u_{i,1}, u_{i,2}, \ldots, u_{i, h_i'}$ reside and determine the values of
$\Minv(\ell)$ and $\Acc(\ell)$
for all $1 \le \ell \le h_i$. Another $O(h_i)$ time is required to
find
$\Pref(h_i')$ and $\Next(h_i')$ (and hence $\pi_{h_i',p_{h_i'}}$ and $\sigma_{h_i',1}$).
Then, it requires $O(h_i \log\log h_i)$ time to determine ${\textit b\_time}^{\beta^{j}_{x, v_i}}(v_i,
\BB{v_i,x})$, sequentially, for $j=h_i'-1, h_i'-2, \ldots, 1$, as explained below.

Suppose $u_{i,j+1} \in \Buck[\hat{\ell}]$ under $\beta^j_{x,v_i}$. We first prove the following lemma.
\begin{lemma}
\label{lemma:bucket-change}
$\beta^{j}_{x, v_i}$ and $\beta^{j+1}_{x, v_i}$ induce the same
$\Buck[\ell]$ for all $\ell \in \{1, 2, \ldots, h_i\}-
\{\tau_{j+1}, \hat{\ell}\}$.
\end{lemma}
\begin{proof}
Recall that when $1 \le j \le h_i' - 1$, $\beta^{j}_{x, v_i}$ and $\beta^{j+1}_{x, v_i}$
differ in the weights of edges in
$\setof{(v_i, u_{i,j+1})} + E(\BB{u_{i,j+1}, v_i})$
(refer to Figure~\ref{figure:gamma_scenario}).
It implies that among $u_{i,1}, u_{i,2}, \ldots, u_{i, h_i'}$, only $u_{i,j+1}$
resides in different buckets
under $\beta^{j}_{x, v_i}$ and $\beta^{j+1}_{x, v_i}$.
The lemma then follows,
because $u_{i, j+1} \in \Buck[\tau_{j+1}]$ under $\beta^{j+1}_{x, v_i}$.
\end{proof}
\bigskip

The proof above also implies that $u_{i,j}$ resides in the same bucket under
$\beta^{j}_{x, v_i}, \beta^{j+1}_{x, v_i}, \ldots, \linebreak \beta^{h_i'}_{x, v_i}$. Hence,
$\tau_j$ can be determined in constant time.
With Fact~\ref{fact:preprocessing-broadcast-time}, the value of
$w^{\beta^{j}_{x, v_i}}_{v_i,
u_{i,j+1}} + {\textit b\_time}^{\beta^{j}_{x, v_i}}(u_{i,j+1},
\BB{u_{i,j+1}, v_i})$ can be obtained, and hence $\Buck[\hat{\ell}]$ can be determined, also in constant time.
According to Lemma~\ref{lemma:bucket-change},
it takes constant time to obtain
$\Minv(1), \linebreak \Minv(2), \ldots, \Minv(h_i), |\Buck[1]|, |\Buck[2]|, \ldots, |\Buck[h_i]|$
under $\beta^{j}_{x, v_i}$ from $\Minv(1), \linebreak \Minv(2), \ldots,  \Minv(h_i), |\Buck[1]|, |\Buck[2]|, \ldots, |\Buck[h_i]|$
under $\beta^{j+1}_{x, v_i}$.
%

Since $w^{\beta^{h_i'}_{x, v_i}}_{v_i, u_{i,j}} + {\textit
b\_time}^{\beta^{h_i'}_{x, v_i}}(u_{i,j}, \BB{u_{i,j}, v_i})
\ge w^{\beta^{h_i'}_{x, v_i}}_{v_i, u_{i,j+1}} +  {\textit
b\_time}^{\beta^{h_i'}_{x, v_i}}(u_{i,j+1}, \BB{u_{i,j+1}, v_i})$,\linebreak $u_{i, j} \in \Buck[\tau_j]$
under $\beta^{j}_{x, v_i}$ and $\beta^{h_i'}_{x, v_i}$,
and $u_{i, j+1}
\in \Buck[\tau_{j+1}]$
under $\beta^{j+1}_{x, v_i}$ and $\beta^{h_i'}_{x, v_i}$,
we have $\tau_j \le \tau_{j+1}$.
Also, since $w^{\beta^{j+1}_{x, v_i}}_{v_i, u_{i,j+1}} + {\textit
b\_time}^{\beta^{j+1}_{x, v_i}}(u_{i,j+1}, \BB{u_{i,j+1}, v_i})
\ge w^{\beta^{j}_{x, v_i}}_{v_i, u_{i,j+1}} + {\textit
b\_time}^{\beta^{j}_{x, v_i}}(u_{i,j+1}, \BB{u_{i,j+1}, v_i})$,
we have $\tau_{j+1} \le \hat{\ell}$.
So, by Lemma~\ref{lemma:bucket-change}, $\beta^{j}_{x, v_i}$ and $\beta^{j+1}_{x, v_i}$ induce the same
$\Buck[\ell]$ for all $1\le \ell \le \tau_{j} -1$, i.e., the same
$\Acc(\tau_j-1)$.
In this way, $\Acc(\tau_j-1)$ remains the same under
$\beta^{j}_{x, v_i}, \beta^{j+1}_{x, v_i}, \ldots, \beta^{h_i'}_{x, v_i}$,
and hence $\Acc(\tau_j)$ (=$\Acc(\tau_{j}-1) + |\Buck[\tau_j]|$ under $\beta^j_{x, v_i}$) under $\beta^{j}_{x, v_i}$ can be determined in constant time.
Therefore, it takes $O(h_i)$ time to
determine the values of $\Minv(\tau_j) + \Acc(\tau_j) \cdot \setuptime$ under $\beta^j_{x, v_i}$
for all $1 \le j \le h_i' - 1$.

Also, as a consequence of $\tau_j \le \tau_{j+1}$ and
$\beta^{j}_{x, v_i}$, $\beta^{j+1}_{x, v_i}$ inducing the same
$\Buck[\ell]$ for all $1\le \ell \le \tau_{j} -1$,
$\Pref(j)$ is a sublist of $\Pref(j+1)$ for all $1 \le j \le h_i' - 1$.
Since $\pi_{h_i', p_{h_i'}}$ is available, $O(h_i)$ time is enough to determine
$\pi_{j, p_j}$ for all
$1 \le j \le h_i' - 1$.
Therefore, it takes $O(h_i)$ time to determine the
values of $\Minv(\pi_{j,p_j})+\Acc(\pi_{j,p_j}) \cdot \setuptime$ under $\beta^j_{x, v_i}$ for all
$1 \le j \le h_i' - 1$.
The computation of $\Minv(\sigma_{j,1})+\Acc(\sigma_{j,1}) \cdot \setuptime$ under $\beta^j_{x, v_i}$ for all
$1 \le j \le h_i' - 1$,
which is detailed in Section~\ref{subsection:maintain_chain}, takes
$O(h_i \log\log h_i)$ time.

\subsection{Computation of $\Minv(\sigma_{j,1})+ \Acc(\sigma_{j,1}) \cdot \setuptime$ under $\beta^j_{x, v_i}$}
\label{subsection:maintain_chain}
Without loss of generality,
we suppose $h_i' = h_i$. First we claim that $\Next(j) = (\Buck[\sigma_{j,1}], \linebreak \Buck[\sigma_{j,2}],\ldots,
\Buck[\sigma_{j,q_j}])$ for all $1 \le j < h_i$ can be constructed in $O(h_i\log\log h_i)$ time, which will be verified later.
It was shown earlier that the values of $\Minv(\ell)$ under $\beta^{j}_{x,
v_i}$ for all $1 \le \ell \le h_i$ and $1 \le j < h_i$ can be determined in $O(h_i)$ time.
On the other hand,
we show in the next paragraph that the values of $\Acc(\sigma_{j,1})$ under $\beta^{j}_{x,
v_i}$ for all $1 \le j < h_i$ can be determined in additional $O(h_i)$ time.

Define $\Delta_{j, t} = (\Acc(\sigma_{j,t}) - \Acc(\sigma_{j,t-1})$ under $\beta^{j}_{x,
v_i}) - (\Acc(\sigma_{j,t}) - \Acc(\sigma_{j,t-1})$
under $\beta^{h_i}_{x, v_i})$, where $1 \le t \le q_j$ and
$\sigma_{j, 0}$ is set to $\tau_{j}$.
The values of $\Delta_{j, t}$ for all $1 \le t \le q_j$
can be obtained, while constructing $\Next(j)$.
If $\Next(j)$ is empty, then $\sigma_{j,t}$ and $\Delta_{j,t}$ are not defined.
When $t = 1$, we have
\vspace{-8pt}
\begin{equation}
\label{equation:difference-relation}
\Delta_{j, 1} = (\Acc(\sigma_{j,1}) - \Acc(\tau_j)\textrm{ under }\beta^{j}_{x,
v_i}) - (\Acc(\sigma_{j,1}) - \Acc(\tau_j)
\textrm{ under }\beta^{h_i}_{x, v_i}).
\vspace{-8pt}
\end{equation}
Now that the values of $\Acc(\tau_j)$ under $\beta^{j}_{x,
v_i}$ for all $1\le j < h_i$ and $\Acc(\ell)$ under $\beta^{h_i}_{x, v_i}$ for all $1 \le \ell \le h_i$
can be determined in $O(h_i)$ time (shown earlier), the values of $\Acc(\sigma_{j,1})$ under $\beta^{j}_{x,
v_i}$ for all $1 \le j < h_i$ can be determined in additional $O(h_i)$ time according to (\ref{equation:difference-relation}).

Next we show the construction of $\Next(j)$ and $\Delta_{j, t}$ ($1 \le t \le q_j$) from $\Next(j+1)$ and $\Delta_{j+1,t}$
($1 \le t \le q_{j+1}$), where
$1 \le j < h_i$.
We use the van Emde Boas priority queue
\cite{van76} to store the indices, i.e., $\sigma_{j, t}$'s,
of $\Next(j)$ for all $1\le j \le h_i$.
Initially, $\Next(j)$ is empty for all $1\le j \le h_i$ (and hence $\Delta_{j, t}$ is not defined for all $1 \le t \le q_j$).
The construction comprises four
sequential stages, i.e., Stage 1 to Stage 4, which can find all buckets
$\Buck[\ell]\in\Next(j)$
for $\ell \in \{\hat{\ell}+1, \hat{\ell}+2, \ldots, h_i\}, \{\hat{\ell}\}, \{\tau_{j+1}+1, \tau_{j+1}+2, \ldots, \hat{\ell}-1\}$, and
$\{\tau_j+1, \tau_j+2, \ldots, \tau_{j+1}\}$,
respectively, where $\tau_{j+1} \le \hat{\ell} \le h_i$.
If $\hat{\ell} = h_i$, then Stage 1 is omitted and the construction starts with Stage 2.
If $\hat{\ell} \le \tau_{j+1} + 1$, then Stage 3 is omitted
(Stage 2 is also omitted as $\hat{\ell} = \tau_{j+1}$).
If $\tau_{j+1} = \tau_{j}$, then Stage 4 is omitted.

As a consequence of Lemma~\ref{lemma:bucket-change}, when $\hat{\ell} > \tau_{j+1}$, we have
\vspace{6pt}
\begin{equation}
\label{equation:change-scenario}
\Acc(c) \textrm{ under }\beta^{j}_{x, v_i} =\left\{
\begin{array}{ll}
\Acc(c) - 1\textrm{ under }\beta^{j+1}_{x, v_i},&\textrm{if }\tau_{j+1} \le c \le \hat {\ell}-1; \\[12pt]
\Acc(c)\textrm{ under }\beta^{j+1}_{x, v_i},&\textrm{if }\hat{\ell} \le c \le h_i. \end{array} \right.
\end{equation}

%


\vspace{6pt}
\noindent {\bf Stage 1.} $\ell \in \{\hat{\ell}+1, \hat{\ell}+2, \ldots, h_i\} = L_1$.
As a consequence of Lemma~\ref{lemma:bucket-change}, the values of $\Minv(\ell)$ and $\Acc(\ell)$
remain the same under $\beta^{j}_{x, v_i}$ and $\beta^{j+1}_{x, v_i}$, which
implies $\Buck[\ell] \in \Next(j)$ if and only if
$\Buck[\ell] \in \Next(j+1)$. Moreover, those buckets
$\Buck[\ell]\in\Next(j)$,
if they exist,
constitute a sublist $(\Buck[\sigma_{j+1,\hat{t}}],\Buck[\sigma_{j+1,\hat{t}+1}],\ldots,
\Buck[\sigma_{j+1,q_{j+1}}])$
of\linebreak $\Next(j+1)$, where $\Buck[\sigma_{j+1,\hat{t}}]$
is the bucket that immediately succeeds $\Buck[\hat{\ell}]$
in $\Next(j+1)$, i.e.,
$\sigma_{j+1, \hat{t}} = \min\{\sigma_{j+1, t} \mid \sigma_{j+1, t} > \hat{\ell}$ and $0\le t \le q_{j+1}\}$.
If $\Next(j+1)$ is empty or $\sigma_{j+1, t} \le \hat{\ell}$ for all $1 \le t \le q_{j+1}$,
then $\Next(j)$ remains empty and $\Delta_{j,t}$ for all $1 \le t \le q_{j}$ remain not defined.

Otherwise,
by the aid of the van Emde Boas priority queue, the value of $\sigma_{j+1,\hat{t}}$ can be determined in $O(\log\log h_i)$ time. Then, $\Buck[\sigma_{j+1, \hat{t}}]$ can be found in $\Next(j+1)$ and hence
the sublist of $\Next(j+1)$ can be added to $\Next(j)$ in constant time. That is,\linebreak we have
$(\Buck[\sigma_{j,\ddot{t}}],  \Buck[\sigma_{j,\ddot{t}+1}], \ldots, \Buck[\sigma_{j,q_j}]) = (\Buck[\sigma_{j+1,\hat{t}}],  \Buck[\sigma_{j+1,\hat{t}+1}], \ldots, \linebreak \Buck[\sigma_{j+1,q_{j+1}}])$, where $q_j - \ddot{t} = q_{j+1} - \hat{t}$.
Moreover,
we have $\Acc(\sigma_{j,\ddot{t}}),  \Acc(\sigma_{j,\ddot{t}+1}), \ldots, \linebreak\Acc(\sigma_{j,q_j})$ under $\beta^j_{x,v_i}$
equal to $\Acc(\sigma_{j+1,\hat{t}}), \Acc(\sigma_{j+1,\hat{t}+1}), \ldots, \Acc(\sigma_{j+1,q_{j+1}})$ under $\beta^{j+1}_{x,v_i}$, respectively,
which implies $\Delta_{j, \ddot{t}+1}, \Delta_{j, \ddot{t}+2}, \ldots,  \Delta_{j, q_j}$ equal to $\Delta_{j+1, \hat{t}+1}, \Delta_{j+1, \hat{t}+2}, \ldots,\Delta_{j+1, q_{j+1}}$, respectively.
We also set $\Delta_{j, \ddot{t}} = \Delta_{j+1, \hat{t}}$ temporarily.
It may help the construction of
$\Next(j)$ in subsequent stages. 

\vspace{6pt}
\noindent {\bf Stage 2.} $\ell \in \{\hat{\ell}\} = L_2$, where $\hat{\ell} > \tau_{j+1}$.
We use $\Buck[\sigma_{j+1, \breve{t}}]$ to denote the
bucket that immediately precedes $\Buck[\hat{\ell}]$ in $\{\Buck[\tau_{j+1}]\} \cup \Next(j+1)$,
i.e., $\sigma_{j+1, \breve{t}} =
\max\{\sigma_{j+1,t}\mid\sigma_{j+1, t} < \hat{\ell}$ and $0\le t \le q_{j+1}\}$,
where $\sigma_{j+1, 0}=\tau_{j+1}$.
By the aid of the van Emde Boas priority queue,
$\sigma_{j+1,\breve{t}}$ can be found in $O(\log\log h_i)$ time.
If $\Next(j)$ remains empty after Stage 1,
then $\Buck[\hat{\ell} + 1],\Buck[\hat{\ell} + 2],\ldots,\Buck[h_i]$ under $\beta^{j}_{x, v_i}$
are all empty.
Since $\Buck[\hat{\ell}]$ is not empty under $\beta^{j}_{x, v_i}$, we have
$\hat{\ell} = \sigma_{j, q_j}$ and add $\Buck[\hat{\ell}]$ to $\Next(j)$.
We set $\Delta_{j,q_j} = 1 + \Delta_{j+1, q_{j+1}}$ if $\Buck[\hat{\ell}] \in \Next(j+1)$,
and $1 - (\Acc(\sigma_{j, q_j}) - \Acc(\sigma_{j+1, \breve{t}})$ under $\beta^{h_i}_{x,v_i})$ if $\Buck[\hat{\ell}] \not\in \Next(j+1)$.
Let $\nabla = (\Acc(\sigma_{j, q_j}) - \Acc(\sigma_{j+1,\breve{t}})$ under $\beta^{j}_{x,v_i}) - (\Acc(\sigma_{j, q_j}) - \Acc(\sigma_{j+1,\breve{t}})$
under $\beta^{h_i}_{x, v_i})$.
For both cases, we have $\Delta_{j, q_j} = \nabla$, as explained below.

According to (\ref{equation:change-scenario}),
we have $\nabla = 1 + (\Acc(\sigma_{j, q_j}) - \Acc(\sigma_{j+1,\breve{t}})$ under
$\beta^{j+1}_{x,v_i}) - (\Acc(\sigma_{j, q_j}) - \Acc(\sigma_{j+1,\breve{t}})$
under $\beta^{h_i}_{x, v_i})$.
When $\Buck[\hat{\ell}] \in \Next(j+1)$,
we have ($\sigma_{j, q_j}=$) $\hat{\ell} = \sigma_{j+1, q_{j+1}}$ and
$\sigma_{j+1,\breve{t}} = \sigma_{j+1, q_{j+1} - 1}$, i.e.,
$\nabla = 1 + \Delta_{j+1, q_{j+1}}$.
When $\Buck[\hat{\ell}] \not\in \Next(j+1)$, we have
$\sigma_{j+1,\breve{t}} = \tau_{j+1}$ if $\Next(j+1)$ is empty,
and $\sigma_{j+1, q_{j+1}}$ otherwise.
It follows that we have
$\Buck[\sigma_{j+1, \breve{t}}+1], \Buck[\sigma_{j+1, \breve{t}}+2],\ldots,
\Buck[\hat{\ell}]$ under $\beta^{j+1}_{x,v_i}$ all empty,
and so $\Acc(\sigma_{j+1, \breve{t}}) = \Acc(\hat{\ell})$ ($=\Acc(\sigma_{j, q_j})$) under $\beta^{j+1}_{x,v_i}$,
i.e., $\nabla = 1 - (\Acc(\sigma_{j, q_j}) - \Acc(\sigma_{j+1, \breve{t}})$ under $\beta^{h_i}_{x,v_i})$.
Also notice that the
value of $\Delta_{j, q_j}$ above is temporary.

In the rest of Stage 2, we consider the situation that Stage 1 is carried out and $\Next(j)$ is not
empty after Stage 1. Then we have $\Next(j) =
(\Buck[\sigma_{j,\ddot{t}}], \Buck[\sigma_{j,\ddot{t}+1}], \ldots, \linebreak\Buck[\sigma_{j,q_j}])$
after Stage 1.
We claim
$\Buck[\hat{\ell}] \in \Next(j)$ if and only if
$\Minv(\hat{\ell})$ under $\beta^j_{x,v_i} + \Acc(\sigma_{j+1, \hat{t}-1}) \cdot \setuptime$
under $\beta^{h_i}_{x,v_i} \ge
\Minv(\sigma_{j+1, \hat{t}})$ under $\beta^j_{x,v_i} +  \Acc(\sigma_{j+1,\hat{t}}) \cdot \setuptime$ under
$\beta^{h_i}_{x,v_i} + \Delta_{j+1, \hat{t}} \cdot \setuptime$, as explained below.
It follows that whether
$\Buck[\hat{\ell}] \in \Next(j)$ or not can be determined in constant time.

The inequality of the if-and-only-if statement can be written as
$\Minv(\hat{\ell})$ under $\beta^j_{x,v_i} + \Acc(\sigma_{j+1, \hat{t}-1}) \cdot \setuptime$
under $\beta^{j+1}_{x,v_i} \ge
\Minv(\sigma_{j+1, \hat{t}})$ under $\beta^j_{x,v_i} +\Acc(\sigma_{j+1,\hat{t}}) \cdot \setuptime$ under
$\beta^{j+1}_{x,v_i}$.
%
If $\Buck[\hat{\ell}] \in \Next(j+1)$, then we have
$\hat{\ell} = \sigma_{j+1,\hat{t}-1}$.
Since
$\Buck[\hat{\ell}]$ is not empty and $\hat{\ell} > \tau_{j+1} ~(\ge \tau_j)$, we have
$\Buck[\hat{\ell}] \in \Next(j)$ if and only if
$\Minv(\hat{\ell}) + \Acc(\hat{\ell}) \cdot \setuptime \ge
\Minv(\sigma_{j+1, \hat{t}}) + \Acc(\sigma_{j+1,\hat{t}}) \cdot \setuptime$ under
$\beta^{j}_{x,v_i}$.
The latter inequality is equivalent to the inequality above,
as a consequence of (\ref{equation:change-scenario}).

Then we assume $\Buck[\hat{\ell}] \not\in \Next(j+1)$, which implies $\hat{\ell} > \sigma_{j+1,\hat{t}-1}$.
If $\Buck[\hat{\ell}] \not\in \Next(j)$, then we have $\Minv(\hat{\ell}) + \Acc(\hat{\ell}) \cdot \setuptime < \Minv(\sigma_{j+1,\hat{t}})
+ \Acc(\sigma_{j+1,\hat{t}}) \cdot \setuptime$ under $\beta^{j}_{x,v_i}$, from which the inequality
$\Minv(\hat{\ell})$ under $\beta^j_{x,v_i} + \Acc(\sigma_{j+1, \hat{t}-1}) \cdot \setuptime$
under $\beta^{j+1}_{x,v_i}
< \Minv(\sigma_{j+1, \hat{t}})$ under $\beta^j_{x,v_i} +  \Acc(\sigma_{j+1,\hat{t}}) \cdot \setuptime$ under
$\beta^{j+1}_{x,v_i}$ can be derived, as a consequence of (\ref{equation:change-scenario}) and $\Acc(\hat{\ell}) \ge \Acc(\sigma_{j+1,\hat{t}-1})$
under $\beta^{j+1}_{x,v_i}$.
If $\Buck[\hat{\ell}] \in \Next(j)$, then we have $\Minv(\hat{\ell}) +  \Acc(\hat{\ell}) \cdot \setuptime 
\ge \Minv(\sigma_{j+1,\hat{t}})
+ \Acc(\sigma_{j+1,\hat{t}}) \cdot \setuptime$ under $\beta^{j}_{x,v_i}$, from which the inequality $\Minv(\hat{\ell})$ under $\beta^j_{x,v_i} + \Acc(\sigma_{j+1, \hat{t}-1}) \cdot \setuptime$
under $\beta^{j+1}_{x,v_i} \ge
\Minv(\sigma_{j+1, \hat{t}})$ under $\beta^j_{x,v_i} + \Acc(\sigma_{j+1,\hat{t}}) \cdot \setuptime$ under
$\beta^{j+1}_{x,v_i}$ can be derived, because $\Acc(\hat{\ell}) = \Acc(\sigma_{j+1,\hat{t}-1})$ under $\beta^{j+1}_{x, v_i}$.
The latter equality is verified below.

To prove $\Acc(\hat{\ell}) = \Acc(\sigma_{j+1,\hat{t}-1})$ under $\beta^{j+1}_{x, v_i}$, it suffices to show that
$\Buck[\sigma_{j+1,\hat{t}-1} + 1],\linebreak \Buck[\sigma_{j+1,\hat{t}-1} + 2], \ldots, \Buck[\hat{\ell}]$
under $\beta^{j+1}_{x, v_i}$ are all empty.
Suppose to the contrary that there exists
$\sigma_{j+1, \hat{t}-1}+ 1 \le \ell' \le \hat{\ell}$
such that under $\beta^{j+1}_{x,v_i}$, $\Buck[\ell']$ is not empty and
 $\Buck[\ell'+1], \Buck[\ell'+2], \ldots,
\Buck[\hat{\ell}]$ are all empty.
Then we have $\Acc(\ell') = \Acc(\hat{\ell})$ under $\beta^{j+1}_{x, v_i}$.
As a consequence of (\ref{equation:change-scenario}), we have
$\Minv(\ell')+ \Acc(\ell') \cdot \setuptime$ under $\beta^{j+1}_{x, v_i} \ge
\Minv(\hat{\ell})+ \Acc(\hat{\ell}) \cdot \setuptime$ under $\beta^{j}_{x, v_i}$.
Further, since $\Minv(\sigma_{j+1, \hat{t}}) +  \Acc(\sigma_{j+1,\hat{t}}) \cdot \setuptime >
\Minv(\ell') +  \Acc(\ell') \cdot \setuptime$
under $\beta^{j+1}_{x, v_i}$, we have
$\Minv(\sigma_{j+1, \hat{t}}) +  \Acc(\sigma_{j+1,\hat{t}}) \cdot \setuptime$
under $\beta^{j+1}_{x, v_i} > \Minv(\hat{\ell})+\Acc(\hat{\ell}) \cdot \setuptime$ under $\beta^{j}_{x, v_i}$. Then, by Lemma~\ref{lemma:bucket-change}, the inequality
$\Minv(\sigma_{j+1, \hat{t}}) + \Acc(\sigma_{j+1,\hat{t}}) \cdot \setuptime$ under $\beta^{j}_{x, v_i} >
\Minv(\hat{\ell}) + \Acc(\hat{\ell}) \cdot \setuptime$ under $\beta^{j}_{x, v_i}$ can result, which contradicts to $\Buck[\hat{\ell}] \in \Next(j)$.



Next, the computations of $\Next(j)$ and $\Delta_{j, t}$ are described.
If $\Buck[\hat{\ell}] \in \Next(j)$, then
we add $\Buck[\hat{\ell}]$
to the front of $\Next(j)$, i.e.,
$\sigma_{j, \ddot{t}-1} = \hat{\ell}$.
As a consequence of (\ref{equation:change-scenario}),
$\sigma_{j,\ddot{t}} = \sigma_{j+1, \hat{t}}$, and
$\Acc(\sigma_{j, \ddot{t}-1}) =\Acc(\hat{\ell}) = \Acc(\sigma_{j+1,\hat{t}-1})$
under $\beta^{j+1}_{x, v_i}$, we have $\Delta_{j,\ddot{t}} =
(\Acc(\sigma_{j,\ddot{t}}) - \Acc(\sigma_{j, \ddot{t}-1})$ under $\beta^{j}_{x,
v_i}) - (\Acc(\sigma_{j,\ddot{t}}) - \Acc(\sigma_{j, \ddot{t}-1})
$ under $\beta^{h_i}_{x, v_i}) =
\Delta_{j+1,\hat{t}} - (\Acc(\sigma_{j+1,\hat{t}-1}) - \Acc(\sigma_{j, \ddot{t}-1})$
under $\beta^{h_i}_{x,v_i})$, which can be computed in constant time.
In addition, we set $\Delta_{j, \ddot{t}-1} = 1 + \Delta_{j+1,\hat{t}-1}$ if $\hat{\ell} =
\sigma_{j+1,\hat{t}-1}$, and $1 -
(\Acc(\sigma_{j, \ddot{t}-1}) - \Acc(\sigma_{j+1,\breve{t}})$ under $\beta^{h_i}_{x, v_i})$ if $\hat{\ell} >
\sigma_{j+1,\hat{t}-1}$.
Let
$\nabla' = (\Acc(\sigma_{j, \ddot{t}-1}) - \Acc(\sigma_{j+1,\breve{t}})$ under $\beta^{j}_{x,v_i}) -
(\Acc(\sigma_{j, \ddot{t}-1}) - \Acc(\sigma_{j+1,\breve{t}})$
under $\beta^{h_i}_{x, v_i})$.
For both cases, we have $\Delta_{j, \ddot{t}-1} = \nabla'$,
as a consequence of (\ref{equation:change-scenario}) and the fact that we have
$\sigma_{j+1, \breve{t}} = \sigma_{j+1, \hat{t} - 2}$
if $\hat{\ell} = \sigma_{j+1, \hat{t} - 1}$, and
$\sigma_{j+1, \hat{t} - 1}$ if $\hat{\ell} >
\sigma_{j+1, \hat{t} - 1}$.

If $\Buck[\hat{\ell}] \not\in \Next(j)$,
then $\Next(j)$ remains unchanged.
We set
$\Delta_{j, \ddot{t}} = 1 + \Delta_{j+1, \hat{t} - 1} + \Delta_{j+1,\hat{t}}$ if $\hat{\ell} =
\sigma_{j+1, \hat{t} - 1}$, and
$1 + \Delta_{j+1,\hat{t}}$ if $\hat{\ell} >
\sigma_{j+1, \hat{t} - 1}$.
Let
$\nabla'' = (\Acc(\sigma_{j,\ddot{t}}) - \Acc(\sigma_{j+1,\breve{t}})$ under $\beta^{j}_{x,v_i}) -
(\Acc(\sigma_{j,\ddot{t}}) - \Acc(\sigma_{j+1,\breve{t}})$
under $\beta^{h_i}_{x, v_i})$.
For both cases, we have $\Delta_{j, \ddot{t}} = \nabla''$,
as a consequence of (\ref{equation:change-scenario}) and $\sigma_{j,\ddot{t}}
= \sigma_{j+1,\hat{t}}$.
The values of $\Delta_{j, \ddot{t}-1}$ (as $\Buck[\hat{\ell}] \in \Next(j)$)
and $\Delta_{j, \ddot{t}}$ (as $\Buck[\hat{\ell}] \not\in \Next(j)$)
above are temporary.


Recall that we use the van Emde Boas priority queue to store the indices, i.e., $\sigma_{j+1,t}$'s,
of $\Next(j+1)$, and construct $\Next(j)$ from  $\Next(j+1)$. We need to insert $\hat{\ell}$ into the queue if $\Buck[\hat{\ell}] \not\in \Next(j+1)$ and $\Buck[\hat{\ell}] \in \Next(j)$,
and delete $\hat{\ell}$ from the queue if $\Buck[\hat{\ell}] \in \Next(j+1)$ and $\Buck[\hat{\ell}] \not\in \Next(j)$.
Both take $O(\log\log h_i)$ time.

\vspace{6pt}
\noindent {\bf Stage 3.} $\ell \in \{\tau_{j+1}+1, \tau_{j+1}+2, \ldots, \hat{\ell}-1\} = L_3$, where
$\hat{\ell} > \tau_{j+1}+1$.
We first claim that if $\Buck[\ell] \not\in \Next(j+1)$, then $\Buck[\ell] \not\in \Next(j)$, as explained below.
Suppose $\Buck[\ell] \not\in \Next(j+1)$.
If $\Buck[\ell]$ is empty under $\beta^{j+1}_{x,v_i}$, then $\Buck[\ell]$ is also empty under $\beta^{j}_{x,v_i}$,
according to Lemma~\ref{lemma:bucket-change}.
Otherwise, there exists some $\Buck[\ell']$, where $\ell < \ell'$, satisfying
$\Minv(\ell) +  \Acc(\ell) \cdot \setuptime < \Minv(\ell') +
 \Acc(\ell') \cdot \setuptime$ under $\beta^{j+1}_{x, v_i}$.
As a consequence of Lemma~\ref{lemma:bucket-change} and (\ref{equation:change-scenario}),
we have
$(\Minv(\ell) + \Acc(\ell) \cdot \setuptime$ under $\beta^{j}_{x, v_i}) = (\Minv(\ell) + (\Acc(\ell) - 1) \cdot \setuptime$
under $\beta^{j+1}_{x, v_i})$.
Similarly, we have
$(\Minv(\ell') + \Acc(\ell') \cdot \setuptime$
under $\beta^{j}_{x, v_i}) = (\Minv(\ell') + \Acc(\ell') \cdot \setuptime$ under $\beta^{j+1}_{x, v_i})$
if $\ell' \in L_1$, and
$\ge (\Minv(\ell') + (\Acc(\ell') - 1) \cdot \setuptime$ under $\beta^{j+1}_{x, v_i})$
if $\ell' \in L_2 \cup L_3$.
It follows that we have $\Minv(\ell) + \Acc(\ell) \cdot \setuptime < \Minv(\ell') + \Acc(\ell') \cdot \setuptime$
under $\beta^{j}_{x,v_i}$, i.e., $\Buck[\ell] \not\in \Next(j)$.

So, we only need to decide whether $\Buck[\ell] \in \Next(j)$ or not
for each $\Buck[\ell] \in \Next(j+1)$.
All $\Buck[\ell]$'s in $\Next(j+1)$,
if they exist,
constitute a sublist, i.e., \linebreak $(\Buck[\sigma_{j+1,1}],\Buck[\sigma_{j+1,2}],\ldots,
\Buck[\sigma_{j+1,\breve{t}}])$, of $\Next(j+1)$.
Suppose that
$\Buck[\sigma_{j, \Head}]$ is the first
bucket in $\Next(j)$ after Stage 2.
For $1 \le t \le \breve{t}$, since
$\sigma_{j+1,t} > \tau_{j+1} ~(\ge \tau_j)$, we have
$\Buck[\sigma_{j+1,t}] \in \Next(j)$ if and only if
$\Minv(\sigma_{j+1,t}) +  \Acc(\sigma_{j+1,t}) \cdot \setuptime \ge
\Minv(\sigma_{j+1, \Head}) + \Acc(\sigma_{j+1,\Head}) \cdot \setuptime$ under
$\beta^{j}_{x,v_i}$, as a consequence of Lemma~\ref{lemma:bucket-change}.
Let
$\DT_t = (\Acc(\sigma_{j,\Head}) -
\Acc(\sigma_{j+1,t})$ under $\beta^{j}_{x,v_i}) - (\Acc(\sigma_{j,\Head}) -
\Acc(\sigma_{j+1,t})$ under $\beta^{h_i}_{x, v_i})$. Then, we have
$\Buck[\sigma_{j+1,t}] \in \Next(j)$ if and only if
$\Minv(\sigma_{j+1,t})$ under $\beta^j_{x,v_i} +  \Acc(\sigma_{j+1, t}) \cdot \setuptime$
under $\beta^{h_i}_{x,v_i} \ge \DT_t \cdot \setuptime + \Minv(\sigma_{j,\Head})$ under $\beta^j_{x,v_i} + \Acc(\sigma_{j,\Head}) \cdot \setuptime$ under
$\beta^{h_i}_{x,v_i}$.
It means that whether
$\Buck[\sigma_{j+1,t}] \in \Next(j)$ or not can be determined in constant time as long as $\DT_t$ is available.
%


Recall that $\Delta_{j, \Head}$ ($=\nabla$ or $\nabla'$ or
$\nabla''$) was set to
$(\Acc(\sigma_{j,\Head}) -
\Acc(\sigma_{j+1,\breve{t}})$ under $\beta^{j}_{x,v_i}) - (\Acc(\sigma_{j,\Head}) -
\Acc(\sigma_{j+1,\breve{t}})$ under $\beta^{h_i}_{x, v_i})$ after Stage 2, which
is identical with $\DT_{\breve{t}}$.
As a consequence of (\ref{equation:change-scenario}),
we have $\DT_{\breve{t}-1} = (\Acc(\sigma_{j,\Head}) -
\Acc(\sigma_{j+1,\breve{t}})$ under $\beta^{j}_{x,v_i}) - (\Acc(\sigma_{j,\Head}) -
\Acc(\sigma_{j+1,\breve{t}})$ under $\beta^{h_i}_{x, v_i}) + (\Acc(\sigma_{j+1,\breve{t}}) -
\Acc(\sigma_{j+1,\breve{t}-1})$ under $\beta^{j}_{x,v_i}) - (\Acc(\sigma_{j+1,\breve{t}}) -
\Acc(\sigma_{j+1,\breve{t}-1})$ under $\beta^{h_i}_{x, v_i}) =
\DT_{\breve{t}} + \Delta_{j+1, \breve{t}}$,
which can be computed in constant time.
Similarly, $\DT_{\breve{t}-2}, \DT_{\breve{t}-3}, \linebreak \ldots, \DT_1$
each can be obtained in constant time, by the aid of (\ref{equation:change-scenario}).
%
%


If $\Buck[\sigma_{j+1, t}] \not\in \Next(j)$ for all $1 \le t \le \breve{t}$,
then $\Next(j)$ remains unchanged and
we set $\Delta_{j, \Head} = \DT_1 + \Delta_{j+1,1}$.
Otherwise,
suppose $\Buck[\sigma_{j+1,t'}] \in \Next(j)$
and $\Buck[\sigma_{j+1,t}] \not\in \Next(j)$
for all $t' < t \le \breve{t}$. Then,
as a consequence of Lemma~\ref{lemma:bucket-change},
we have $\Buck[\sigma_{j+1,1}], \linebreak \Buck[\sigma_{j+1,2}],
\ldots, \Buck[\sigma_{j+1,t'-1}]$ all contained in $\Next(j)$. They
can be added in constant time to the front of $\Next(j)$,  i.e.,
$(\Buck[\sigma_{j,t''}],  \Buck[\sigma_{j,t''+1}], \ldots, \Buck[\sigma_{j,\Head-1}]) =
(\Buck[\sigma_{j+1,1}], \Buck[\sigma_{j+1,2}], \ldots,
\Buck[\sigma_{j+1,t'}])$,
where $\Head - t'' = t'$.
Moreover, as a consequence of (\ref{equation:change-scenario}),
we have $(\Delta_{j, t''+1}, \Delta_{j, t''+2}, \ldots,  \Delta_{j, \Head-1})=(\Delta_{j+1, 2}, \Delta_{j+1, 3}, \ldots,\Delta_{j+1, t'})$.
We also set $\Delta_{j, t''} =
\Delta_{j+1,1}$ and $\Delta_{j, \Head} = \DT_{t'}$.
It follows that we have $
\Delta_{j, \Head} = (\Acc(\sigma_{j,\Head}) - \Acc(\sigma_{j+1,t'})$ under $\beta^{j}_{x,
v_i}) - (\Acc(\sigma_{j,\Head}) - \Acc(\sigma_{j+1,t'})$
under $\beta^{h_i}_{x, v_i}) = (\Acc(\sigma_{j,\Head}) - \Acc(\sigma_{j,\Head-1})$ under $\beta^{j}_{x,
v_i}) - (\Acc(\sigma_{j,\Head}) - \Acc(\sigma_{j,\Head-1})$
under $\beta^{h_i}_{x, v_i})$.

Let
$\Buck[\sigma_{j, \Head'}]$ be the first
bucket in the current $\Next(j)$, and
$\nabla''' = (\Acc(\sigma_{j, \Head'}) - \Acc(\tau_{j+1})$ under $\beta^{j}_{x,v_i}) -
(\Acc(\sigma_{j, \Head'}) - \Acc(\tau_{j+1})$
under $\beta^{h_i}_{x, v_i})$, where $\Head' = \Head$ if $\Buck[\sigma_{j+1, t}] \not\in \Next(j)$ for all $1 \le t \le \breve{t}$, and $\Head' = t''$ else. As a consequence of (\ref{equation:difference-relation})
and (\ref{equation:change-scenario}), we have
$\Delta_{j+1,1} =(\Acc(\sigma_{j+1,1}) - \Acc(\tau_{j+1})$ under $\beta^{j}_{x,
v_i}) - (\Acc(\sigma_{j+1,1}) - \Acc(\tau_{j+1})$
under $\beta^{h_i}_{x, v_i})$. It follows that we have
$\Delta_{j, \Head'} = \nabla'''$ for both cases of $\Head' = \Head$ and $\Head' = t''$.
The value of $\Delta_{j, \Head'}$
is temporary.

Recall that we have to delete $\sigma_{j+1, t}$ from the van Emde Boas priority queue,
if $\Buck[\sigma_{j+1, t}] \not\in \Next(j)$.
If there are $\delta_j$ buckets in the sublist $(\Buck[\sigma_{j+1,1}], \Buck[\sigma_{j+1,2}],\linebreak\ldots,
\Buck[\sigma_{j+1,\breve{t}}])$ of $\Next(j+1)$ that are not in $\Next(j)$, then
it requires $O(\delta_j \log\log h_i)$ time
to delete these buckets from
the queue.

%
%
%

\vspace{6pt}
\noindent {\bf Stage 4.} $\ell \in \{\tau_j+1, \tau_j+2, \ldots, \tau_{j+1}\} = L_4$.
Since $u_{i,j}$ resides
in the same bucket under $\beta^{j}_{x, v_i}$ and
$\beta^{h_i}_{x, v_i}$ (refer to the proof of Lemma~\ref{lemma:bucket-change}),
we have $u_{i, j} \in \Buck[\tau_j]$
and $u_{i, j+1} \in \Buck[\tau_{j+1}]$
under $\beta^{h_i}_{x, v_i}$. It means that
$\Buck[\ell']$ is empty under $\beta^{h_i}_{x, v_i}$
for all $\ell' \in L_4 - \setof{\tau_{j+1}}$.
Again, as a consequence of Lemma~\ref{lemma:bucket-change},
$\beta^{j}_{x, v_i}$ and $\beta^{h_i}_{x, v_i}$ induce
the same $\Buck[\ell']$
for all $\ell' \in L_4 - \setof{\tau_{j+1}}$,
implying that $\Buck[\ell']$ is empty under $\beta^{j}_{x, v_i}$ for all $\ell' \in L_4 - \setof{\tau_{j+1}}$.
So, we only need to decide whether $\Buck[\tau_{j+1}] \in \Next(j)$ or not.
Also notice that we have $\Acc(\tau_{j+1}) = \Acc(\tau_j) + |\Buck[\tau_{j+1}]|$
under $\beta^j_{x, v_i}$.

If $\Buck[\tau_{j+1}]$ is empty under $\beta^j_{x, v_i}$,
then we have $\Buck[\tau_{j+1}] \not\in \Next(j)$.
In this case,
$\Delta_{j,1}$ was previously set to $\Delta_{j, \ddot{t}}$ (in Stage 1) or $\Delta_{j, \Head}$ ($= \nabla$ or $\nabla'$ or $\nabla''$ in Stage 2) or $\Delta_{j, \Head'}$ ($=\nabla'''$ in Stage 3), and these previous values are all equal to $(\Acc(\sigma_{j,1}) -
\Acc(\tau_{j+1})$ under $\beta^{j}_{x,v_i}) - (\Acc(\sigma_{j,1}) -
\Acc(\tau_{j+1})$ under $\beta^{h_i}_{x, v_i})$. The final value (i.e., the value of (\ref{equation:difference-relation})) of $\Delta_{j,1}$ can be obtained in constant time by
subtracting $(\Acc(\tau_{j+1}) - \Acc(\tau_{j})
$ under $\beta^{h_i}_{x, v_i})$ from the previous value of $\Delta_{j,1}$, because
$\Acc(\tau_j) = \Acc(\tau_{j+1})$
under $\beta^{j}_{x, v_i}$.
In subsequent discussion, we assume that $\Buck[\tau_{j+1}]$ is not empty under $\beta^j_{x, v_i}$.

If $\Next(j)$ is empty at the end of Stage 3,
then $\Buck[\tau_{j+1}+1],\Buck[\tau_{j+1}+2], \ldots, \linebreak\Buck[h_i]$ are all empty under $\beta^{j}_{x, v_i}$.
In this case, there is only one bucket, i.e., $\Buck[\tau_{j+1}]$, in $\Next(j)$.
We set $\Delta_{j,1} = |\Buck[\tau_{j+1}]|$
under $\beta^j_{x, v_i} - (\Acc(\sigma_{j,1}) - \Acc(\tau_j)
$ under $\beta^{h_i}_{x, v_i})$, which can be computed
in constant time. The value is identical with the value of (\ref{equation:difference-relation}), because
$\sigma_{j,1} = \tau_{j+1}$ and $\Acc(\tau_{j+1})  = \Acc(\tau_j) + |\Buck[\tau_{j+1}]|$
under $\beta^j_{x, v_i}$.

Otherwise, let
$\Buck[\sigma_{j, \Head''}]$ be the first
bucket in $\Next(j)$.
In this case,
$\Delta_{j,\Head''}$ was set to $\Delta_{j, \ddot{t}}$ (in Stage 1) or $\Delta_{j, \Head}$ ($= \nabla$ or $\nabla'$ or $\nabla''$ in Stage 2) or $\Delta_{j, \Head'}$ ($=\nabla'''$ in Stage 3), and they are all equal to
$(\Acc(\sigma_{j,\Head''}) -
\Acc(\tau_{j+1})$ under $\beta^{j}_{x,v_i}) - (\Acc(\sigma_{j,\Head''}) -
\Acc(\tau_{j+1})$ under $\beta^{h_i}_{x, v_i})$.
We have $\Buck[\tau_{j+1}] \in \Next(j)$ if
and only if $\Minv(\tau_{j+1}) + \Acc(\tau_{j+1}) \cdot \setuptime \ge
\Minv(\sigma_{j, \Head''}) + \Acc(\sigma_{j,\Head''}) \cdot \setuptime$
under $\beta^j_{x, v_i}$.
The inequality can be written as
$\Minv(\tau_{j+1})$ under $\beta^j_{x, v_i} + \Acc(\tau_{j+1}) \cdot \setuptime$
under $\beta^{h_i}_{x, v_i} \ge \Delta_{j, \Head''} \cdot \setuptime +
\Minv(\sigma_{j, \Head''})$ under $\beta^j_{x, v_i} + \Acc(\sigma_{j,\Head''}) \cdot \setuptime$
under $\beta^{h_i}_{x, v_i}$, which can
be determined in constant time.
Further, if $\Buck[\tau_{j+1}] \not\in \Next(j)$,
then $\Next(j)$ remains unchanged, i.e., $\Head'' = 1$.
The final value of $\Delta_{j,1}$ can be obtained in constant time
by adding $(|\Buck[\tau_{j+1}]|$ under $\beta^{j}_{x,
v_i}- (\Acc(\tau_{j+1}) - \Acc(\tau_{j})
$ under $\beta^{h_i}_{x, v_i}))$ to the previous value of $\Delta_{j,1}$,
because $\Acc(\tau_{j+1})  = \Acc(\tau_j) + |\Buck[\tau_{j+1}]|$
under $\beta^j_{x, v_i}$.

If $\Buck[\tau_{j+1}] \in \Next(j)$, then we have
$\sigma_{j, 1}= \tau_{j+1}$ and $\Head'' = 2$.
We add $\Buck[\tau_{j+1}]$
to the front of $\Next(j)$, 
and set $\Delta_{j,1} = |\Buck[\tau_{j+1}]|$
under $\beta^j_{x, v_i} - (\Acc(\sigma_{j,1}) - \Acc(\tau_j)
$ under $\beta^{h_i}_{x, v_i})$, which is the same as the situation that $\Next(j)$ is empty at the end of Stage 3.
The previous value of $\Delta_{j,2}$ is identical with its final value.
If Stage 4 is omitted (as $\tau_{j+1} = \tau_j$), then
$\Next(j)$ obtained after Stage 3 is what we desire and the previous value of
$\Delta_{j,1}$ is identical with its
final value.

The overall time complexity for constructing $\Next(h_i), \Next(h_i-1), \ldots, \Next(1)$
is computed as $\sum_{1\le j\le h_i}O(\delta_j\log\log h_i) = O((\sum_{1\le j\le h_i}\delta_j)\cdot\log\log h_i)$,
where $\delta_j$ is the number of buckets in the sublist $(\Buck[\sigma_{j+1,1}], \Buck[\sigma_{j+1,2}],\ldots,
\Buck[\sigma_{j+1,\breve{t}}])$ of $\Next(j+1)$
that are not in $\Next(j)$.
Since $\Next(h_i)$ is empty initially and at most two buckets, i.e., $\Buck[\hat{\ell}]$ and $\Buck[\tau_{j+1}]$,
are newly added to each $\Next(j)$ (refer to Stage 2 and Stage 4), we have
$\sum_{1 \le j \le h_i} \delta_j \le 2(h_i-1)$. Therefore, the overall time complexity
is $O(h_i \log\log h_i)$.

\subsection{Proofs of
Fact~\ref{fact:preprocessing-broadcast-time} and Fact~\ref{fact:preprocessing-increasing-order}}
\label{subsection:preprocessing}
%
We first prove Fact~\ref{fact:preprocessing-broadcast-time} below.
Let $s^+$ ($s^-$) denote the scenario of $T$ which has
$w^{s^+}_{u,v} = w^+_{u,v}$ ($w^{s^-}_{u,v} = w^-_{u,v}$) for all
edges $(u,v)\in E(T)$. Then, we have ${\textit b\_time}^{\beta^{j}_{x,
v_i}}(u_{i,k}, \BB{u_{i,k}, v_i}) = {\textit
b\_time}^{s^+}(u_{i,k}, \BB{u_{i,k}, v_i})$ if $k \le j$, and
${\textit b\_time}^{s^-}(u_{i,k}, \BB{u_{i,k}, v_i})$
else (refer to Figure~\ref{figure:gamma_scenario}). The required preprocessing is to determine
${\textit b\_time}^{s^+}(u, \BB{u, v})$ and ${\textit b\_time}^{s^-}(u,
\BB{u, v})$ for all\linebreak edges $(u,v) \in E(T)$. Clearly, after the preprocessing is done,
each ${\textit b\_time}^{\beta^{j}_{x,
v_i}}(u_{i,k}, \BB{u_{i,k}, v_i})$ can be determined in constant time. In the following,
we show that the preprocessing can be completed in $O(n\log n)$ time.
Only the scenario $s^+$ is considered; the scenario
$s^-$ can be treated all the same.

As shown in \cite{Su2016}, $O(n)$ time
is sufficient to find a
broadcast center $\kappa \in {\textit B\_Ctr}^{s^+}$ of $T$ and compute
${\textit b\_time}^{s^+}(u, \BB{u, \kappa})$ for all $u \in V(T) -
\setof{\kappa}$ under $s^+$.
Let $\eta(u)$ denote the
neighbor of $u$ in $\OB{u, \kappa}$.
With the work of \cite{Su2016}, we only need to determine
${\textit
b\_time}^{s^+}(\eta(u), \BB{\eta(u), u})$ for all $u \in
V(T) - \setof{\kappa}$, in order to complete the preprocessing, which depends on the value of $\dis_{\kappa, u}$.

For the vertices $u$ with $\dis_{\kappa, u} = 1$,
it takes $O(h \log h)$ time to determine
an arrangement $(q_1, q_2, \ldots, q_h)$ of $N_T(\kappa)$
such that $w^{s^+}_{q_j, \kappa} + {\textit b\_time}^{s^+}(q_j,
\BB{q_j, \kappa})$ is nonincreasing as $j$ increases from $1$ to
$h$, where $h = |N_T(\kappa)|$.
We need to compute ${\textit b\_time}^{s^+}(\kappa,
\BB{\kappa, q_j})$ for all $1 \le j \le h$.
When $j = 1$, by Lemma~\ref{lemma:optimal-sequence},
$(q_2, q_3, \ldots, q_h)$ is an optimal sequence for $\kappa$ to broadcast a message to
$N_{\BB{\kappa, q_1}}(\kappa)$ under $s^+$. Besides, we have
${\textit b\_time}^{s^+}(\kappa, \BB{\kappa, q_1}) = \max\{(i - 1) \cdot \setuptime + w^{s^+}_{\kappa, q_i} + {\textit b\_time}^{s^+}(q_i, \BB{q_i, \kappa}) \mid
1 < i \le h\}$, computable in additional $O(h)$ time.

When $j > 1$,
$q_1, q_2,\ldots, q_h$ are inserted into
buckets $\Buck[1], \Buck[2], \ldots, \Buck[h]$ as before, and $\Minv(\ell),\Acc(\ell)$ for all $1
\le \ell \le h$ are determined. Both require additional $O(h)$ time.
Suppose $q_j \in \Buck[z_j]$, and let $\lambda(q_j)
=\min\{w^{s^+}_{\kappa, q_{i}} + {\textit b\_time}^{s^+}(q_i,
\BB{q_i, \kappa}) \mid q_{i} \in \Buck[z_j] -\setof{q_j}\}$.
Also, we define
$\pi^-(\ell) = \max\{0, \Minv(t) + \Acc(t) \cdot \setuptime \mid 1 \le t < \ell\}$ and
$\pi^+(\ell') = \max\{0, \Minv(t) + \Acc(t) \cdot \setuptime \mid \ell' < t \le h - 1\}$, where
$1 \le \ell \le h$ and $1 \le \ell' \le h-1$.
Additional $O(h)$ time is sufficient to compute
$\lambda(q_j), \pi^-(\ell)$, and $\pi^+(\ell')$ for
all $1 \le j \le h, \linebreak 1 \le \ell \le h$, and $1\le \ell' \le h-1$.

Notice that after inserting the vertices of $N_{\BB{\kappa, q_j}}(\kappa)$ ($=N_T(\kappa) - \setof{q_j}$)
into buckets $\Buck'[1], \Buck'[2], \ldots, \Buck'[h-1]$,
we have ${\textit b\_time}^{s^+}(\kappa, \BB{\kappa, q_j})
=\max\{\Minv'(\ell) + \Acc'(\ell) \cdot \setuptime \mid 1 \le \ell \le h - 1\}$
(refer to the two paragraphs immediately after Fact~\ref{fact:preprocessing-broadcast-time}), where $\Minv'(\ell)$ and
$\Acc'(\ell)$ are defined on $\Buck'[1], \Buck'[2], \ldots, \Buck'[h-1]$.
Clearly, we have $\Minv'(\ell) = \Minv(\ell)$ if $\ell \not= z_j$;
$\Minv'(\ell) = \lambda(q_j)$ if $\ell = z_j$;
$\Acc'(\ell) = \Acc(\ell)$ if $\ell < z_j$;
$\Acc'(\ell) = \Acc(\ell) - 1$ if $\ell \ge z_j$. Hence,
${\textit b\_time}^{s^+}(\kappa, \BB{\kappa, q_j})$
can be obtained in constant time by finding the maximum of
$\max\{\Minv(\ell) + \Acc(\ell) \cdot \setuptime \mid 1 \le \ell < z_j\}$
($=\pi^-(z_j)$), $\lambda(q_j) + (\Acc(z_j)-1) \cdot \setuptime$, and
$\max\{\Minv(\ell) + (\Acc(\ell) -1) \cdot \setuptime \mid z_j < \ell \le h-1\}$
($=\pi^+(z_j) - 1$).
Consequently, it takes $O(h)$ time to determine
${\textit b\_time}^{s^+}(\kappa, \BB{\kappa, q_j})$
for all $1 < j \le h$.
%


For the vertices $u$ with $\dis_{\kappa, u} =  2$, we
need to compute ${\textit
b\_time}^{s^+}(\eta(u), \BB{\eta(u), u})$, where
$\eta(u) \in \{q_1, q_2, \ldots, q_h\}$.
Let us consider those vertices
$u$ with $\eta(u) = q_j$, where $1 \le j \le h$.
Since ${\textit b\_time}^{s^+}(\kappa,
\BB{\kappa, q_j})$ and
${\textit b\_time}^{s^+}(u, \BB{u, q_j})$ $
(= {\textit b\_time}^{s^+}(u, \BB{u, \kappa}))$ are available,
we can determine
an arrangement $(p_1, p_2, \ldots, p_{h_j})$ of
$N_T(q_j)$ such that
$w^{s^+}_{p_k, q_j} + {\textit b\_time}^{s^+}(p_k,
\BB{p_k, q_j})$ is nonincreasing as $k$ increases from $1$ to
$h_j$ in $O(h_j \log h_j)$ time, where $h_j = |N_T(q_j)|$.
Then it takes additional $O(h_j)$ time to determine
${\textit b\_time}^{s^+}(q_j, \BB{q_j, p_k})$
for all $1 \le k \le h_j$.
It follows that $\sum_{1\le j \le h}O(h_j \log h_j)$ time
is sufficient to determine ${\textit
b\_time}^{s^+}(\eta(u), \BB{\eta(u), u})$ for all $u$
having $\eta(u) \in \{q_1, q_2, \ldots, q_h\}$,
where $h_j = |N_T(q_j)|$.

For the other vertices $u$ (i.e, $\dis_{\kappa, u} > 2$),
${\textit
b\_time}^{s^+}(\eta(u), \BB{\eta(u), u})$
can be obtained, similarly, in
the sequence of $\dis_{\kappa, u} = 3, 4, \ldots$, by the aid of
${\textit b\_time}^{s^+}(\eta(\eta(u)),
\BB{\eta(\eta(u)), \eta(u)})$ and
${\textit b\_time}^{s^+}(u, \BB{u, \eta(u)})$ $
(= {\textit b\_time}^{s^+}(u, \BB{u, \kappa}))$, where
${\textit b\_time}^{s^+}(\eta(\eta(u)),
\BB{\eta(\eta(u)), \eta(u)})$ is \linebreak available (because $\dis_{\kappa, \eta(u)} = \dis_{\kappa, u} -1$).
For each $\eta(u)$, $O(|N_T(\eta(u))| \log |N_T(\eta(u))|)$ time is  spent to determine
an arrangement $(g_1, g_2, \ldots, g_{|N_T(\eta(u))|})$ of
$N_T(\eta(u))$ such that \linebreak
$w^{s^+}_{g_k, \eta(u)} + {\textit b\_time}^{s^+}(g_k,
\BB{g_k, \eta(u)})$ is nonincreasing as $k$ increases from $1$ to
$|N_T(\eta(u))|$. Then,
${\textit b\_time}^{s^+}(\eta(u), \BB{\eta(u), g_k})$
for all $1 \le k \le |N_T(\eta(u))|$ can be obtained in additional $O(|N_T(\eta(u))|)$ time.

%
%
%
%

According to the discussion above,
the time required to determine ${\textit b\_time}^{s^+}(\eta(u), \linebreak \BB{\eta(u),
u})$ for all $u \in V(T) - \setof{\kappa}$
is computed as the summation of $O(|N_T(v)| \log |N_T(v)|)$
for all $v \in \{\eta(u) \mid u \in V(T) - \setof{\kappa}\}$,
which is bounded by
$O(n\log n)$.
This completes the proof of Fact~\ref{fact:preprocessing-broadcast-time}.
%

Next we present a proof of Fact~\ref{fact:preprocessing-increasing-order}.
The required preprocessing is to determine
a vertex ordering $(u_1,
u_2, \ldots, u_h)$ of $N_T(v)$ for each $v \in V(T)$
such that $w^{s^+}_{v, u_{k}}
+ {\textit b\_time}^{s^+}(u_{k}, \BB{u_{k},
v})$ is nonincreasing as $k$ increases from $1$ to $h$.
Recall that in the proof of Fact~\ref{fact:preprocessing-broadcast-time} above, ${\textit b\_time}^{s^+}(u, \BB{u, v})$
for all edges $(u,v) \in E(T)$ are determined.
It follows that the preprocessing can be
done in $O(n \log n)$ time. Moreover, since $w^{\beta^{h_i}_{x,
v_i}}_{v_i, u_{i, k}} + {\textit b\_time}^{\beta^{h_i}_{x,
v_i}}(u_{i,k}, \BB{u_{i,k}, v_i}) = w^{s^+}_{v_i, u_{i,k}} + {\textit
b\_time}^{s^+}(u_{i,k}, \BB{u_{i,k}, v_i})$,
the vertex
ordering, i.e., $(u_{i, 1},
u_{i, 2}, \ldots, u_{i, h_i})$, of $N_{\BB{v_i,x}}(v_i)$ can be
obtained in $O(h_i)$ time by deleting the neighbor of $v_i$ in
$\OB{v_i, x}$ from the vertex ordering of $N_T(v_i)$.

\section{Conclusion}\label{section:conclusion}


In this paper, a broadcasting problem with edge weight uncertainty was treated on
heterogeneous tree networks under the postal model.
All previous
broadcasting problems (refer to \cite{Barnoy00,Barnoy94,Barnoy97,Slater1981,Ri1988,Kh2006,Maja2017,Bortolussi2020,Su2016}) assumed deterministic edge weights.
One challenging problem encountered in this paper is to determine a worst-case scenario from an infinite number of candidates. To find the worst-case scenario, we first restricted it to a finite set of scenarios, and then search the finite set for it by the prune-and-search strategy.




The following results were obtained in this paper. For each vertex $x$ of a tree $T$, there are at most $n-1$ candidates for the worst-case scenario $\ddot{s}(x)$, where $n$ is the number of vertices in $T$. Besides, $\ddot{s}(x)$ (and the maximum regret of $x$) and a minmax-regret broadcast center of $T$ (and its maximum regret) can be found in $O(n \log\log n)$ time and $O(n \log n \log\log n)$ time, respectively. These results may be useful to the researchers who are interested in broadcasting problems.

It is still unknown if $O(n \log n \log\log n)$ time is the best result for finding a minmax-regret broadcast center of $T$. It is not easy to reduce the time complexity to $O(n \log n)$. Instead, it is more likely to derive a lower bound or make a limited improvement on the time complexity of the problem.

\small
\setlength{\baselineskip}{13pt}

\bibliographystyle{abbrv}
\bibliography{mrbc_full}
\end{document}